\pdfoutput=1
\documentclass[12pt]{article}
\usepackage[utf8]{inputenc}
\usepackage[english]{babel}
\usepackage{amsmath}
\usepackage{amssymb}
\usepackage{amsthm}
\usepackage{amsfonts}
\usepackage{mathrsfs}
\usepackage{bm}
\usepackage{graphicx}
\usepackage{booktabs}
\usepackage{multirow}
\usepackage{appendix}
\usepackage{enumerate}
\usepackage{url}
\usepackage{authblk}
\usepackage{xcolor}
\usepackage{xr}
\usepackage{adjustbox}
\usepackage{array}
\usepackage{tabularx}
\usepackage{tikz}
\usepackage{natbib}
\usepackage{algorithm}
\usepackage{algpseudocode}
\usepackage[labelfont=bf]{caption}
\usepackage{threeparttable}
\usepackage{makecell}
\usepackage{rotating}
\usepackage{placeins}
\usepackage{etoc}

\newcommand{\compacttablesetup}{%
    \setlength{\tabcolsep}{2.8pt}%
    \renewcommand{\arraystretch}{0.70}%
    \setlength{\aboverulesep}{0.20ex}%
    \setlength{\belowrulesep}{0.21ex}%
}

\usepackage[colorlinks=true,linkcolor=red,citecolor=blue]{hyperref}
\usetikzlibrary{arrows.meta,positioning,calc}

\graphicspath{{./art/}}

\newcommand{\blind}{1}

\makeatletter
\renewcommand{\arabic}[1]{\expandafter\@arabic\csname c@#1\endcsname}
\makeatother

\definecolor{inputblue}{RGB}{24,98,230}
\definecolor{hiddenteal}{RGB}{62,201,201}
\definecolor{outputblue}{RGB}{116,158,236}
\definecolor{edgeblue}{RGB}{44,118,255}
\definecolor{edgeteal}{RGB}{62,201,201}
\definecolor{edgelight}{RGB}{120,220,230}

\tikzset{
  nnfig/.style={
    font=\small,
    line cap=round,
    line join=round
  },
  input node/.style={
    circle,
    draw=none,
    fill=inputblue,
    minimum size=9mm,
    inner sep=0pt
  },
  hidden node/.style={
    circle,
    draw=none,
    fill=hiddenteal,
    minimum size=9mm,
    inner sep=0pt
  },
  output node/.style={
    circle,
    draw=none,
    fill=outputblue,
    minimum size=9mm,
    inner sep=0pt
  },
  input edge/.style={
    -{Latex[length=1.8mm,width=1.3mm]},
    draw=edgeblue,
    line width=0.55pt
  },
  hidden edge/.style={
    -{Latex[length=1.8mm,width=1.3mm]},
    draw=edgeteal,
    line width=0.55pt
  },
  hidden edge light/.style={
    -{Latex[length=1.8mm,width=1.3mm]},
    draw=edgelight,
    line width=0.55pt
  },
  output edge/.style={
    -{Latex[length=1.8mm,width=1.3mm]},
    draw=outputblue,
    line width=0.55pt
  }
}

\definecolor{figaccent}{HTML}{2E5E78}
\definecolor{figink}{HTML}{252525}

\tikzset{
  journal figure/.style={
    font=\small,
    text=figink,
    line cap=round,
    line join=round
  },
  split node/.style={
    draw=figink!80,
    rounded corners=1.2pt,
    fill=white,
    minimum width=18mm,
    minimum height=6.5mm,
    inner sep=2pt,
    align=center
  },
  treatment plus/.style={
    circle,
    draw=figink!80,
    fill=figaccent!13,
    minimum size=6.5mm,
    inner sep=0pt,
    font=\small\bfseries
  },
  treatment minus/.style={
    circle,
    draw=figink!80,
    fill=figink!7,
    minimum size=6.5mm,
    inner sep=0pt,
    font=\small\bfseries
  },
  tree edge/.style={
    draw=figink!65,
    semithick
  },
  panel title/.style={
    font=\small\bfseries,
    anchor=west
  },
  figure note/.style={
    font=\footnotesize,
    text=figink!68,
    align=center
  },
  neuron/.style={
    circle,
    draw=figink!75,
    fill=white,
    minimum size=5.5mm,
    inner sep=0pt
  },
  input neuron/.style={
    neuron,
    fill=figink!4
  },
  relu neuron/.style={
    neuron,
    draw=figaccent,
    line width=.7pt
  },
  network edge/.style={
    draw=figink!48,
    line width=.55pt
  },
  network edge strong/.style={
    draw=figaccent!90!black,
    line width=.8pt
  },
  network box/.style={
    draw=figink!75,
    rounded corners=1.5pt,
    fill=white,
    minimum height=9mm,
    inner xsep=5pt,
    align=center
  },
  flow arrow/.style={
    -{Latex[length=1.8mm,width=1.3mm]},
    draw=figink!70,
    semithick
  }
}

\tikzset{
  output neuron/.style={
    neuron,
    draw=figink!90,
    line width=.8pt,
    minimum size=10mm,
    font=\scriptsize
  },
  omitted edge/.style={
    network edge,
    densely dashed
  },
  layer label/.style={
    font=\scriptsize,
    align=center,
    text=figink
  }
}

\def\T{{ \mathrm{\scriptscriptstyle T} }}

\newcommand{\calX}{\mathcal{X}}
\newcommand{\calA}{\mathcal{A}}

\newcommand{\calD}{\mathcal{D}}
\newcommand{\calL}{\mathcal{L}}
\newcommand{\calG}{\mathcal{G}}
\newcommand{\calH}{\mathcal{H}}

\newcommand{\mO}{\mathcal{O}}
\newcommand{\bX}{\mathbf{X}}
\newcommand{\bx}{\mathbf{x}}
\newcommand{\lb}{\left\{}
\newcommand{\rb}{\right\}}
\newcommand{\lsb}{\left(}
\newcommand{\rsb}{\right)}
\newcommand{\lbb}{\left[}
\newcommand{\rbb}{\right]}
\newcommand{\bE}{\mathbb{E}}
\newcommand{\mone}{\mathbb{I}}
\newcommand{\mP}{\mathbb{P}}

\newcolumntype{Y}{>{\raggedright\arraybackslash}X}
\newcolumntype{P}[1]{>{\raggedright\arraybackslash}p{#1}}
\newcolumntype{C}[1]{>{\centering\arraybackslash}p{#1}}

\makeatletter
\newcommand*{\addFileDependency}[1]{
\typeout{(#1)}
\@addtofilelist{#1}
\IfFileExists{#1}{}{\typeout{No file #1.}}
}
\makeatother

\newtheorem{assumption}{Assumption}
\newtheorem{theorem}{Theorem}
\newtheorem{corollary}{Corollary}
\newtheorem{example}{Example}
\newtheorem{lemma}{Lemma}
\begin{document}

\def\spacingset#1{\renewcommand{\baselinestretch}{#1}\small\normalsize} \spacingset{1}


\if1\blind
{
  \title{\bf Orthogonal double residual learning for optimal individualized treatment rules}
\author{
  Jiaqi Tong$^{1}$
  and Fan Li$^{1,*}$
  \vspace{0.2cm} 
  \\
  $^{1}$Department of Biostatistics, Yale School of Public Health, New\\ Haven, CT, USA \\
  $^{*}$\emph{email}: fan.f.li@yale.edu
}
  \date{}
  \maketitle
} \fi

\if0\blind
{
  \bigskip
  \bigskip
  \bigskip
  \begin{center}
    {\LARGE\bf Orthogonal double residual learning for optimal individualized treatment rules}
  \end{center}
  \medskip
} \fi

\bigskip
\begin{abstract} 
Individualized treatment rules (ITRs) map baseline characteristics to treatment recommendations, with the optimal ITR maximizing expected reward or policy welfare. Indirect methods may require restrictive modeling assumptions, whereas direct methods can be sensitive to nuisance estimation error and limited overlap. We propose orthogonal double residual learning (ODRL), a two-stage, cross-fitted framework that directly targets the optimal ITR through cost-sensitive classification using the product of treatment and outcome residuals. To our knowledge, ODRL is the first direct method with a universally Neyman orthogonal objective requiring neither restrictive modeling assumptions nor inverse propensity score weighting. Thus, nuisance estimation errors affect regret through a second-order product, and ODRL remains robust under limited overlap. The Fisher consistent objective accommodates general decision rule sieves. We establish nonasymptotic high probability value function regret bounds relative to the Bayes classifier for VC classes, including linear rules and decision trees, and calibrated regret bounds for surrogate relaxations using support vector machines and deep ReLU neural networks. We further show that generic surrogate relaxations need not preserve orthogonality, whereas bounded score hinge learning does. Simulations demonstrate strong performance across complex and linear decision boundaries, limited overlap, and working model misspecification. Applications to the Right Heart Catheterization study and the Oxford Net Zero experiment illustrate interpretable treatment or policy recommendations. The \texttt{odrlITR} R package implements ODRL. 
\end{abstract}

\noindent%
\vfill

\newpage
\spacingset{1.7} 

\section{Introduction}

\subsection{Background and related literature}
In causal inference, average treatment effects may mask substantial variation in treatment effects across individuals. This heterogeneity raises a fundamental decision problem: how should treatment be assigned based on an individual’s characteristics? In precision medicine, an individualized treatment rule (ITR) maps baseline patient characteristics to a recommended treatment; in economics, the analogous object is a treatment assignment policy. An optimal rule maximizes the value function, defined as the expected outcome or reward under deployment of the rule in the target population. The expanding availability of electronic health records, insurance claims, clinical registries, and administrative databases has generated substantial interest in learning such rules from observational data \citep{qian2011performance,zhao2012estimating,zhao2019efficient,athey2021policy}.

Existing methods for learning ITRs fall into two broad categories. The first, referred to as \emph{indirect methods}, includes Q-learning \citep{qian2011performance}, A-learning \citep{schulte2015q}, D-learning \citep{qi2018d}, E-learning \citep{mo2022efficient}, and related variants \citep{zhao2009reinforcement,wallace2015doubly}. These methods estimate either the quality function, defined as the conditional outcome mean, or the blip function, corresponding to the conditional average treatment effect (CATE), and obtain the optimal decision rule by taking the sign of the estimated CATE. The term ``indirect'' reflects that the decision rule is not learned by directly maximizing the value function. Despite their simplicity, indirect methods have several limitations. First, they are primarily model-based, relying on parametric or semiparametric specifications and differing mainly in their estimation strategies. Their consistency therefore requires correct specification of all or part of the nuisance model. Second, even when the quality function is estimated nonparametrically or using data-adaptive machine learning tools, the slow convergence rate associated with estimating a complex quality function propagates to the value function regret bound. This is undesirable when the true decision boundary is substantially simpler than the CATE surface.

Methods in the second category are known as \emph{direct methods} \citep{zhao2012estimating,zhou2017residual,zhao2019efficient,mi2019bagging,kallus2021more,athey2021policy,zhou2023offline}. These methods first construct an estimator \(\widehat V\) of the value function \(V\) and then obtain the optimal ITR by maximizing \(\widehat V(d)\) over the decision rule space \(\calD\). Methods in this category differ primarily in their construction of \(\widehat V\). Examples include outcome weighted learning (OWL) \citep{zhao2012estimating} and residual weighted learning (RWL) \citep{zhou2017residual,mi2019bagging}, both of which use an inverse probability weighting (IPW) estimator of $V$, as well as doubly robust (DR) methods, which use an {augmented inverse probability weighting (AIPW) estimator of $V$} \citep{zhao2019efficient,athey2021policy,zhou2023offline}. The AIPW-based direct methods offer two advantages. First, consistency requires correct specification of either the quality function model or the propensity score model, but not necessarily both. Second, these methods can be implemented within a two-stage debiased machine learning framework \citep{dml,foster2023orthogonal}, with the nuisance functions estimated in the first stage using data-adaptive machine learning tools. The resulting second-stage empirical risk minimization objective is \emph{Neyman orthogonal} (see Example~3.4 of \cite{foster2023orthogonal}), so nuisance estimation errors affect the regret bound only through second-order terms.

In observational studies, however, treatment assignment often reflects clinical judgment, contraindications, institutional practice, and self-selection. These mechanisms can lead to limited overlap, with treatment probabilities close to zero or one for some covariate profiles. In such regions, the data provide little information about one treatment alternative, and the device of IPW may become numerically unstable due to dominating weights of a small number of individuals. In fact, a common limitation of the IPW- and AIPW-based direct methods is that their objectives explicitly involve division by the propensity score, rendering them vulnerable to positivity violation. To mitigate this limitation, \cite{kallus2021more} extends IPW- and AIPW-based direct methods to accommodate optimally retargeted value functions. 
When the retargeting weights depend on an estimated propensity score, however, Neyman orthogonality is not automatic because the dependence of the weights on the nuisance function must also be included in the derivative calculation. Recent covariate balancing ITR methods also avoid the device of IPW by imposing empirical covariate balance constraints in an optimization problem with a linear decision boundary \citep{lee2024effective}.

Table~\ref{tab:literature-comparison} summarizes selective existing methods for estimating optimal ITRs and distinguishes four properties that are often conflated in the ITR literature: whether the treatment rule is learned directly, whether the framework relies on restrictive modeling assumptions, whether the objective used by a direct method is Neyman orthogonal \citep{foster2023orthogonal}, and whether the objective contains an IPW component. To the best of our knowledge, none of the existing methods simultaneously targets the ITR directly, avoids restrictive modeling assumptions, uses a Neyman orthogonal objective---therefore robust to nuisance function estimation errors---and remains robust under lack of overlap.

\subsection{Our contributions}
We propose \emph{orthogonal double residual learning (ODRL)}, 
a two-stage, cross-fitted causal machine learning framework for estimating the optimal ITR. To the best of our knowledge, it is the first direct method based on a Neyman orthogonal objective that requires no restrictive modeling assumptions and does not exploit the numerically unstable IPW component. The core of our approach builds upon a novel objective function. That is, once the nuisance functions entering the new objective are estimated using machine learning tools, the second-stage empirical risk minimization can be implemented using existing computational routines for direct methods. We study a broad class of such procedures that fall into two main approaches: (i) exact optimization over restricted sieves of binary treatment rules \citep{athey2021policy,zhou2023offline}; and (ii) surrogate risk minimization using general machine learning methods to estimate a treatment rule score, including support vector machines (SVMs) \citep{zhao2012estimating,zhou2017residual,zhao2019efficient} and deep neural networks (DNNs) \citep{mi2019bagging}. Our specific contributions are threefold.

First, we propose a double residual loss that can be interpreted as a cost-sensitive binary classification objective with label
`$\operatorname{sign}(\epsilon_A\epsilon_Y)$' and weight
`$|\epsilon_A\epsilon_Y|$', where $\epsilon_A$ and $\epsilon_Y$ denote the treatment and outcome residuals, respectively. The proposed objective therefore requires no explicit division by the propensity score. We show that the double residual loss is closely connected to Robinson's residualization approach for CATE estimation \citep{robinson1988root,nie2021quasi}. In particular, when the CATE is restricted to functions taking values in $\{-1,1\}$, the Robinson-type squared criterion is equivalent for optimization purposes to the proposed double residual loss. Despite this connection, the validity of our method does not require the true CATE to be dichotomous, and direct learning in the spirit of Robinson residualization remains underdeveloped in the ITR literature. We then prove that our loss is both Fisher consistent and universal Neyman orthogonal.

Second, we derive a generic, nonasymptotic, high probability bound on the value function regret relative to the Bayes classifier when the second-stage empirical risk minimization is performed over a general sieve \citep{chen2007large}. We then specialize this result to binary treatment rule sieves whose combinatorial complexity is controlled by their {Vapnik--Chervonenkis (VC)} dimension. This specialization covers several treatment rule classes commonly used in the ITR and policy learning literature, including linear and decision tree rules \citep{manski1975maximum,athey2021policy,zhou2023offline}.

Third, we derive a generic value function regret bound for surrogate risk minimization. Our analysis accommodates a broad family of classification calibrated surrogate losses, including the hinge, squared hinge, and logistic losses, as well as general second-stage score sieves, such as those generated by SVMs and DNNs. Exact optimization of the underlying cost-sensitive classification problem over a rich treatment rule class may be computationally intractable and, for many classes, NP-hard. Surrogate relaxation, on the other hand, provides a practical alternative that reduces the computational burden while retaining some statistical guarantees \citep{bartlett2006convexity,zhao2012estimating,zhou2017residual,zhao2019efficient}. We explain that surrogate relaxation generally does not preserve universal Neyman orthogonality without additional restrictions. As a particularly interesting case, we show that Neyman orthogonality is preserved under the hinge loss with bounded sieves, yielding a computationally tractable procedure whose regret remains only second-order sensitive to nuisance estimation errors.

\begin{table}[ht!]
\caption{A selective comparison between existing ITR methods with the proposed orthogonal double residual learning (ODRL) approach. Entries refer to the cited formulations rather than every possible extension. Here, $\mu(a,\bX):=\bE(Y\mid A=a,\bX)$ denotes the quality function, $m(\bX):=\bE(Y\mid\bX)$ denotes the marginal outcome regression, and $\pi(a,\bX):=\Pr(A=a\mid\bX)$ denotes the propensity score. ``Required'' indicates that the corresponding function must be correctly specified. Paired ``Either'' entries indicate that correct specification of either nuisance function is sufficient. ``Working'' denotes a working model that is not required for Fisher consistency. ``Flexible'' indicates that the nuisance function may be estimated using machine learning, and ``--'' indicates that the corresponding function is not required.  }
\label{tab:literature-comparison}
\centering

\begingroup
\scriptsize
\linespread{1}\selectfont
\compacttablesetup

\setlength{\tabcolsep}{3.2pt}
\renewcommand{\arraystretch}{0.90}

\begin{tabularx}{\linewidth}{
    @{}
    >{\raggedright\arraybackslash}X
    >{\centering\arraybackslash}p{0.072\linewidth}
    *{3}{>{\centering\arraybackslash}p{0.069\linewidth}}
    >{\centering\arraybackslash}p{0.105\linewidth}
    >{\centering\arraybackslash}p{0.070\linewidth}
    >{\centering\arraybackslash}p{0.140\linewidth}
    @{}
}
\toprule
\addlinespace[0.5em]
Method
&
Route
&
\multicolumn{3}{c}{Modeling assumptions}
&
\shortstack{Neyman\\orthogonality}
&
\shortstack{Cross-\\fitting?}
&
\shortstack{Explicit division \\by propensity\\score?}
\\[-0.15em]
\cmidrule(lr){3-5}
& & $\mu$ & $m$ & $\pi$ & & &
\\[0.3em]
\midrule
\addlinespace[0.5em]
Q-learning \citep{qian2011performance}
&
Indirect
&
Required
&
--
&
--
&
NA
&
No
&
No
\\

\addlinespace[1.0em]

A-learning \citep{schulte2015q}
&
Indirect
&
--
&
Either
&
Either
&
NA
&
No
&
No
\\

\addlinespace[1.0em]

D-learning \citep{qi2018d}
&
Indirect
&
--
&
--
&
Required
&
NA
&
No
&
Yes
\\

\addlinespace[1.0em]

E-learning \citep{mo2022efficient}
&
Indirect
&
--
&
Either
&
Either
&
NA
&
No
&
Yes
\\\addlinespace[0.3em]

\midrule
\addlinespace[0.5em]
OWL
\citep{zhao2012estimating}
&
Direct
&
--
&
--
&
Required
&
No
&
No
&
Yes
\\

\addlinespace[1.0em]

RWL \citep{zhou2017residual}
&
Direct
&
Working
&
--
&
Required
&
No
&
No
&
Yes
\\

\addlinespace[1.0em]

EARL
\citep{zhao2019efficient}
&
Direct
&
Either
&
--
&
Either
&
Generally not
&
Optional
&
Yes
\\

\addlinespace[1.0em]

CAIPWL
\citep{athey2021policy,zhou2023offline}
&
Direct
&
Flexible
&
--
&
Flexible
&
Yes
&
Yes
&
Yes
\\

\addlinespace[1.0em]

Optimal retargeting \citep{kallus2021more}
&
Direct
&
Flexible
&
--
&
Flexible
&
No
&
Optional
&
No
\\

\addlinespace[1.0em]

\textbf{ODRL (proposed)}
&
\textbf{Direct}
&
\textbf{--}
&
\textbf{Flexible}
&
\textbf{Flexible}
&
\textbf{Yes}
&
\textbf{Yes}
&
\textbf{No}
\\
\addlinespace[0.3em]
\bottomrule
\end{tabularx}

\endgroup
\end{table}

The remainder of this article is organized as follows. Section~\ref{sec:setup} introduces the notation, identification assumptions, along with nonparametric identification of the value function and optimal ITR. Section~\ref{sec:odrl} develops the double residual loss and establishes its statistical properties. Section~\ref{sec:genera-regret-bound} presents the two-stage cross-fitted ODRL procedure and derives a general value function regret bound. Sections~\ref{sec:exact-sieves} and~\ref{sec:surrogate-sieves} specialize the framework to exact optimization over binary VC sieves and surrogate relaxation over score sieves, respectively. Section~\ref{sec:simulation} reports the simulation experiments. Section~\ref{sec:data-rhc} applies the proposed methods to observational data from the Right Heart Catheterization study. Section~\ref{sec:conclusion} concludes.

\section{Setup and notation}\label{sec:setup}
We consider a randomized or an observational study with a binary treatment $A\in\calA:=\{-1,1\}$.
An analyst observes $n$ independent and identically distributed copies of the data triplet $\mO_i=(Y_i,A_i,\bX_i)^\top$, where $Y_i$ denotes the outcome or reward, with higher values being favorable, and $\bX_i\in\calX$ denotes the $p$-vector of baseline covariates. We assume $\bE(|Y|)<\infty$ throughout. To enable personalized recommendations for precision medicine or policy evaluation, a central quantity of interest is the individualized treatment rule, defined as a decision rule $d:\calX\to\calA$. The optimal ITR maximizes the expected potential reward under that rule, which is given by
\begin{align*}
    d^\ast(\bX)\in\arg\max_{d\in\calD}V(d),
\end{align*}
where $V(d):=\bE\lb Y(d(\bX))\rb$ is the \emph{value function}, $Y(a)$ for $a\in\calA$ denotes the potential reward under assignment $a$, $\calD\subseteq\check{\calD}$ is the specified class of decision rules that encodes domain knowledge, and $\check{\calD}$ contains all measurable decision rules. To identify the value function and the optimal ITR, we invoke the following standard, structural assumptions.
\begin{assumption}\label{asp:iden}
We assume (i) consistency, $Y=Y(A)$ \citep{rubin1974estimating}; (ii) treatment ignorability, $\{Y(1),Y(-1)\}\perp A|\bX$; and (iii) positivity, $\pi(a,\bX)=\Pr(A=a|\bX)\geq\epsilon_0$ for all $a\in\calA$ and $\bX$, for some constant $\epsilon_0\in(0,1/2]$.    
\end{assumption}
Under Assumption \ref{asp:iden}, \cite{qian2011performance} showed that when $\calD=\check{\calD}$, the value function and optimal ITR can be identified as
\begin{align}
    V(d)=&\bE\lbb\frac{\mathbb{I}\lb A=d(\bX)\rb}{\pi(A,\bX)}Y\rbb=\bE\lbb\mu\lb d(\bX),\bX\rb\rbb,\label{eq:iden-value-fun}\\
    d^\ast(\bX)\in&\arg\max_{a\in\calA}\mu(a,\bX),
    \notag
\end{align}
where $\mone(\cdot)$ is the indicator function and $\mu(A,\bX)=\bE(Y\mid A,\bX)$ denotes the conditional outcome mean or the \emph{quality function}, namely the mean reward under treatment $A$ given covariates $\bX$. Alternatively, the optimal ITR is also determined by the sign of the conditional average treatment effect:
\begin{align}
    d^\ast(\bX)=\text{sign}\lb \tau(\bX)\rb,\notag
\end{align}
where $\tau(\bX)=\mu(1,\bX)-\mu(-1,\bX)$ represents the CATE, and $\text{sign}(x)=\mathbb{I}(x\geq0)-\mathbb{I}(x<0)$. 

\section{Orthogonal double residual learning}\label{sec:odrl}
\subsection{Constructing the double residual loss function} 
Hereafter, we retain consistency and ignorability from Assumption \ref{asp:iden} and replace the positivity condition by the following weaker version:
\begin{assumption}[\emph{Weak positivity}]\label{assump:weak-positivity}
We assume that $\epsilon_n\leq\pi(1,\bX)\leq1-\epsilon_n$ for all $\bX$, where $\epsilon_n\in(0,1/2]$ may depend on $n$. Thus, overlap may remain fixed or may deteriorate with $n$.   
\end{assumption}

Our proposed method constructs a direct learning objective whose unrestricted population optimizer agrees with the value optimal rule while avoiding division by the propensity score. In ITR learning, optimal retargeting and covariate balancing likewise modify the learning objective \citep{kallus2021more,lee2024effective}, but these constructions do not by themselves establish the Neyman orthogonality. The construction of residual-based objectives has received considerable attention in CATE estimation \citep{nie2021quasi,chernozhukov2024applied,morzywolek2023weighted}, but little attention in the ITR literature. To proceed, we propose a simple Neyman orthogonal loss function based on the treatment residual $\epsilon_A:=A-\bE(A\mid\bX)=A-e(\bX)$ and the outcome residual $\epsilon_Y:=Y-m(\bX)=Y-\bE(Y\mid\bX)$:
\begin{align}\label{eq:objective-double residual-classification}
    \calL(d):=\bE\lbb |\epsilon_A\epsilon_Y|\mone\lb d(\bX)\neq\text{sign}(\epsilon_A\epsilon_Y) \rb\rbb
    =\bE\lb \frac{|\epsilon_A\epsilon_Y|-\epsilon_A\epsilon_Yd(\bX)}{2}\rb,
\end{align}
where $e(\bX)=2\pi(1,\bX)-1$ holds by definition. Minimization of the loss function in \eqref{eq:objective-double residual-classification} defines a supervised cost-sensitive binary classification problem \citep{bartlett2006convexity} with covariates $\bX$, label $\text{sign}(\epsilon_A\epsilon_Y)$, and weights $|\epsilon_A\epsilon_Y|$. Intuitively, $|\epsilon_A\epsilon_Y|$ can be interpreted as the misclassification cost, and the optimal prediction rule seeks to minimize the average misclassification cost $\calL(d)$. The form of the loss $\calL$ motivates the term \emph{double residual}, because $\calL$ depends on the product of two residuals, $\epsilon_A\epsilon_Y$. 

Our construction of the loss function shares the same spirit as the R-learner in the CATE estimation literature \citep{nie2021quasi}, which builds upon the original idea of \cite{robinson1988root}, in that the loss is Neyman orthogonal, as shown in Theorem \ref{prop:Neyman-orthogonality}, and is constructed from double residuals. Here, we provide some further insights into the proposed double residual loss and Robinson's residualization for the CATE. Robinson's residualization states that the CATE $\tau$ minimizes the following squared loss function:
\begin{align*}
    \mathcal{J}(\tau):=\bE\lbb \{\epsilon_Y-2^{-1}\epsilon_A\tau(\bX)\}^2\rbb.
\end{align*}
Suppose that $\tau(\bX)\in\calA=\{-1,1\}$, 
so that $\tau=\text{sign}(\tau)$. Since $\tau^2=1$, there exists a constant $C$ independent of $\tau$ such that $ \mathcal{J}(\tau)=C-\bE\lb \epsilon_A\epsilon_Y \tau(\bX)\rb$. Therefore, when the CATE function is restricted to binary values, minimizing Robinson's squared loss over $\tau$ is equivalent to minimizing $-2^{-1}\bE\lb \epsilon_A\epsilon_Y \tau(\bX)\rb$, which is exactly the term in the proposed double residual loss that depends on $d$. However, unlike Robinson's squared loss, the proposed double residual loss directly targets the ITR, admits a cost-sensitive classification interpretation, is shown to be Fisher consistent in Theorem~\ref{thm:valid-loss}, and does not require the CATE function to take only binary values.

\subsection{Fisher consistency and universal Neyman orthogonality}
The proposed loss function $\calL$ has two notable properties: Fisher consistency, summarized in Theorem \ref{thm:valid-loss}, and Neyman orthogonality, established in Theorem \ref{prop:Neyman-orthogonality}.
 
\begin{theorem}[\emph{Fisher consistency}]\label{thm:valid-loss}
Under the consistency and ignorability conditions in Assumption \ref{asp:iden} and Assumption \ref{assump:weak-positivity}, the rule $d^\ast(\bX)=\operatorname{sign}\{\tau(\bX)\}$ is a minimizer of $\calL$ over all measurable binary decision rules. Every minimizer agrees with $d^\ast$ almost surely on $\{\tau(\bX)\neq0\}$; if $\Pr\{\tau(\bX)=0\}=0$, then $d^\ast$ is the unique minimizer almost surely.
\end{theorem}

Theorem \ref{thm:valid-loss} states that optimizing either the proposed loss or the value function yields the Bayes classifier $d^\ast(\bX)=\operatorname{sign}\{\tau(\bX)\}$ wherever the CATE is nonzero. More precisely, the minimizer of the double residual criterion is $2\pi(1,\bX)\pi(-1,\bX)\tau(\bX)$. The overlap weight \citep{li2018balancing} is strictly positive under Assumption \ref{assump:weak-positivity}, so it does not change the sign of the CATE. On the zero-CATE set either treatment has the same value, and our convention $\operatorname{sign}(0)=1$ selects one of the minimizers. A useful feature of the criterion is that it avoids division by the propensity score, unlike direct methods based on IPW \citep{zhao2012estimating,zhou2017residual} or AIPW \citep{zhao2019efficient,athey2021policy,zhou2023offline}, and can be particularly attractive in the presence of weak overlap.

To show the Neyman orthogonality, write $\eta=(e,m)^\top$, $Z_\eta=\{A-e(\bX)\}\{Y-m(\bX)\}$, and
$\calL(d;\eta)=\bE[\{|Z_\eta|-Z_\eta d(\bX)\}/2]$, with $\eta_0=(e_0,m_0)^\top$ denoting the true values of $\eta$. The second property guarantees that learning the optimal ITR by empirical minimization of the double residual loss $\calL$ is locally insensitive to perturbations of the nuisance functions around their true values. A loss function satisfying this property is called a \emph{Neyman orthogonal loss function} \citep{foster2023orthogonal}. A prominent existing example is the direct method based on the AIPW estimator of the value function \citep{zhao2019efficient,athey2021policy,zhou2023offline}, as shown in the policy learning example in Section 3.4 of \cite{foster2023orthogonal}. The proposed double residual loss $\calL(d)$ provides another example of a Neyman orthogonal loss function, as formalized in the following theorem.




\begin{theorem}[\emph{Universal Neyman orthogonality}]\label{prop:Neyman-orthogonality}
Assume that the displayed expectations and directional derivatives exist. Then the loss $\calL(d;\eta)$ is universally Neyman orthogonal in the sense of Assumption 5 in \cite{foster2023orthogonal}. Let $\mathcal F$ be an ambient real-valued function space containing the binary decision rules. For every $\bar d\in\mathcal F$, every admissible target direction $v\in\mathcal F$, and every nuisance direction $\dot{\eta}:=(\dot e,\dot m)$,
\begin{align*}
D_\eta D_d\calL(\bar d;\eta_0)\lbb v,\dot{\eta}\rbb=0.
\end{align*}
The directional derivative $D_d$ is taken in the ambient space $\mathcal F$.
\end{theorem}

Theorem \ref{prop:Neyman-orthogonality} shows that the double residual loss satisfies universal Neyman orthogonality, thereby mitigating the propagation of nuisance estimation errors. 
We summarize three important features of the proposed double residual loss that are not simultaneously satisfied by existing direct methods. First, Fisher consistency ensures that the double residual loss correctly targets the optimal ITR. Second, its form improves robustness to the empirical instability of inverse propensity score weighting by avoiding division by the propensity score. Finally, Neyman orthogonality improves robustness to nuisance estimation errors. 

\section{Two-stage meta-learning and a general value function regret bound}\label{sec:genera-regret-bound}
For implementation, we propose the two-stage meta-learning procedure in Algorithm \ref{alg:cross_fitted_double_residual}, with first-stage nuisance predictions obtained using any machine learning method under cross-fitting and second-stage empirical risk minimized over a possibly data-adaptive decision rule sieve $\calD_n$ \citep{chen2007large}. Popular sieve specifications in the ITR and policy learning literature include linear score threshold sieves \citep{athey2021policy,zhou2023offline}, decision tree sieves \citep{athey2021policy,zhou2023offline}, support vector machines \citep{zhao2012estimating,zhou2017residual,zhao2019efficient}, and deep neural networks \citep{mi2019bagging}. In this section, we derive a nonasymptotic, sieve-agnostic regret bound based on a high-level rate for the sieve complexity and defer sieve-specific analyses to Examples \ref{ss:linear-sieves}--\ref{ss:relu-sieves}.

\begin{algorithm}[ht!]
\caption{Two-stage cross-fitted orthogonal double residual learning}
\label{alg:cross_fitted_double_residual}
\begin{algorithmic}[1]
\fontsize{10}{10}\selectfont
\Require Observed data
$\{\mathcal O_i=(Y_i,A_i,\bX_i)\}_{i=1}^n$,
$R$ folds, and decision rule sieve
$\mathcal D_n\subseteq\{d:\mathcal X\rightarrow\{-1,1\}\}$
\Ensure Cross-fitted nuisance predictions
$\{\widehat e_i,\widehat m_i\}_{i=1}^n$
and estimated treatment rule $\widehat d_n$

\State \textbf{Stage 1:} Randomly partition $[n]$ into $R$ disjoint folds
$\mathcal F_r$, $r\in[R]$.

\For{$r=1$ to $R$}
    \State Fit
    $\widehat{\boldsymbol\eta}^{(-r)}
    =(\widehat e^{(-r)},\widehat m^{(-r)})^\top$
    using
    $\{\mathcal O_j:j\in\mathcal F_r^c\}$,
    where $\mathcal F_r^c=[n]\setminus\mathcal F_r$.
    
    \State Set
    $\widehat{\boldsymbol\eta}_i
    =(\widehat e_i,\widehat m_i)^\top
    =\widehat{\boldsymbol\eta}^{(-r)}(\bX_i)$
    for all $i\in\mathcal F_r$.
\EndFor

\State \textbf{Stage 2:} For all $i\in[n]$, compute
$\widehat Z_i=(A_i-\widehat e_i)(Y_i-\widehat m_i)$.

\State Estimate
$\displaystyle
\widehat d_n
\in
\operatorname*{arg\,min}_{d\in\mathcal D_n}
\frac{1}{2n}\sum_{i=1}^n
\left\{|\widehat Z_i|-\widehat Z_i d(\bX_i)\right\}.
$

\State Equivalently, solve a cost-sensitive binary classification problem
with covariates $\bX_i$, labels $\operatorname{sign}(\widehat Z_i)$,
and weights $|\widehat Z_i|$.

\end{algorithmic}
\end{algorithm}

For a treatment rule $d$, define its value regret by
\begin{align*}
 \operatorname{Reg}_{V}(d):=
 V(d^\ast)-V(d) =
 \bE\lbb
 |\tau(\bX)|
 \mone\{d(\bX)\neq d^\ast(\bX)\}
 \rbb.
\end{align*}
Interestingly, the excess double residual risk admits a direct decision theoretic interpretation as the overlap weighted regret \citep{li2018balancing}:
\begin{align*}
 \operatorname{Reg}_{\omega}(d):=
 \calL(d;\eta_0)-\calL(d^\ast;\eta_0)=\bE\lbb
 2\pi(1,\bX)\pi(-1,\bX)|\tau(\bX)|
 \mone\{d(\bX)\neq d^\ast(\bX)\}
 \rbb.     
\end{align*}
Moreover, under Assumption~\ref{assump:weak-positivity}, it follows that $ 2\epsilon_n(1-\epsilon_n)\operatorname{Reg}_{V}(d)
 \leq
 \operatorname{Reg}_{\omega}(d)
 \leq
 2^{-1}\operatorname{Reg}_{V}(d)$.
Consequently, a bound on the excess double residual risk is
immediately a bound on the value regret of the learned ITR. To derive a general value function regret bound, it is convenient to remove the component of the double residual loss that is constant with respect to the target rule. The same device is used for the AIPW-based method in Example 3.4 of \cite{foster2023orthogonal}. Define the target-dependent linear criterion $\mathcal R(d;\eta)=\bE\{\ell^{\mathrm{linear}}(d;\eta)\}$, where $\ell^{\mathrm{linear}}(d;\eta)=-2^{-1}Z_\eta d(\bX)$.
For every pair of rules $d,d'$ and every nuisance value $\eta$, $\calL(d;\eta)-\calL(d';\eta)
 =\mathcal R(d;\eta)-\mathcal R(d';\eta)$ and the same identity holds for empirical risks. Thus, replacing $\ell$ by $\ell^{\mathrm{linear}}$ changes neither empirical minimizers nor excess risks. This replacement also removes the nondifferentiable nuisance-only term $|Z_\eta|/2$. The target derivative of $\mathcal R$ is the same as that of $\calL$, so the universal Neyman orthogonality in Theorem \ref{prop:Neyman-orthogonality} also holds naturally for $\mathcal R$. 

For ease of exposition but without loss of generality, we present the theoretical analysis under a two-fold sample split, although the implementation in Algorithm \ref{alg:cross_fitted_double_residual} uses multi-fold cross-fitting to make more efficient use of the data. When the number of folds is finite and the fold sizes are approximately equal, the bound for the cross-fitted implementation follows by applying analogous arguments fold by fold. We use the first fold $\mathcal F_1$ to estimate the nuisance functions and the independent second fold $\mathcal F_2$ for empirical risk minimization. Let $P f=\bE\{f(\mO)\}$ and let $\mathbb P_2$ be the empirical measure on the second split, whose size is $n_2$, so that $\mathbb P_2f=n_2^{-1}\sum_{i\in\mathcal F_2}f(\mO_i)$. Let $\widehat\eta=(\widehat e,\widehat m)^\top$ be fitted using the first split. 
We allow an approximate empirical minimizer satisfying
\begin{align}\label{eq:sieve-erm}
 \mathbb P_2\ell^{\mathrm{linear}}(\widehat d_n;\widehat\eta)
 \leq\inf_{d\in\calD_n}\mathbb P_2\ell^{\mathrm{linear}}(d;\widehat\eta)+\rho_n,
\end{align}
where $\rho_n\geq0$ records optimization error. Since $\calL(d;\eta)-\calL(d';\eta)
 =\mathcal R(d;\eta)-\mathcal R(d';\eta)$, this is equivalent to approximate minimization of the original double residual loss. Define $a_n:=\inf_{d\in\calD_n}\lb\calL(d;\eta_0)-\calL(d^\ast;\eta_0)\rb$ and $\Delta_n(\calD_n;\widehat\eta):=\sup_{d\in\calD_n}
 \left|(\mathbb P_2-P)\lb Z_{\widehat\eta}(\mO)d(\bX)\rb\right|$.
Here, \(a_n\) denotes the sieve approximation error, namely, the smallest
excess population double residual risk attainable over \(\mathcal{D}_n\)
relative to the optimal ITR \(d^\ast\). In particular, \(a_n=0\) whenever
\(d^\ast \in \mathcal{D}_n\). The term \(\Delta_n\) is the conditional uniform
empirical process deviation of the nuisance plug-in linear criterion over
\(\mathcal{D}_n\). It captures the second-stage estimation error arising from
the complexity of the sieve and the stochastic behavior of the plug-in
residual product $Z_{\widehat\eta}$. Depending on the structure of \(\mathcal{D}_n\),
\(\Delta_n\) can be controlled using suitable maximal inequalities based on VC dimension, bracketing entropy, covering numbers, or Rademacher complexity \citep{vanderVaartWellner2023}. The following Theorem \ref{thm:sieve-regret} is the first main result of this paper. It provides a value function regret bound for a generic sieve.

\begin{theorem}[\emph{Orthogonal sieve oracle inequality}]\label{thm:sieve-regret}
Suppose that the consistency and ignorability conditions in Assumption \ref{asp:iden} and Assumption \ref{assump:weak-positivity} hold. Assume that the displayed risks and directional derivatives are finite and that the nuisance errors are square integrable. Conditional on the nuisance-training split $\mathcal{F}_1$, let $d_n^\circ\in\calD_n$ be any comparator. If a second-stage learner $\widehat d_n\in\calD_n$ satisfies, for some $r_{2,n}\geq0$,
\begin{align}\label{eq:generic-second-stage-rate}
 \mathcal R(\widehat d_n;\widehat\eta)-\mathcal R(d_n^\circ;\widehat\eta)
 \leq r_{2,n},
\end{align}
then it obeys the oracle excess risk bound $\calL(\widehat d_n;\eta_0)-\calL(d_n^\circ;\eta_0)
 \leq r_{2,n}+
 \|\widehat e-e_0\|_2\|\widehat m-m_0\|_2$,
where $\|f\|_2:=(\int |f(\bX)|^2d P_{\bX})^{1/2}$ denotes the usual $L_2(P_{\bX})$ norm. In particular, every $\widehat d_n$ satisfying \eqref{eq:sieve-erm} guarantees \eqref{eq:generic-second-stage-rate}, simultaneously for all $d_n^\circ\in\calD_n$, with
$r_{2,n}=\Delta_n(\calD_n;\widehat\eta)+\rho_n$. Consequently,
\begin{align*}
 \calL(\widehat d_n;\eta_0)-\calL(d^\ast;\eta_0)
 &\leq a_n+\Delta_n(\calD_n;\widehat\eta)+\rho_n+
 \|\widehat e-e_0\|_2\|\widehat m-m_0\|_2,
 \\
 V(d^\ast)-V(\widehat d_n)
 &\leq\frac{a_n+\Delta_n(\calD_n;\widehat\eta)+\rho_n+
 \|\widehat e-e_0\|_2\|\widehat m-m_0\|_2}
 {2\epsilon_n(1-\epsilon_n)}.
\end{align*}
Thus, any conditional high probability bound for $\Delta_n$, together with first-stage error bounds, gives a high probability regret bound by direct substitution.
\end{theorem}

Several implications of Theorem \ref{thm:sieve-regret} are worth emphasizing. First, the oracle excess risk bound admits a modular interpretation. That is, if a second-stage procedure achieves a conditional plug-in excess risk bound \(r_{2,n}\), then its oracle excess risk is bounded by \(r_{2,n}\) plus the nuisance product remainder. Second, the resulting regret bound separates four distinct sources, including the sieve approximation error \(a_n\), the second-stage empirical process error \(\Delta_n(\calD_n;\widehat\eta)\), the numerical optimization or computational error \(\rho_n\), and the first-stage nuisance estimation error \(\|\widehat e-e_0\|_2\|\widehat m-m_0\|_2\). In particular, (i) the optimization or computational error may remain nonnegligible for exact optimization over a rich decision rule sieve \(\calD_n\); and (ii) \(a_n\) measures how well the sieve \(\calD_n\) approximates the decision rule class \(\calD\) containing the optimal ITR. Enriching \(\calD_n\) may reduce \(a_n\) while increasing \(\Delta_n\). 
Third, universal Neyman orthogonality renders the nuisance contribution bilinear rather than first order. If \(\|\widehat e-e_0\|_2=o_{\mP}(r_{e,n})\) and \(\|\widehat m-m_0\|_2=o_{\mP}(r_{m,n})\), then the first-stage contribution is \(o_{\mP}(r_{e,n}r_{m,n})\), rather than the \(o_{\mP}(r_{e,n}+r_{m,n})\) contribution that generally arises without orthogonality. Thus, if both nuisance estimators converge at rates \(o_{\mP}(n^{-1/4})\), their product is \(o_{\mP}(n^{-1/2})\); if one nuisance is estimated at the parametric rate, consistency of the other suffices for the product remainder to be \(o_{\mP}(n^{-1/2})\). Fourth, the weak positivity factor appears only when the overlap weighted classification excess risk is converted into value function regret. In particular, if overlap deteriorates as the sample size increases, as reflected by \(\epsilon_n=o(1)\), the regret may converge more slowly, reflecting the increased statistical difficulty under weaker overlap.

\section{Exact optimization with binary VC sieves}\label{sec:exact-sieves}
The oracle inequality in Theorem \ref{thm:sieve-regret} is stated for a general decision rule sieve. We now specialize the value function regret bound to exact empirical optimization over a binary treatment rule sieve $\calD_n\subseteq\{-1,1\}^{\calX}$, where $\{-1,1\}^{\calX}$ denotes the set of all measurable functions from $\calX$ to $\{-1,1\}$. Exact empirical risk minimization for the cost-sensitive binary classification problem in \eqref{eq:objective-double residual-classification} is combinatorial and typically NP-hard for sufficiently complex sieves. To facilitate computation by controlling sieve complexity, we use the VC dimension as a natural measure of the combinatorial complexity of a binary treatment rule sieve. 

To interpret the VC dimension, consider $r$ covariate values $x_1,\ldots,x_r\in\calX$. The class $\calD_n$ \emph{shatters} these points if it can realize all $2^r$ possible treatment assignments. That is, for each label vector $a=(a_1,\ldots,a_r)\in\{-1,1\}^r$, there exists a rule $d_a\in\calD_n$ such that $d_a(x_j)=a_j$ for all $j\in[r]$. The VC dimension $\operatorname{VC}(\calD_n)$ is the largest $r$ for which such a shattered set exists and is infinite if arbitrarily large finite sets can be shattered. Thus, the VC dimension measures the combinatorial richness of the treatment rule sieve. A larger value permits more distinct treatment assignments on a finite sample and hence greater scope for overfitting. Structured binary classes of this type are also used by \cite{athey2021policy} and \cite{zhou2023offline} for policy learning with doubly robust AIPW risk. 
Let $v_n:=\operatorname{VC}(\calD_n)$, $\bar v_n:=v_n\vee1$, and $q_{n,\mathrm{VC}}(\delta):=
\sqrt{\{\bar v_n+\log(2/\delta)\}/{n_2}}$, where $a\vee b:=\max(a,b)$. 
Theorem \ref{thm:vc-sieve-regret} below provides a value function regret bound for sieves whose complexity is measured by the VC dimension. It can be viewed as a special case of the general bound in Theorem \ref{thm:sieve-regret}, obtained by applying maximal inequalities tailored to VC classes together with McDiarmid's inequality to bound $\Delta_n$.

\begin{theorem}[\emph{VC sieve regret under exact optimization}]\label{thm:vc-sieve-regret}
Suppose that the consistency and ignorability conditions in Assumption \ref{asp:iden} and the weak positivity condition in Assumption \ref{assump:weak-positivity} hold. Let $\calD_n$ be nonempty, have $v_n<\infty$, and be fixed or measurable with respect to $\mathcal F_1$. Let $\widehat d_n\in\calD_n$ satisfy the approximate empirical minimization condition \eqref{eq:sieve-erm}. Assume, for some $M<\infty$, that $|Y|\leq M$ almost surely, $\|\widehat m\|_\infty\leq M$, and $\|\widehat e\|_\infty\leq1$, where $\|f\|_\infty:=\sup_{x}f(x)$. Conditional on the nuisance training split, there is a universal constant $C>0$ such that, for every $\delta\in(0,1)$, with probability at least $1-\delta$,
\begin{align}
 V(d^\ast)-V(\widehat d_n)
 \leq\frac{a_n+C M q_{n,\mathrm{VC}}(\delta)
 +\rho_n+\|\widehat e-e_0\|_2\|\widehat m-m_0\|_2}
 {2\epsilon_n(1-\epsilon_n)}.
 \label{eq:vc-sieve-regret}
\end{align}
Alternatively, a variance-adaptive regret bound is obtained by replacing \(q_{n,\mathrm{VC}}(\delta)\) in \eqref{eq:vc-sieve-regret} with $\{
\bar\omega_n+\|\widehat e-e_0\|_2^2
\}^{1/2}
q_{n,\mathrm{VC}}(\delta)
+q_{n,\mathrm{VC}}^2(\delta)$, where $\bar\omega_n
:=\bE\{\pi(1,\bX)\pi(-1,\bX)\}$.
\end{theorem}

Several remarks are in order. First, the stochastic term in our regret bound has the same leading order dependence, \(\{\operatorname{VC}(\calD_n)/n_2\}^{1/2}\), on the complexity of the treatment rule class and the sample size as the bound in \citet{athey2021policy}. The regret targets, however, differ. In \citet{athey2021policy}, regret is measured relative to the best treatment rule within a prespecified class, whereas our comparator is the unrestricted first-best ITR, \(d^\ast(\bX)=\operatorname{sign}\{\tau(\bX)\}\), equivalently the Bayes classifier. The
restriction to $\calD_n$ therefore enters through the sieve
approximation error $a_n$. If $\calD_n$ contains the Bayes classifier,
then $a_n=0$. Moreover, their guarantee is asymptotic, whereas \eqref{eq:vc-sieve-regret} provides a nonasymptotic high probability bound. Second, Theorem~\ref{thm:vc-sieve-regret} provides two regret bounds: the envelope-based bound in \eqref{eq:vc-sieve-regret} and an additional variance-adaptive bound. The latter can be sharper when \(\bar\omega_n+\|\widehat e-e_0\|_2^2=O_{\mP}\{\epsilon_n(1-\epsilon_n)\}\) and \(q_{n,\mathrm{VC}}^2(\delta)=o\{\epsilon_n(1-\epsilon_n)\}\). Under these conditions, the leading stochastic term improves from order \(M q_{n,\mathrm{VC}}(\delta)/\{\epsilon_n(1-\epsilon_n)\}\) to \(M q_{n,\mathrm{VC}}(\delta)/\{\epsilon_n(1-\epsilon_n)\}^{1/2}\). Third, the bound displays the usual
approximation--complexity tradeoff: enriching the sieve can reduce
$a_n$, but may increase
$v_n=\operatorname{VC}(\calD_n)$. Finally, we conclude with two classical examples of binary VC sieves: linear score threshold sieves and decision tree sieves.

\begin{figure}[ht!]
\centering
\captionsetup{width=1\linewidth}
\begin{tikzpicture}[journal figure,x=1cm,y=1cm]

  \node[split node] (root)  at (3.60,4.10) {$X_1<t_1\ ?$};
  \node[split node] (left)  at (1.85,2.70) {$X_2<t_2\ ?$};
  \node[split node] (right) at (5.35,2.70) {$X_2<t_3\ ?$};

  \node[treatment plus]  (ll) at (0.95,1.05) {$+1$};
  \node[treatment minus] (lr) at (2.75,1.05) {$-1$};
  \node[treatment minus] (rl) at (4.45,1.05) {$-1$};
  \node[treatment plus]  (rr) at (6.25,1.05) {$+1$};

  \draw[tree edge] (root) -- (left)
    node[midway,fill=white,inner sep=1pt,font=\scriptsize] {yes};

  \draw[tree edge] (root) -- (right)
    node[midway,fill=white,inner sep=1pt,font=\scriptsize] {no};

  \draw[tree edge] (left) -- (ll)
    node[midway,fill=white,inner sep=1pt,font=\scriptsize] {yes};

  \draw[tree edge] (left) -- (lr)
    node[midway,fill=white,inner sep=1pt,font=\scriptsize] {no};

  \draw[tree edge] (right) -- (rl)
    node[midway,fill=white,inner sep=1pt,font=\scriptsize] {yes};

  \draw[tree edge] (right) -- (rr)
    node[midway,fill=white,inner sep=1pt,font=\scriptsize] {no};


\end{tikzpicture}
\caption{A depth-$2$ axis-aligned shallow decision tree. Each internal node selects a coordinate and threshold, and each terminal node assigns treatment $+1$ or $-1$.}
\label{fig:tree-sieve}
\end{figure}
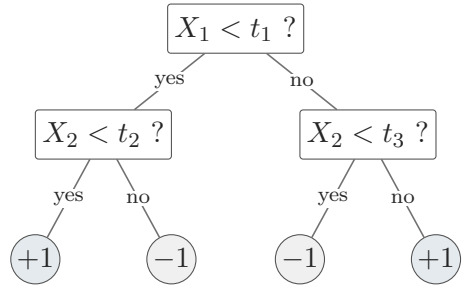

\begin{example}[\emph{Linear score threshold sieves}]\label{ss:linear-sieves}
 Let $b_{K_n}:\calX\to\mathbb R^{K_n}$ be a set of basis functions and consider
\begin{align}\label{eq:linear-series-sieve}
 \calD_n^{\mathrm{lin}}
 =\left\{\bX\mapsto\operatorname{sign}
 \{\beta_0+\boldsymbol{\beta}^\top b_{K_n}(\bX)\}:
 (\beta_0,\boldsymbol{\beta}^\top)^\top\in\mathbb R^{K_n+1}\right\}.
\end{align}
The class has VC dimension at most $K_n+1$ \citep{vanderVaartWellner2023}. When the raw covariates are used as the dictionary, $b_{K_n}(\bX)=\bX$ and $K_n=p$, so \eqref{eq:linear-series-sieve} reduces to the class of linear rules considered by \cite{athey2021policy}. Estimation of the coefficient vector $(\beta_0,\boldsymbol{\beta}^\top)^\top$ dates back to 
\cite{manski1975maximum}.
More generally, the basis functions $b_{K_n}$ can be constructed using wavelets, splines, polynomials, and trigonometric functions, with tensor products used for $p\geq2$ \citep{chen2007large}. Theorem \ref{thm:vc-sieve-regret} then gives a stochastic term of order $M\sqrt{\{K_n+1+\log(2/\delta)\}/{n_2}}$. 
\end{example}

\begin{example}[\emph{Decision tree sieves}]\label{ss:tree-sieves}
A depth-$0$ tree is a constant rule in $\{-1,1\}$. Recursively, a
depth-$L$ axis-aligned tree chooses a coordinate $j\in[p]$ and a
threshold $t\in\mathbb R$, sends observations with $X_j< t$ and
$X_j\geq t$ to two trees of depth at most $L-1$, and assigns a treatment
label at every terminal node. Figure~\ref{fig:tree-sieve} illustrates a shallow, axis-aligned decision tree with $L=2$ and $p=2$. Let
$\calD_{L,p}^{\mathrm{tree}}$ denote this class. 
We derive the following bound on the VC dimension of decision tree sieves in Section \ref{ss:tree-VC-dimension} of the Supplementary Material using the Hamming entropy bound in Lemma~4 of \cite{zhou2023offline}:
\begin{align}\label{eq:tree-vc-bound}
 \operatorname{VC}(\calD_{L,p}^{\mathrm{tree}})
 \leq C_{\mathrm{tree}}2^L\log\{2p(L+2)\}.
\end{align}
Notably, the VC dimension grows exponentially with depth $L$ but only logarithmically with the ambient number of splitting variables $p$. Comparable results are also given in Section~2.2 (p.~143) of \cite{athey2021policy}. Consequently, for $\calD_n=\calD_{L_n,p}^{\mathrm{tree}}$, combining Theorem \ref{thm:vc-sieve-regret} and \eqref{eq:tree-vc-bound} yields, with conditional probability at least $1-\delta$,
\begin{equation}
\scalebox{0.85}{%
  \(\displaystyle
  V(d^\ast)-V(\widehat d_n^{\mathrm{tree}})
  \leq
  \frac{
    a_{n,\mathrm{tree}}
    + C M\sqrt{
        \displaystyle\frac{
          C_{\mathrm{tree}}2^{L_n}\log\{2p(L_n+2)\}
          + \log\left(\frac{2}{\delta}\right)
        }{n_2}
      }
    + \rho_n
    + \|\widehat e-e_0\|_2\|\widehat m-m_0\|_2
  }{
    2\epsilon_n(1-\epsilon_n)
  }
  \)
},
\label{eq:tree-sieve-regret}
\end{equation}
where $a_{n,\mathrm{tree}}$ denotes the approximation error $a_n$ evaluated over the tree class. 
The term $\rho_n$ captures numerical optimization error, which typically persists under greedy or other approximate tree fitting procedures but vanishes ($\rho_n=0$) under global search over finite depth trees.
\end{example}

\section{Surrogate relaxation over score sieves}\label{sec:surrogate-sieves}
\subsection{Value function regret bound under a generic surrogate relaxation}
As noted in Section \ref{sec:exact-sieves}, exact empirical optimization of the cost-sensitive binary classification risk can be computationally demanding. A common alternative in the classification and ITR literature \citep{zhao2012estimating,zhou2017residual,zhao2019efficient} is \emph{surrogate relaxation}. For a measurable score $g:\calX\to\mathbb R$, let $d_g(\bX)=\operatorname{sign}\{g(\bX)\}$ and $S_\eta=\operatorname{sign}(Z_\eta)$. Then the double residual risk can be written exactly as $\calL(d_g;\eta)
 =\bE\lbb |Z_\eta|\mone\{d_g(\bX)\neq S_\eta\}\rbb$.
Let $\phi:\mathbb R\to[0,\infty)$ be a convex, classification-calibrated margin loss. Its surrogate double residual loss and population risk are  $\ell_\phi(g;\eta)
 =|Z_\eta|\phi\{S_\eta g(\bX)\}$ and $\calL_\phi(g;\eta)=\bE\{\ell_\phi(g;\eta)\}$, respectively. The second-stage empirical risk minimization is performed over a score sieve $g\in\calG_n$ rather than a binary decision rule sieve $d\in\calD_n$. The score sieve is not restricted to binary-valued functions; common specifications include a \emph{reproducing kernel Hilbert space (RKHS)} ball for SVMs \citep{zhao2012estimating,zhou2017residual,zhao2019efficient} and a DNN class \citep{mi2019bagging}. Write $c_0(\bX)=\bE(|Z_{\eta_0}|\mid\bX)$, $\bar c_0=\bE(|Z_{\eta_0}|)$, and $\calL_\phi^\ast=\inf_{g:\calX\to\mathbb R}\calL_\phi(g;\eta_0)$. Let $\psi_\phi$ be the $\psi$-transform in Definition~2 of \citet{bartlett2006convexity}, restricted to $[0,1]$. 
Representative losses and their $\psi$-transforms are listed in Table \ref{tab:surrogate-calibration}.

\begin{table}[ht!]
\centering
\caption{Convex margin losses and their $\psi$-transforms. The convention $0\log0=0$ is used and $x_+=x\vee0$. Universal Neyman orthogonality is guaranteed by the bounded-score hinge construction in Section \ref{ss:bounded-hinge}; it is not automatic for the other losses.}
\label{tab:surrogate-calibration}

\begingroup
\footnotesize
\linespread{1}\selectfont
\compacttablesetup

\setlength{\tabcolsep}{25.2pt}
\renewcommand{\arraystretch}{0.99}
\setlength{\aboverulesep}{0.8ex}
\setlength{\belowrulesep}{0.8ex}

\begin{tabularx}{\linewidth}{@{}
>{\raggedright\arraybackslash}p{0.16\linewidth}
>{\raggedright\arraybackslash}p{0.23\linewidth}
>{\raggedright\arraybackslash}X
@{}}
\toprule
\textbf{Loss} & \textbf{$\phi(u)$} & \textbf{$\psi_\phi(t)$, $0\leq t\leq1$}\\
\midrule
Hinge & $(1-u)_+$ & $t$\\
Squared hinge & $(1-u)_+^2$ & $t^2$\\
Exponential & $\exp(-u)$ & $1-\sqrt{1-t^2}$\\
Logistic & $\log\{1+\exp(-u)\}$ & $\frac{1+t}{2}\log(1+t)+\frac{1-t}{2}\log(1-t)$\\
\bottomrule
\end{tabularx}

\endgroup
\end{table}

Let $\calG_n$ be nonempty and measurable with respect to $\mathcal F_1$. We allow a $\rho_n$-approximate second-stage empirical minimizer $\widehat g_n\in\calG_n$ satisfying
\begin{align}\label{eq:surrogate-erm}
 \mathbb P_2\ell_\phi(\widehat g_n;\widehat\eta)
 \leq
 \inf_{g\in\calG_n}\mathbb P_2\ell_\phi(g;\widehat\eta)+\rho_n,
 \qquad \rho_n\geq0.
\end{align}
Define the surrogate approximation, empirical process, and nuisance sensitivity terms by $a_{n,\phi}:=\inf_{g\in\calG_n}\{\calL_\phi(g;\eta_0)-\calL_\phi^\ast\}$, $ \Delta_{n,\phi}(\calG_n;\widehat\eta):=\sup_{g\in\calG_n}
 |(\mathbb P_2-P)\ell_\phi(g;\widehat\eta)|$, and $\Gamma_{n,\phi}(\calG_n;\widehat\eta):=\sup_{g,g'\in\calG_n}
 |\{\calL_\phi(g;\eta_0)-\calL_\phi(g';\eta_0)\}
 -\{\calL_\phi(g;\widehat\eta)-\calL_\phi(g';\widehat\eta)\}|$.
The quantities $a_{n,\phi}$ and $\Delta_{n,\phi}$ are the surrogate analogues of $a_n$ and $\Delta_n$ in Section \ref{sec:genera-regret-bound}. The term $\Gamma_{n,\phi}$ isolates the change in pairwise surrogate-risk differences caused by nuisance estimation. This term is needed because a generic surrogate relaxation need not inherit the universal Neyman orthogonality.

\begin{theorem}[\emph{Calibrated surrogate sieve oracle inequality}]\label{thm:surrogate-oracle}
Suppose that the consistency and ignorability conditions in Assumption \ref{asp:iden} and the weak positivity condition in Assumption \ref{assump:weak-positivity} hold. Let $\phi$ be a convex, classification-calibrated map from $\mathbb R$ into $[0,\infty)$. Assume that the displayed risks are finite. If $\bar c_0>0$, then every measurable score $g$ satisfies
\begin{align}
 \psi_\phi\left\{
 {\operatorname{Reg}_{\omega}(d_g)}/{\bar c_0}
 \right\}
 \leq
 \{\calL_\phi(g;\eta_0)-\calL_\phi^\ast\}/{\bar c_0}.\notag
\end{align}
Consequently, any measurable $g_\phi^\ast$ attaining $\calL_\phi^\ast$ induces a Fisher consistent rule: $d_{g_\phi^\ast}=d^\ast$ almost surely on $\{\tau(\bX)\neq0\}$. Moreover, every $\widehat g_n$ satisfying \eqref{eq:surrogate-erm} obeys
\begin{align}
 V(d^\ast)-V(d_{\widehat g_n})
 \leq\frac{\bar c_0}{2\epsilon_n(1-\epsilon_n)}
 \psi_\phi^{-1}\left(
 \frac{a_{n,\phi}+2\Delta_{n,\phi}(\calG_n;\widehat\eta)
 +\rho_n+\Gamma_{n,\phi}(\calG_n;\widehat\eta)}{\bar c_0}
 \right),\notag
\end{align}
where, for $0\leq s \leq \infty$, $\psi_\phi^{-1}(s)=\sup\{t\in[0,1]:\psi_\phi(t)\leq s\}$. If $\bar c_0=0$, then $\tau(\bX)=0$ almost surely under weak positivity, and every rule has zero value regret.
\end{theorem}

The weighted calibration inequality, the corresponding Fisher consistency results, and the value function regret bound in Theorem \ref{thm:surrogate-oracle} parallel Proposition~3.1(a)--(b) and Theorem~3.1 of \citet{zhao2019efficient}. Unlike Theorem~3.1, however, our result requires only $0<\bar c_0=\bE\{c_0(\bX)\}=\bE(|Z_{\eta_0}|)<\infty$, rather than uniform upper and positive lower bounds on $c_0(\bX)$. 
The nuisance-sensitivity term admits the following direct bound. Suppose that $|g|\leq B_n$ for every $g\in\calG_n$, let
$K_{\phi,n}=\sup_{|u|\leq B_n}\phi(u)<\infty$, and assume that $Y-m_0(\bX)$, $\widehat e-e_0$, and $\widehat m-m_0$ are square integrable. In Section \ref{ss:proof-nuisance-surrogate-bound} of the Supplementary Material, we prove that
\begin{align}\label{eq:generic-surrogate-nuisance-bound}
 \Gamma_{n,\phi}(\calG_n;\widehat\eta)\leq K_{\phi,n}\Bigl\{
 \|Y-m_0(\bX)\|_{2}\|\widehat e-e_0\|_2
 +\|\widehat m-m_0\|_2+\|\widehat e-e_0\|_2\|\widehat m-m_0\|_2
 \Bigr\}.
\end{align}
Thus, absent additional loss structure, this generic bound contains terms that are linear in the first-stage errors. A sharper product-form remainder requires additional cancellation. For the hinge loss on score sieves bounded by one, the surrogate criterion becomes linear in the score, exact universal Neyman orthogonality is recovered, and nuisance errors enter only at second order, as shown in Section \ref{ss:bounded-hinge}.

\subsection{Hinge loss with bounded score sieves and universal Neyman orthogonality}\label{ss:bounded-hinge}
Finally, the ordinary hinge loss has a special property when the score sieve is bounded by one. Let $\calG_n\subseteq[-1,1]^{\calX}$, where $[-1,1]^{\calX}$ denotes the collection of all measurable functions from $\calX$ to $[-1,1]$. For $g\in\calG_n$, the surrogate relaxation under the hinge loss is  $\calL_H(g;\eta)=\bE\{\ell_H(g;\eta)\}$, where $x_+=x\vee0$ and $\ell_H(g;\eta):=|Z_\eta|\{1-\operatorname{sign}(Z_\eta)g(\bX)\}_+=|Z_\eta|-Z_\eta g(\bX)$.
Thus, after removing the nonsmooth nuisance-only term $|Z_\eta|$, the target-dependent part of the hinge criterion is linear: $\calL_H^{\mathrm{linear}}(g;\eta)=\bE\{\ell_H^{\mathrm{linear}}(g;\eta)\}$, where
\begin{align}\label{eq:hinge-affine-criterion}
 \ell_H^{\mathrm{linear}}(g;\eta)=-Z_\eta g(\bX).
\end{align}
For every pair $g,g'\in\calG_n$ and every nuisance value $\eta$, $\calL_H(g;\eta)-\calL_H(g';\eta)
 =\calL_H^{\mathrm{linear}}(g;\eta)
 -\calL_H^{\mathrm{linear}}(g';\eta)$,
and the same identity holds for empirical risks. Hence the hinge and linear criteria have the same empirical minimizers and the same excess risks on $\calG_n$. 
Define the hinge approximation error and the target-dependent empirical process fluctuation by $a_n^H=\inf_{g\in\calG_n}
 \{\calL_H(g;\eta_0)-\calL^\ast_{H}(\eta_0)\}$ and $\Delta_{n,H}^{\mathrm{linear}}(\calG_n;\widehat\eta)=\sup_{g\in\calG_n}
 \left|(\mathbb P_2-P)
 \{Z_{\widehat\eta}g(\bX)\}\right|$, where $\calL^\ast_{H}(\eta_0)=\calL^\ast_{\phi}(\eta_0)$ with $\phi$ being the hinge loss. The following theorem shows that the core statistical properties established in Theorems \ref{prop:Neyman-orthogonality} and \ref{thm:sieve-regret} are preserved under hinge-loss surrogate relaxation over bounded sieves.

\begin{theorem}
\label{thm:bounded-hinge-sieve}
Suppose that the consistency and ignorability conditions in Assumption \ref{asp:iden} and the weak positivity condition in Assumption \ref{assump:weak-positivity} hold. Assume that the displayed risks and directional derivatives are finite and that $\widehat e-e_0$ and $\widehat m-m_0$ are square integrable. Let $\widehat g_n$ be a $\rho_n$-approximate empirical hinge minimizer over $\calG_n\subseteq[-1,1]^{\calX}$. The linear criterion in \eqref{eq:hinge-affine-criterion} is universally Neyman orthogonal such that for every ambient target direction $v$ and nuisance direction $(\dot e,\dot m)$, $D_\eta D_g\calL_H^{\mathrm{linear}}(\bar g;\eta_0)
 \lbb v,(\dot e,\dot m)\rbb=0$.
 Moreover, 
\begin{align}
\label{eq:bounded-hinge-regret}
 V(d^\ast)-V\{\operatorname{sign}(\widehat g_n)\}
 \leq\frac{a_n^H+2\Delta_{n,H}^{\mathrm{linear}}(\calG_n;\widehat\eta)+\rho_n
 +2\|\widehat e-e_0\|_2\|\widehat m-m_0\|_2}
 {2\epsilon_n(1-\epsilon_n)}.\notag
\end{align}
\end{theorem}

Several observations follow. First, Theorem \ref{thm:bounded-hinge-sieve} is agnostic to the score architecture and can therefore be specialized to any bounded sieves. Second, universal orthogonality is a property of the linear-equivalent bounded-hinge criterion under the range restriction $|g|\leq1$, rather than of an arbitrary surrogate relaxation. Third, if the score class $\mathcal{G}_n$ is convex, the resulting empirical optimization problem is convex and can therefore be solved efficiently. By contrast, an unbounded hinge score or a different surrogate must generally be handled by Theorem \ref{thm:surrogate-oracle}. We next specialize the result to the SVMs and DNN classes in Examples \ref{ss:svm-sieves} and \ref{ss:relu-sieves}.

\begin{example}[\emph{Support vector machines}]\label{ss:svm-sieves}
To accommodate nonlinear scores while controlling their complexity, let $K_n:\calX\times\calX\to\mathbb R$ be a symmetric positive semidefinite kernel, and let $\calH_{K_n}$ denote its associated RKHS with inner product $\langle \cdot,\cdot\rangle_{\calH_{K_n}}$ and norm $\|\cdot\|_{\calH_{K_n}}$. For every $x\in\calX$, the kernel section $K_n(x,\cdot)$ belongs to $\calH_{K_n}$ and satisfies the reproducing identity $g(x)=\langle g,K_n(x,\cdot)\rangle_{\calH_{K_n}}$, for any $g\in\calH_{K_n}$. Thus, the kernel determines the sieve representation, whereas the RKHS norm controls the complexity of the score. For example, $K_n(x,x')=b_{K_n}(x)^\top b_{K_n}(x')$ recovers a linear score in a chosen basis dictionary, while a Gaussian kernel generates flexible nonlinear scores. 
For $\lambda_n>0$ the second-stage SVM solves the penalized problem $ \widehat g_{\lambda_n}
 \in\arg\min_{g\in\calH_{K_n}}
 \{\mathbb P_2\ell_\phi(g;\widehat\eta)
 +\lambda_n\|g\|_{\calH_{K_n}}^2\}$. Although $\calH_{K_n}$ may be infinite-dimensional, the representer theorem implies that a minimizer can be taken to have the finite expansion $ \widehat g_{\lambda_n}(\cdot)
 =\sum_{i\in\mathcal F_2}\alpha_iK_n(\bX_i,\cdot)$. Assume that $\sup_{x\in\calX}K_n(x,x)\leq\kappa_n^2$. By the reproducing identity and Cauchy--Schwarz inequality, $ |g(x)|\leq \|g\|_{\calH_{K_n}}K_n(x,x)^{1/2}
 \leq\kappa_n\|g\|_{\calH_{K_n}}$. As proved in Section~\ref{ss:proof-penalized-equivalence} of the Supplementary Material, for a radius $R_n$ corresponding to $\lambda_n$, the penalized problem has the equivalent constrained form over the RKHS ball $ \calG_n^{\mathrm{svm}}(R_n)
 =\{g\in\calH_{K_n}:\|g\|_{\calH_{K_n}}\leq R_n\}$, and every $g$ in this sieve satisfies $\|g\|_\infty\leq\kappa_nR_n$; setting $\lambda_n=0$ corresponds formally to optimization over the entire RKHS with $R_n=\infty$. Supplementary Material Lemma \ref{lem:svm-surrogate-complexity} gives the corresponding RKHS Rademacher complexity bound, which may be substituted for $\Delta_{n,\phi}(\calG_n;\widehat\eta)$ in Theorem \ref{thm:surrogate-oracle}. For the hinge loss, universal Neyman orthogonality is retained by imposing the bounded-score restriction $ \calG_n^{\mathrm{bsvm}}(R_n)
 =\{g\in\calH_{K_n}:\|g\|_{\calH_{K_n}}\leq R_n,
 \ \|g\|_\infty\leq1\}$. The resulting optimization over $\calG_n^{\mathrm{bsvm}}(R_n)$ is a semi-infinite convex program for which efficient software may not be readily available. To mitigate this difficulty, we compute $\widehat f_n$ as a $\rho_n$-approximate empirical hinge minimizer over $\calG_n^{\mathrm{svm}}(R_n)$ and apply hard tanh clipping to obtain, $\widehat{g}_n':=T_1(\widehat f_n)$, where $T_1(u)=-1+(u+1)_+-(u-1)_+$. We show that $\mathbb P_2\ell_H(\widehat g_n';\widehat\eta)
\leq
\inf_{g\in\calG_n^{\mathrm{bsvm}}(R_n)}
\mathbb P_2\ell_H(g;\widehat\eta)+\rho_n$. Thus, by the contraction inequality, this construction yields the same value regret rate as if $\widehat{g}_n'$ is a $\rho_n$-approximate minimizer over $\calG_n^{\mathrm{bsvm}}(R_n)$ even though $\widehat g_n'$ may not belong to $\calG_n^{\mathrm{bsvm}}(R_n)$. Write $a_{n,\mathrm{svm}}^H=\inf_{g\in\calG_n^{\mathrm{bsvm}}(R_n)}
 \{\calL_H(g;\eta_0)-\calL^\ast_H(\eta_0)\}$ and $ q_{n,\mathrm{svm}}(\delta)
 ={\kappa_nR_n}/{\sqrt{n_2}}
 +\sqrt{{\log(2/\delta)}/{n_2}}$. 
The corollary below specializes Theorem \ref{thm:bounded-hinge-sieve} to the SVM score sieve $\calG_n^{\mathrm{bsvm}}(R_n)$.

\begin{corollary}\label{cor:bounded-svm}
Assume the conditions of Theorem \ref{thm:bounded-hinge-sieve}, $|Y|\leq M$, $\|\widehat m\|_\infty\leq M$, and $\|\widehat e\|_\infty\leq1$. Conditional on the nuisance-training split, with probability at least $1-\delta$, it follows that $V(d^\ast)-V\{\operatorname{sign}(\widehat{g}_n')\}
 \leq 2^{-1}\epsilon_n^{-1}(1-\epsilon_n)^{-1}\{a_{n,\mathrm{svm}}^H+16 Mq_{n,\mathrm{svm}}(\delta)
 +\rho_n+2\|\widehat e-e_0\|_2\|\widehat m-m_0\|_2\}$. 
\end{corollary}

\end{example}

\begin{figure}[ht!]
\centering
\captionsetup{width=.97\linewidth}
\begin{adjustbox}{max width=.98\linewidth}
\begin{tikzpicture}[journal figure,x=1cm,y=1cm]

  \node[layer label,text width=1.85cm] at (0.85,5.05)
    {input layer\\$p_{0,n}=p$};
  \node[layer label,text width=1.85cm] at (3.05,5.05)
    {hidden layer $1$\\$p_{1,n}$ units};
  \node[layer label,text width=1.85cm] at (5.15,5.05)
    {hidden layer $2$\\$p_{2,n}$ units};
  \node[layer label] at (6.85,5.05) {$\cdots$};
  \node[layer label,text width=1.90cm] at (8.55,5.05)
    {hidden layer $L_n$\\$p_{L_n,n}$ units};
  \node[layer label,text width=2.05cm] at (10.65,5.05)
    {output layer\\$p_{L_n+1,n}=1$};
  \node[layer label,text width=2.25cm] at (13.40,5.05)
    {treatment rule};

  \node[input neuron] (x1) at (0.85,4.05) {};
  \node[input neuron] (x2) at (0.85,2.90) {};
  \node at (0.85,1.80) {$\vdots$};
  \node[input neuron] (xp) at (0.85,0.60) {};
  \node[anchor=east] at ($(x1.west)+(-.12,0)$) {$X_1$};
  \node[anchor=east] at ($(x2.west)+(-.12,0)$) {$X_2$};
  \node[anchor=east] at ($(xp.west)+(-.12,0)$) {$X_p$};

  \node[relu neuron] (h11) at (3.05,4.05) {$\sigma$};
  \node[relu neuron] (h12) at (3.05,2.90) {$\sigma$};
  \node at (3.05,1.80) {$\vdots$};
  \node[relu neuron] (h1p) at (3.05,0.60) {$\sigma$};

  \node[relu neuron] (h21) at (5.15,4.05) {$\sigma$};
  \node[relu neuron] (h22) at (5.15,2.90) {$\sigma$};
  \node at (5.15,1.80) {$\vdots$};
  \node[relu neuron] (h2p) at (5.15,0.60) {$\sigma$};

  \node[relu neuron] (hL1) at (8.55,4.05) {$\sigma$};
  \node[relu neuron] (hL2) at (8.55,2.90) {$\sigma$};
  \node at (8.55,1.80) {$\vdots$};
  \node[relu neuron] (hLp) at (8.55,0.60) {$\sigma$};

  \draw[network edge] (x1) -- (h11);
  \draw[network edge] (x1) -- (h12);
  \draw[network edge] (x2) -- (h11);
  \draw[network edge] (x2) -- (h1p);
  \draw[network edge] (xp) -- (h12);
  \draw[network edge] (xp) -- (h1p);

  \draw[network edge] (h11) -- (h21);
  \draw[network edge] (h11) -- (h22);
  \draw[network edge] (h12) -- (h21);
  \draw[network edge] (h12) -- (h2p);
  \draw[network edge] (h1p) -- (h22);
  \draw[network edge] (h1p) -- (h2p);

  \draw[omitted edge] (h21) -- (hL1);
  \draw[omitted edge] (h22) -- (hL2);
  \draw[omitted edge] (h2p) -- (hLp);
  \node[fill=white,inner xsep=2pt,font=\large] at (6.85,4.05) {$\cdots$};
  \node[fill=white,inner xsep=2pt,font=\large] at (6.85,2.90) {$\cdots$};
  \node[fill=white,inner xsep=2pt,font=\large] at (6.85,0.60) {$\cdots$};

  \node[output neuron] (g) at (10.65,2.32) {$f(\bX)$};
  \draw[network edge] (hL1) -- (g);
  \draw[network edge] (hL2) -- (g);
  \draw[network edge] (hLp) -- (g);
  \node[figure note] at (10.65,3.80)
    {$\widehat g_n(\bX)=T_1\{f(\bX)\}$};
  \node[figure note] at (10.65,1.40)
    {$\lvert\widehat g_n(\bX)\rvert\leq1$};

  \node[network box,text width=2.75cm] (rule) at (13.40,2.32)
    {$\widehat d_n(\bX)
      =\operatorname{sign}\{\widehat g_n(\bX)\}$};
  \draw[flow arrow] (g) -- (rule);


\end{tikzpicture}
\end{adjustbox}
\caption{A schematic representation of the bounded sparse ReLU deep neural networks in Example~\ref{ss:relu-sieves}. The layer-width vector is $\mathbf p_n=(p_{0,n},\ldots,p_{L_n+1,n})$, the network output is constrained to lie in $[-1,1]$, and its sign defines the individualized treatment rule.}
\label{fig:sparse-relu-sieve}
\end{figure}
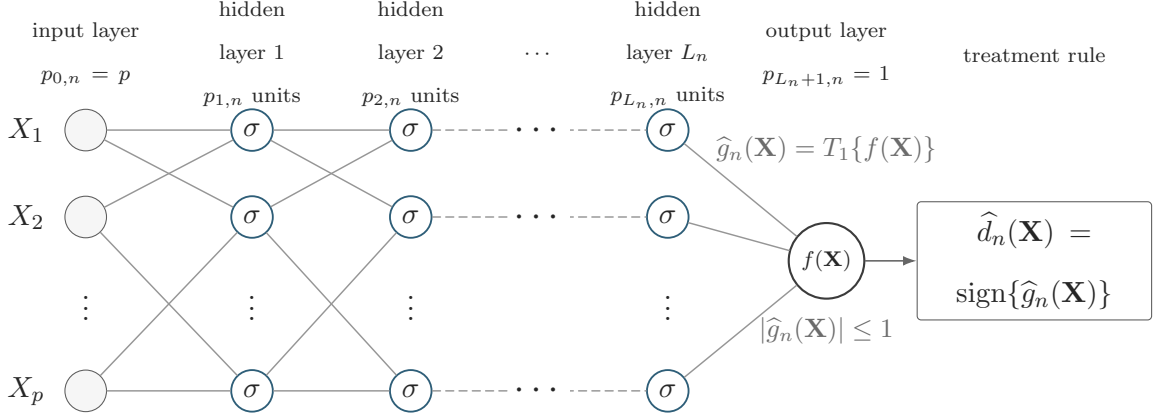

\begin{example}[\emph{Deep ReLU neural networks}]\label{ss:relu-sieves}
A feedforward deep neural network forms a score by alternating affine transformations with a nonlinear activation. For the nonlinear activation, we use the rectified linear unit (ReLU), $\sigma(u)=u_+$, applied componentwise. Composing ReLU layers produces a flexible piecewise affine score; the architectural complexity is summarized by the depth, layer widths, and number of nonzero parameters. Without loss of generality and after a componentwise rescaling, let $\bX\in[0,1]^p$. For depth $L$ and width vector $\mathbf p=(p_0,\ldots,p_{L+1})$ with $p_0=p$ and $p_{L+1}=1$, let $W_\ell\in\mathbb R^{p_{\ell+1}\times p_\ell}$ for $0\leq\ell\leq L$ be the within-layer weights and $v_\ell\in\mathbb R^{p_\ell}$ for $1\leq\ell\leq L$ be the within-layer intercepts. Write $\sigma_v(u)=(u-v)_+$ componentwise and construct $ h_0(\bX)=\bX$, $h_\ell(\bX)=\sigma_{v_\ell}\{W_{\ell-1}h_{\ell-1}(\bX)\}$ for $\ell=1,\ldots,L$, and $f(\bX)=W_Lh_L(\bX)$. Following \citet{schmidt2020nonparametric}, let $\mathcal F(L,\mathbf p,s,F)$ contain the resulting networks satisfying $\max_{0\leq\ell\leq L}\|W_\ell\|_\infty
 \vee\max_{1\leq\ell\leq L}\|v_\ell\|_\infty\leq1$, $\sum_{\ell=0}^L\|W_\ell\|_0
 +\sum_{\ell=1}^L\|v_\ell\|_0\leq s$, and $ \|f\|_\infty\leq F$, where $\|\cdot\|_\infty:=\max_{i,j}|\cdot_{ij}|$ denotes the maximum absolute entry that bounds the within-layer coefficients, and $\|\cdot\|_0:=\sum_{i,j}\mone(\cdot_{i,j}\neq0)$ denotes the number of nonzero entries that measure sparsity. For the bounded-hinge analysis, set $F=1$. This output restriction does not alter the induced treatment rule. That is, any scalar score $u$ can be post-processed using the fixed hard tanh output activation $T_1(u)$, which takes values in $[-1,1]$ and preserves the sign of $u$. Figure \ref{fig:sparse-relu-sieve} illustrates this construction. Set $\calG_n=\mathcal F(L_n,\mathbf p_n,s_n,1)$ and $V_n=\prod_{\ell=0}^{L_n+1}(p_{\ell,n}+1)$. The sparse-network metric entropy bound in Lemma 5 of \citet{schmidt2020nonparametric} yields the empirical process factor $q_n(\delta)=\sqrt{[{(s_n+1)\log\{2n_2(L_n+1)V_n^2\}+\log(2/\delta)}]/{n_2}}+n_2^{-1}$, which governs the second-stage fluctuation $\Delta_{n,H}^{\mathrm{linear}}(\calG_n;\widehat\eta)$ and has leading order $\sqrt{(s_n+1)/n_2}$ up to a logarithmic factor. The corollary below is a direct specialization of Theorem \ref{thm:bounded-hinge-sieve}. The universal Neyman orthogonality preserves the bilinear product of the two nuisance estimation errors. 
\begin{corollary}\label{thm:relu-regret}
Suppose that the consistency and ignorability conditions in Assumption \ref{asp:iden} and the weak positivity condition in Assumption \ref{assump:weak-positivity} hold, $|Y|\leq M$ almost surely, $\|\widehat m\|_\infty\leq M$, and $\|\widehat e\|_\infty\leq1$. Let $\widehat g_n$ be a $\rho_n$-approximate empirical hinge minimizer over $\calG_n$ and set $\widehat d_n=\operatorname{sign}(\widehat g_n)$. Conditional on the nuisance-training split, for every $\delta\in(0,1)$, with probability at least $1-\delta$, $V(d^\ast)-V(\widehat d_n)
 \leq 2^{-1}\epsilon_n^{-1}(1-\epsilon_n)^{-1}\{a_n^H+16 Mq_n(\delta)+\rho_n
 +2\|\widehat e-e_0\|_2\|\widehat m-m_0\|_2\}$.
\end{corollary}
\end{example}

\section{Simulation experiments}\label{sec:simulation}

\begin{table}[ht!]
\caption{Core features of the simulation scenarios. The upper panel summarizes the decision boundaries and degree of overlap under the true data generating processes. Limited overlap indicates that some realized propensity scores may be close to $0$ or $1$. The lower panel reports the methods used by the direct methods to estimate the corresponding nuisance functions. ``SL'' denotes Super Learner \citep{vanderlaan2007super}.}
\label{tab:simulation-features}
\centering

\begingroup
\footnotesize
\linespread{1}\selectfont
\compacttablesetup

\setlength{\tabcolsep}{4.2pt}
\renewcommand{\arraystretch}{1.08}
\renewcommand{\tabularxcolumn}[1]{m{#1}}

\begin{tabularx}{\linewidth}{@{}>{\raggedright\arraybackslash\small}m{0.215\linewidth}*{4}{>{\centering\arraybackslash}X}@{}}
\toprule
\textbf{Feature} & \textbf{DGP 1} & \textbf{DGP 2} & \textbf{DGP 3} & \textbf{DGP 4} \\[0.30em]
\midrule
\addlinespace[0.35em]
Decision boundary & \shortstack[c]{Nonlinear\\polyhedral} & Depth-2 tree & Linear & Depth-2 tree \\[0.65em]
Overlap & Sufficient & Limited & Sufficient & Sufficient \\
\addlinespace[0.30em]
\midrule
\addlinespace[0.35em]
$\widehat{\pi}$ & SL & SL & \shortstack[c]{Parametric\\correct} & \shortstack[c]{Parametric\\misspecified} \\[0.65em]
$\widehat{m}$ & SL & SL & \shortstack[c]{Parametric\\correct} & \shortstack[c]{Parametric\\correct} \\[0.65em]
$\widehat{\mu}$ & SL & SL & \shortstack[c]{Parametric\\correct} & \shortstack[c]{Parametric\\correct} \\
\addlinespace[0.25em]
\bottomrule
\end{tabularx}

\endgroup
\end{table}

We evaluate ODRL under four data generating processes (DGPs) that vary in decision boundary complexity, degree of overlap, and nuisance estimation methods. Each Monte Carlo sample contains \(n=1000\) observations. As summarized in Table~\ref{tab:simulation-features}, DGP~1 has a complex nonlinear polyhedral decision boundary under sufficient overlap; DGP~2 combines a depth-2 axis-aligned tree boundary with limited overlap; DGP~3 is a fully parametric, correctly specified linear benchmark; and DGP~4 has a depth-2 tree boundary and sufficient true overlap but uses a misspecified logistic propensity score working model that can yield fitted probabilities close to one. 
Complete data generating mechanisms are given in Supplementary Section~\ref{ss:supporting-info-simulation}.


We consider five ODRL implementations: exact optimization over a binary VC sieve, using the linear-rule sieve in Example~\ref{ss:linear-sieves} for DGP~3 and the decision tree sieve in Example~\ref{ss:tree-sieves} for the remaining DGPs; a bounded SVM with hinge loss (Example~\ref{ss:svm-sieves}); an SVM with logistic loss (Example~\ref{ss:svm-sieves}); a ReLU network with hinge loss (Example~\ref{ss:relu-sieves}); and a ReLU network with logistic loss (Example~\ref{ss:relu-sieves}). We compare these implementations with the nine representative direct and indirect methods listed in Table~\ref{tab:dgp1234-results}, together with an additional indirect method proposed by \citet{wallace2015doubly}. All direct methods use common five-fold cross-fitted nuisance predictions. In particular, the direct ODRL implementation and the methods of \citet{athey2021policy} and \citet{kallus2021more} use the same policy class and optimization routine within each DGP, so the comparison isolates differences in their policy learning objectives rather than in flexibility of the rule class. Nuisance functions
are estimated by Super Learner \citep{vanderlaan2007super} in DGPs~1 and~2, by correctly specified
parametric models in DGP~3, and by a misspecified propensity model
together with correctly specified outcome models in DGP~4. Across 500
replications, we evaluate regret, misclassification rate (MCR), and
runtime. Further implementation and tuning details are provided in
Supplementary Section~\ref{ss:supporting-info-simulation}. 


Table~\ref{tab:dgp1234-results} shows that the ODRL implementations are either the best-performing methods or strongly competitive across all four settings when paired with a suitable second-stage sieve. In DGP~1, the ODRL ReLU-logistic and ReLU-hinge learners achieve the two smallest regrets and MCRs. Notably, the bounded-score ReLU-hinge learner nearly matches its logistic-loss counterpart, although logistic loss does not preserve universal Neyman orthogonality. The ODRL SVM-logistic and SVM-hinge learners also outperform existing direct methods based on surrogate relaxation, including efficient augmentation and relaxation learning (EARL) \citep{zhao2019efficient}, RWL, and OWL, while requiring less computation. RWL is the most computationally intensive method, requiring more than $2.7$ hours, approximately 5,000 times the runtime of the ODRL SVM-hinge learner. By contrast, all three decision-tree-based methods have MCRs near 49\%, reflecting the nonvanishing approximation error induced by imposing a shallow-tree policy class on the complex decision boundary. Thus, flexible ReLU and SVM sieves provide practical surrogate-based approaches to learning highly nonlinear ITRs.

DGP~3 gives the parametric benchmark. When the nuisance
functions and optimal boundary are all linear and correctly specified,
All direct methods perform similarly and generally slightly underperform the indirect methods. This modest advantage of the indirect methods is consistent with the efficiency gain available under a correctly specified outcome or blip model, a pattern also observed in \cite{mo2022efficient}. Computation likewise favors the model-based procedures in this setting. In particular, exact optimization over linear rules requires substantially more computation. In this benchmark, direct learning entails some loss of parametric efficiency and computational convenience in exchange for avoiding restrictive model assumptions.

In DGP~2, where the true propensity score has support approximately \([0.004,0.996]\) and the depth-2 tree class is correctly specified, direct ODRL attains the smallest regret and MCR, \(0.0201\) and \(1.61\%\), respectively, compared with \(0.0275\) and \(2.20\%\) for the method of \citet{kallus2021more} and \(0.0823\) and \(6.58\%\) for that of \citet{athey2021policy}. Under surrogate relaxation, the ODRL-based methods also outperform EARL, RWL, and OWL, whose criteria contain inverse propensity score factors and yield substantially larger regrets. DGP~4 represents a plausible practical setting in which true overlap is sufficient, but a misspecified logistic-linear propensity score model yields fitted probabilities as high as \(0.992\). Direct ODRL again performs best, with a regret of \(0.0159\) and an MCR of \(0.75\%\), whereas, for example, the corresponding regrets are \(0.0860\) for the method of \cite{athey2021policy}, and \(0.1428\) for the method of \citet{kallus2021more}. 

\begin{sidewaystable}[p]
\caption{Monte Carlo performance under DGPs~1--4 over 500 replications.
Entries are mean regret and MCR, with Monte Carlo standard errors in
parentheses, and mean runtime in seconds; MCR is reported on the
$\times100\%$ scale. Lower regret and MCR are better; bold indicates the best-performing method within each DGP.}
\label{tab:dgp1234-results}
\centering

\begingroup
\small
\linespread{1}\selectfont
\compacttablesetup

\newcommand{\rescell}[2]{\shortstack{#1\\(#2)}}
\newcommand{\bestrescell}[2]{\shortstack{\textbf{#1}\\\textbf{(#2)}}}
\newcommand{\phcell}{--}

\setlength{\tabcolsep}{2.1pt}
\renewcommand{\arraystretch}{1.06}
\renewcommand{\tabularxcolumn}[1]{m{#1}}

\begin{tabularx}{\textheight}{@{}>{\raggedright\arraybackslash}m{0.215\textheight}*{12}{>{\centering\arraybackslash}X}@{}}
\toprule
\addlinespace[0.25em]
\textbf{Method} & \multicolumn{3}{c}{\textbf{DGP 1}} & \multicolumn{3}{c}{\textbf{DGP 2}} & \multicolumn{3}{c}{\textbf{DGP 3}} & \multicolumn{3}{c}{\textbf{DGP 4}} \\
\cmidrule(lr){2-4}\cmidrule(lr){5-7}\cmidrule(lr){8-10}\cmidrule(lr){11-13}
& \shortstack{Regret\\(MCSE)} & \shortstack{MCR\\(MCSE)} & \shortstack{Runtime\\(s)}
& \shortstack{Regret\\(MCSE)} & \shortstack{MCR\\(MCSE)} & \shortstack{Runtime\\(s)}
& \shortstack{Regret\\(MCSE)} & \shortstack{MCR\\(MCSE)} & \shortstack{Runtime\\(s)}
& \shortstack{Regret\\(MCSE)} & \shortstack{MCR\\(MCSE)} & \shortstack{Runtime\\(s)} \\[0.25em]
\midrule
\addlinespace[0.30em]

ODRL (ReLU hinge)
& \rescell{0.6768}{0.0032} & \rescell{30.42}{0.10} & 460.3
& \rescell{0.2808}{0.0022} & \rescell{22.46}{0.17} & 37.8
& \rescell{0.0609}{0.0070} & \rescell{9.42}{0.32} & 1.9
& \rescell{0.2971}{0.0015} & \rescell{18.26}{0.13} & 37.6 \\

\addlinespace[0.40em]
ODRL (ReLU logistic)
& \bestrescell{0.6754}{0.0033} & \bestrescell{30.37}{0.10} & 314.0
& \rescell{0.2847}{0.0022} & \rescell{22.77}{0.18} & 38.2
& \rescell{0.0430}{0.0056} & \rescell{7.98}{0.26} & 2.1
& \rescell{0.3005}{0.0014} & \rescell{18.57}{0.12} & 38.7 \\

\addlinespace[0.40em]
ODRL (SVM logistic)
& \rescell{0.7987}{0.0040} & \rescell{34.25}{0.12} & 40.5
& \rescell{0.2616}{0.0018} & \rescell{20.93}{0.15} & 13.4
& \rescell{0.0059}{0.0002} & \rescell{3.46}{0.05} & $<0.1$
& \rescell{0.2533}{0.0012} & \rescell{14.89}{0.10} & 7.5 \\

\addlinespace[0.40em]
ODRL (SVM hinge)
& \rescell{0.8974}{0.0029} & \rescell{37.10}{0.09} & 2.0
& \rescell{0.2956}{0.0022} & \rescell{23.65}{0.17} & 7.5
& \rescell{0.0085}{0.0002} & \rescell{4.23}{0.05} & 1.1
& \rescell{0.3352}{0.0027} & \rescell{22.27}{0.24} & 5.8 \\

\addlinespace[0.40em]
ODRL (VC)
& \rescell{1.3126}{0.0013} & \rescell{49.08}{0.04} & 0.5
& \bestrescell{0.0201}{0.0008} & \bestrescell{1.61}{0.06} & 0.5
& \rescell{0.0234}{0.0005} & \rescell{6.98}{0.07} & 2{,}400.4
& \bestrescell{0.0159}{0.0012} & \bestrescell{0.75}{0.05} & 4.3 \\

\addlinespace[0.15em]
\midrule
\addlinespace[0.30em]

EARL \citep{zhao2019efficient}
& \rescell{0.9365}{0.0098} & \rescell{38.23}{0.28} & 359.1
& \rescell{0.4148}{0.0046} & \rescell{33.19}{0.37} & 142.2
& \rescell{0.0052}{0.0001} & \rescell{3.28}{0.04} & 0.2
& \rescell{0.3892}{0.0069} & \rescell{27.15}{0.55} & 105.1 \\

\addlinespace[0.40em]
RWL \citep{zhou2017residual}
& \rescell{1.1302}{0.0059} & \rescell{43.83}{0.17} & 10{,}608.4
& \rescell{0.4145}{0.0064} & \rescell{33.16}{0.51} & 9{,}145.0
& \rescell{0.0099}{0.0003} & \rescell{4.50}{0.06} & 0.8
& \rescell{0.5377}{0.0087} & \rescell{38.85}{0.72} & 676.7 \\

\addlinespace[0.40em]
OWL \citep{zhao2012estimating}
& \rescell{1.2792}{0.0030} & \rescell{48.13}{0.08} & 27.9
& \rescell{0.6112}{0.0016} & \rescell{48.90}{0.13} & 23.2
& \rescell{0.0193}{0.0005} & \rescell{6.34}{0.08} & 2.1
& \rescell{0.6624}{0.0014} & \rescell{49.05}{0.12} & 31.3 \\

\addlinespace[0.40em]
CAIPWL \citep{athey2021policy}
& \rescell{1.3125}{0.0014} & \rescell{49.08}{0.04} & 0.5
& \rescell{0.0823}{0.0049} & \rescell{6.58}{0.39} & 0.5
& \rescell{0.0240}{0.0005} & \rescell{7.07}{0.07} & 2{,}400.3
& \rescell{0.0860}{0.0020} & \rescell{3.08}{0.12} & 4.2 \\

\addlinespace[0.40em]
Kallus \citep{kallus2021more}
& \rescell{1.3136}{0.0013} & \rescell{49.11}{0.04} & 0.5
& \rescell{0.0275}{0.0011} & \rescell{2.20}{0.09} & 0.5
& \rescell{0.0231}{0.0005} & \rescell{6.92}{0.07} & 2{,}400.1
& \rescell{0.1428}{0.0022} & \rescell{5.18}{0.09} & 4.2 \\

\addlinespace[0.15em]
\midrule
\addlinespace[0.30em]

Q-learning \citep{qian2011performance}
& \rescell{1.3354}{0.0010} & \rescell{49.76}{0.03} & 0.1
& \rescell{0.5658}{0.0022} & \rescell{45.26}{0.18} & 0.3
& \rescell{0.0026}{0.0001} & \rescell{2.32}{0.03} & 0.4
& \rescell{0.2030}{0.0011} & \rescell{8.97}{0.10} & 0.9 \\

\addlinespace[0.40em]
dWOLS \citep{wallace2015doubly}
& \rescell{1.3088}{0.0017} & \rescell{48.92}{0.05} & 0.1
& \rescell{0.4200}{0.0024} & \rescell{33.60}{0.20} & 0.1
& \rescell{0.0037}{0.0001} & \rescell{2.78}{0.03} & 0.1
& \rescell{0.2892}{0.0014} & \rescell{17.30}{0.11} & 0.2 \\

\addlinespace[0.40em]
A-learning \citep{schulte2015q}
& \rescell{1.3243}{0.0013} & \rescell{49.41}{0.04} & 0.2
& \rescell{0.4445}{0.0047} & \rescell{35.56}{0.37} & 0.3
& \bestrescell{0.0019}{0.0001} & \bestrescell{1.93}{0.03} & 0.1
& \rescell{0.2018}{0.0013} & \rescell{9.17}{0.12} & 0.3 \\

\addlinespace[0.40em]
E-Learning \citep{mo2022efficient}
& \rescell{1.3129}{0.0016} & \rescell{49.05}{0.05} & 4.7
& \rescell{0.4073}{0.0046} & \rescell{32.58}{0.37} & 5.3
& \rescell{0.0019}{0.0001} & \rescell{1.94}{0.04} & 0.3
& \rescell{0.4563}{0.0103} & \rescell{31.38}{0.88} & 6.6 \\

\addlinespace[0.40em]
D-Learning \citep{qi2018d}
& \rescell{1.3178}{0.0015} & \rescell{49.16}{0.05} & $<0.1$
& \rescell{0.5018}{0.0050} & \rescell{40.14}{0.40} & 0.1
& \rescell{0.0062}{0.0003} & \rescell{3.39}{0.08} & $<0.1$
& \rescell{0.6642}{0.0024} & \rescell{49.21}{0.20} & 0.3 \\

\addlinespace[0.25em]
\bottomrule
\end{tabularx}

\endgroup
\end{sidewaystable}


Overall, the experiments show that ODRL can accommodate multiple representations of the optimal ITR. First, flexible ReLU or SVM sieves are valuable for complex boundaries, whereas exact VC optimization is especially accurate and inexpensive when a shallow tree is appropriate. Across both true and fitted poor overlap regimes, the ODRL-based methods provide the most stable performance.

\section{Empirical analysis of Right Heart Catheterization}
\label{sec:data-rhc}

We apply ODRL to observational data on critically ill adults from the Right Heart Catheterization (RHC) study \citep{hirano2001rhc} in the main manuscript and to data from a randomized multinational field experiment on climate policy in Section~\ref{sec:data-application} of the Supplementary Material. We code RHC within 24 hours of study entry as \(A=1\) and no RHC as \(A=-1\). The reward \(Y\) is survival through 30 days. The analytic sample includes \(5735\) patients, of whom \(2184\) received RHC and \(3551\) did not. 
The published analysis estimated an average treatment effect of \(-6.2\) percentage points on 30-day survival for RHC relative to no RHC. This negative average treatment effect motivated our exploratory analysis of individualized recommendations for RHC. The estimated propensity scores ranged from \(0.001\) to \(0.960\), with \(2.58\%\) below \(0.05\), indicating limited overlap, particularly in the lower tail. A graphical assessment of overlap across treatment groups is presented in Figure~\ref{fig:rhc-propensity-overlap} of the Supplementary Material.

We adjust for 28 baseline covariates that may serve as potential effect modifiers or confounders and span demographic, functional, diagnostic, prognostic, and physiological domains. Consistent with the simulation experiments, we consider five ODRL implementations: bounded ReLU with hinge loss (Example~\ref{ss:relu-sieves}), ReLU with logistic loss (Example~\ref{ss:relu-sieves}), Gaussian RKHS with hinge loss (Example~\ref{ss:svm-sieves}), Gaussian RKHS with logistic loss (Example~\ref{ss:svm-sieves}), and exact optimization over depth-2 decision trees (Example~\ref{ss:tree-sieves}). We additionally consider an Athey--Wager depth-2 policy tree implemented using \texttt{policytree} and Q-learning implemented using \texttt{DynTxRegime}. We obtain five-fold cross-fitted nuisance estimates using Super Learner \citep{vanderlaan2007super} with a library comprising \texttt{SL.mean}, \texttt{SL.glm}, \texttt{SL.glmnet}, \texttt{SL.earth}, and \texttt{SL.gam}. 

Table~\ref{tab:rhc-full-policy} reports pairwise agreement among the optimal ITRs estimated by the five ODRL learners and the two comparison methods in the RHC study. The proportions recommended to receive RHC range from 12.9\% to 43.6\%, with pairwise agreement rates ranging from 58.9\% to 87.6\%. Figure~\ref{fig:rhc-fit-the-fit} presents a depth-2 decision tree obtained by exact optimization of the ODRL objective. The rule recommends RHC for patients with a PaO$_2$/FiO$_2$ ratio \(\leq188.31\) whose APACHE score is \(\leq47.00\), and for patients with a higher PaO$_2$/FiO$_2$ ratio whose baseline temperature is \(\leq36.50\,^{\circ}\mathrm{C}\). Clinically, these splits potentially suggest that the estimated benefit of RHC may depend on hypoxemic respiratory impairment, overall illness severity, and hypothermia as an indicator of physiological instability. Finally, post-learning value estimation using ten-fold cross-fitting yields estimated survival values of \(68.52\%\), \(68.21\%\), \(68.81\%\), \(68.96\%\), \(67.78\%\), \(67.63\%\), and \(68.84\%\) under the estimated optimal ITRs obtained using ReLU-H, ReLU-L, RBF-H, RBF-L, the ODRL decision tree, the Athey--Wager CAIPWL decision tree, and Q-learning, respectively. 
Relative to assigning RHC to every patient, the corresponding improvements in survival are \(3.75\), \(3.43\), \(4.03\), \(4.19\), \(3.01\), \(2.86\), and \(4.07\) percentage points, respectively.

\begin{figure}[ht!]
\centering
\includegraphics[width=0.96\textwidth]
{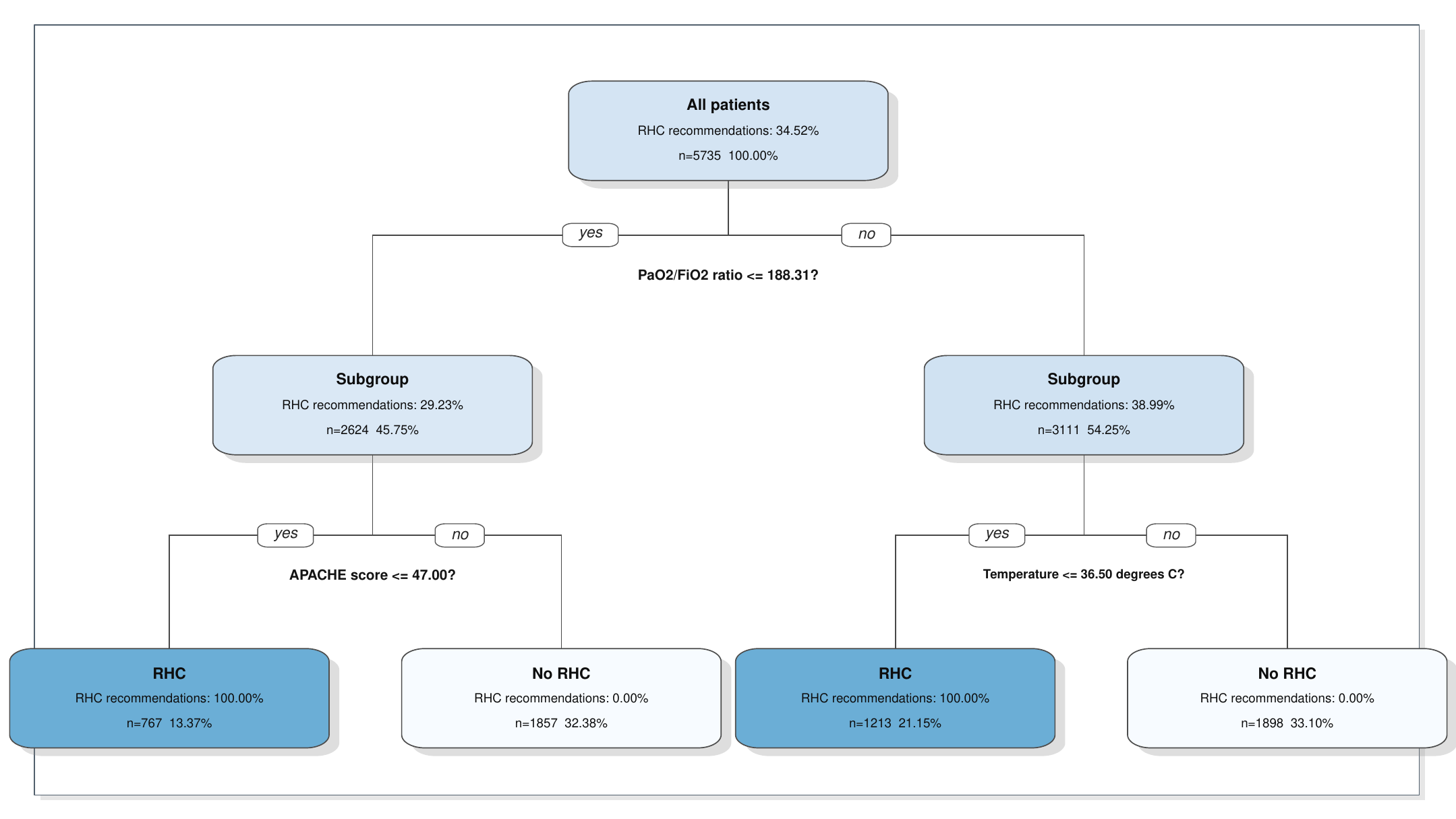}
\caption{Depth-2 axis-aligned decision tree fitted to the RHC data using ODRL. Terminal nodes report the recommended treatment; boxes report the percentage of patients for whom the estimated optimal ITR recommends RHC, the node size, and the corresponding sample proportion. }
\label{fig:rhc-fit-the-fit}
\end{figure}

\begin{table}[ht!]
\centering
\caption{Pairwise agreement among the five fitted ODRL rules and two comparator rules in the RHC study. The row labeled Observed reports agreement with observed RHC receipt. ReLU-H and ReLU-L denote ReLU with hinge and logistic losses, respectively; RBF-H and RBF-L denote Gaussian RKHS with hinge and logistic losses, respectively; ODRL-T denotes exact optimization over depth-2 decision trees; CAIPWL denotes the Athey--Wager depth-2 policy tree.}
\label{tab:rhc-full-policy}

\begingroup
\footnotesize
\linespread{1}\selectfont
\compacttablesetup

\setlength{\tabcolsep}{4.0pt}
\renewcommand{\arraystretch}{0.86}
\setlength{\aboverulesep}{0.15ex}
\setlength{\belowrulesep}{0.9ex}

\begin{tabular}{@{}lccccccc@{}}
\toprule
& ReLU-H & ReLU-L & RBF-H & RBF-L & ODRL-T & CAIPWL & Q-Learning \\
\midrule
ReLU-H & 1 & 0.768 & 0.730 & 0.764 & 0.597 & 0.616 & 0.705 \\[1ex]
ReLU-L &   & 1     & 0.682 & 0.717 & 0.590 & 0.589 & 0.639 \\[1ex]
RBF-H  &   &       & 1     & 0.876 & 0.684 & 0.690 & 0.817 \\[1ex]
RBF-L  &   &       &       & 1     & 0.666 & 0.674 & 0.790 \\[1ex]
ODRL-T &   &       &       &       & 1     & 0.667 & 0.665 \\[1ex]
CAIPWL   &   &       &       &       &       & 1     & 0.681 \\[1ex]
Q-Learning&   &       &       &       &       &       & 1     \\[1ex]
Observed $A$ & 0.490 & 0.481 & 0.543 & 0.525 & 0.515 & 0.537 & 0.561 \\
\midrule
RHC, $n$ (\%) &
2188 (38.2) & 2503 (43.6) & 1253 (21.8) & 1636 (28.5) &
1804 (31.5) & 1617 (28.2) & 740 (12.9) \\
\bottomrule
\end{tabular}

\endgroup
\end{table}

\section{Concluding remarks}\label{sec:conclusion}

We propose orthogonal double residual learning, a two-stage, cross-fitted framework that directly learns optimal individualized treatment rules through a cost-sensitive classification loss based on treatment and outcome residuals. The loss is Fisher consistent and universally Neyman orthogonal, requires no restrictive model assumptions, and contains no inverse propensity score weighting components. Consequently, the proposed method accommodates complex decision boundaries and is robust to limited overlap and nuisance estimation error. We establish nonasymptotic value function regret bounds for exact optimization over binary VC sieves and calibrated surrogate learning over flexible score sieves. As a new insight, we show that hinge loss over bounded-score sieves preserves universal Neyman orthogonality, whereas generic surrogate losses generally do not.

In practice, nuisance functions may be estimated flexibly using cross-fitting, and the second-stage sieve may be selected based on the anticipated complexity of the decision boundary and the desired degree of interpretability. Exact optimization over linear rules or shallow decision trees is attractive when such structures are plausible and interpretability is a priority. For complex nonlinear boundaries, bounded-score hinge learning over SVM or ReLU sieves provides a principled default option, whereas other calibrated surrogates, such as logistic loss, may improve approximation or optimization but introduce first-order sensitivity to nuisance estimation error. The simulations support this recommendation, since flexible SVM and ReLU implementations perform well under complex decision boundaries, exact tree-based ODRL performs best when a shallow tree is correctly specified, and ODRL remains stable under limited overlap. In the correctly specified linear benchmark, model-based indirect methods retain a modest efficiency advantage. For practical implementation of our proposed methods, we also provide the \texttt{odrlITR} R package, available at \url{https://github.com/deckardt98/odrlITR}.

Directions for future research include data-adaptive selection of rule sieves and surrogate losses with corresponding regret guarantees, as well as post-learning inference for the value of learned rules. Extensions to multiarm and multistage treatments and resource-constrained decision rule classes would further broaden the scope of ODRL, and will be left for future research.


\bibliographystyle{chicago}
\bibliography{ref}

@article{nie2021quasi,
  title={Quasi-oracle estimation of heterogeneous treatment effects},
  author={Nie, Xinkun and Wager, Stefan},
  journal={Biometrika},
  volume={108},
  number={2},
  pages={299--319},
  year={2021},
  publisher={Oxford University Press}
}

@article{robinson1988root,
  title={Root-N-consistent semiparametric regression},
  author={Robinson, Peter M},
  journal={Econometrica},
  pages={931--954},
  year={1988},
  publisher={JSTOR}
}

@article{foster2023orthogonal,
  title={Orthogonal statistical learning},
  author={Foster, Dylan J and Syrgkanis, Vasilis},
  journal={The Annals of Statistics},
  volume={51},
  number={3},
  pages={879--908},
  year={2023},
  publisher={Institute of Mathematical Statistics}
}

@article{chen2007large,
  title={Large sample sieve estimation of semi-nonparametric models},
  author={Chen, Xiaohong},
  journal={Handbook of Econometrics},
  volume={6},
  pages={5549--5632},
  year={2007},
  publisher={Elsevier}
}

@article{dml,
    author = {Chernozhukov, Victor and Chetverikov, Denis and Demirer, Mert and Duflo, Esther and Hansen, Christian and Newey, Whitney and Robins, James},
    title = "{Double/debiased machine learning for treatment and structural parameters}",
    journal = {The Econometrics Journal},
    volume = {21},
    number = {1},
    pages = {C1-C68},
    year = {2018},
    month = {01},
    issn = {1368-4221},
    doi = {10.1111/ectj.12097},
    url = {https://doi.org/10.1111/ectj.12097},
    eprint = {https://academic.oup.com/ectj/article-pdf/21/1/C1/27684918/ectj00c1.pdf},
}

@article{morzywolek2023weighted,
  title={On weighted orthogonal learners for heterogeneous treatment effects},
  author={Morzywolek, Pawel and Decruyenaere, Johan and Vansteelandt, Stijn},
  journal={Statistical Science},
  year={2025}
}

@article{li2018balancing,
  title={Balancing covariates via propensity score weighting},
  author={Li, Fan and Morgan, Kari Lock and Zaslavsky, Alan M},
  journal={Journal of the American Statistical Association},
  volume={113},
  number={521},
  pages={390--400},
  year={2018},
  publisher={Taylor \& Francis}
}

@article{qian2011performance,
  title={Performance guarantees for individualized treatment rules},
  author={Qian, Min and Murphy, Susan A},
  journal={Annals of Statistics},
  volume={39},
  number={2},
  pages={1180},
  year={2011}
}

@article{rubin1974estimating,
  title={Estimating causal effects of treatments in randomized and nonrandomized studies.},
  author={Rubin, Donald B},
  journal={Journal of Educational Psychology},
  volume={66},
  number={5},
  pages={688},
  year={1974},
  publisher={American Psychological Association}
}

@article{athey2021policy,
  title={Policy learning with observational data},
  author={Athey, Susan and Wager, Stefan},
  journal={Econometrica},
  volume={89},
  number={1},
  pages={133--161},
  year={2021},
  publisher={Wiley Online Library}
}

@article{chernozhukov2024applied,
  title={Applied causal inference powered by ML and AI},
  author={Chernozhukov, Victor and Hansen, Christian and Kallus, Nathan and Spindler, Martin and Syrgkanis, Vasilis},
  journal={arXiv preprint arXiv:2403.02467},
  year={2024}
}

@article{zhao2019efficient,
  title={Efficient augmentation and relaxation learning for individualized treatment rules using observational data},
  author={Zhao, Ying-Qi and Laber, Eric B and Ning, Yang and Saha, Sumona and Sands, Bruce E},
  journal={Journal of Machine Learning Research},
  volume={20},
  number={48},
  pages={1--23},
  year={2019}
}

@article{zhou2017residual,
  title={Residual weighted learning for estimating individualized treatment rules},
  author={Zhou, Xin and Mayer-Hamblett, Nicole and Khan, Umer and Kosorok, Michael R},
  journal={Journal of the American Statistical Association},
  volume={112},
  number={517},
  pages={169--187},
  year={2017},
  publisher={Taylor \& Francis}
}

@article{zhao2012estimating,
  title={Estimating individualized treatment rules using outcome weighted learning},
  author={Zhao, Yingqi and Zeng, Donglin and Rush, A John and Kosorok, Michael R},
  journal={Journal of the American Statistical Association},
  volume={107},
  number={499},
  pages={1106--1118},
  year={2012},
  publisher={Taylor \& Francis}
}

@article{zhao2009reinforcement,
  title={Reinforcement learning design for cancer clinical trials},
  author={Zhao, Yufan and Kosorok, Michael R and Zeng, Donglin},
  journal={Statistics in Medicine},
  volume={28},
  number={26},
  pages={3294--3315},
  year={2009},
  publisher={Wiley Online Library}
}

@article{zhou2023offline,
  title={Offline multi-action policy learning: Generalization and optimization},
  author={Zhou, Zhengyuan and Athey, Susan and Wager, Stefan},
  journal={Operations Research},
  volume={71},
  number={1},
  pages={148--183},
  year={2023},
  publisher={INFORMS}
}

@article{schulte2015q,
  title={Q-and A-learning methods for estimating optimal dynamic treatment regimes},
  author={Schulte, Phillip J and Tsiatis, Anastasios A and Laber, Eric B and Davidian, Marie},
  journal={Statistical Science},
  volume={29},
  number={4},
  pages={640},
  year={2015}
}

@article{kallus2021more,
  title={More efficient policy learning via optimal retargeting},
  author={Kallus, Nathan},
  journal={Journal of the American Statistical Association},
  volume={116},
  number={534},
  pages={646--658},
  year={2021},
  publisher={Taylor \& Francis}
}

@article{schmidt2020nonparametric,
  title={Nonparametric regression using deep neural networks with {ReLU} activation function},
  author={Schmidt-Hieber, Johannes},
  journal={The Annals of Statistics},
  volume={48},
  number={4},
  pages={1875--1897},
  year={2020}
}

@book{vanderVaartWellner2023,
  author    = {van der Vaart, A. W. and Wellner, Jon A.},
  title     = {Weak Convergence and Empirical Processes:
               With Applications to Statistics},
  edition   = {2},
  series    = {Springer Series in Statistics},
  publisher = {Springer},
  address   = {Cham},
  year      = {2023},
  doi       = {10.1007/978-3-031-29040-4},
  isbn      = {978-3-031-29040-4}
}

@article{mi2019bagging,
  title={Bagging and deep learning in optimal individualized treatment rules},
  author={Mi, Xinlei and Zou, Fei and Zhu, Ruoqing},
  journal={Biometrics},
  volume={75},
  number={2},
  pages={674--684},
  year={2019},
  publisher={Oxford University Press}
}

@article{manski1975maximum,
  title={Maximum score estimation of the stochastic utility model of choice},
  author={Manski, Charles F},
  journal={Journal of Econometrics},
  volume={3},
  number={3},
  pages={205--228},
  year={1975},
  publisher={Elsevier}
}

@article{mcdiarmid1989method,
  title={On the method of bounded differences},
  author={McDiarmid, Colin and others},
  journal={Surveys in Combinatorics},
  volume={141},
  number={1},
  pages={148--188},
  year={1989},
  publisher={Norwich}
}

@article{mo2022efficient,
  title={Efficient learning of optimal individualized treatment rules for heteroscedastic or misspecified treatment-free effect models},
  author={Mo, Weibin and Liu, Yufeng},
  journal={Journal of the Royal Statistical Society Series B: Statistical Methodology},
  volume={84},
  number={2},
  pages={440--472},
  year={2022},
  publisher={Oxford University Press}
}

@article{wallace2015doubly,
  title={Doubly-robust dynamic treatment regimen estimation via weighted least squares},
  author={Wallace, Michael P and Moodie, Erica EM},
  journal={Biometrics},
  volume={71},
  number={3},
  pages={636--644},
  year={2015},
  publisher={Oxford University Press}
}

@article{qi2018d,
author = {Zhengling Qi and Yufeng Liu},
title = {{D-learning to estimate optimal individual treatment rules}},
volume = {12},
journal = {Electronic Journal of Statistics},
number = {2},
publisher = {Institute of Mathematical Statistics and Bernoulli Society},
pages = {3601 -- 3638},
year = {2018},
doi = {10.1214/18-EJS1480},
URL = {https://doi.org/10.1214/18-EJS1480}
}

@article{bartlett2006convexity,
  title={Convexity, classification, and risk bounds},
  author={Bartlett, Peter L and Jordan, Michael I and McAuliffe, Jon D},
  journal={Journal of the American Statistical Association},
  volume={101},
  number={473},
  pages={138--156},
  year={2006},
  publisher={Taylor \& Francis}
}

@book{luenberger1969optimization,
  author    = {Luenberger, David G.},
  title     = {Optimization by Vector Space Methods},
  publisher = {John Wiley \& Sons},
  address   = {New York},
  year      = {1969}
}

@book{bauschke2017convex,
  author    = {Bauschke, Heinz H. and Combettes, Patrick L.},
  title     = {Convex Analysis and Monotone Operator Theory
               in Hilbert Spaces},
  edition   = {2},
  publisher = {Springer},
  address   = {Cham},
  year      = {2017},
  doi       = {10.1007/978-3-319-48311-5}
}

@article{vanderlaan2007super,
  title = {Super Learner},
  author = {van der Laan, Mark J. and Polley, Eric C. and Hubbard, Alan E.},
  journal = {Statistical Applications in Genetics and Molecular Biology},
  year = {2007},
  volume = {6},
  number = {1},
  pages = {Article 25},
  doi = {10.2202/1544-6115.1309}
}

@book{ledoux1991probability,
  author    = {Ledoux, Michel and Talagrand, Michel},
  title     = {Probability in Banach Spaces:
               Isoperimetry and Processes},
  publisher = {Springer-Verlag},
  address   = {Berlin},
  year      = {1991},
  doi       = {10.1007/978-3-642-20212-4}
}

@article{pereira2025climate,
  author  = {Pereira, Miguel M. and Giger, Nathalie and Perez, Maria D. and Axelsson, Kaya},
  title   = {Encouraging Politicians to Act on Climate: A Field Experiment with Local Officials in Six Countries},
  journal = {American Journal of Political Science},
  year    = {2025},
  volume  = {69},
  number  = {1},
  pages   = {148--163},
  doi     = {10.1111/ajps.12841}
}

@book{scholkopf2002learning,
  title     = {Learning with Kernels: Support Vector Machines,
               Regularization, Optimization, and Beyond},
  author    = {Sch{\"o}lkopf, Bernhard and Smola, Alexander J.},
  year      = {2002},
  publisher = {MIT Press},
  address   = {Cambridge, MA}
}

@article{hirano2001rhc,
  author  = {Hirano, Keisuke and Imbens, Guido W.},
  title   = {Estimation of Causal Effects Using Propensity Score Weighting:
             An Application to Data on Right Heart Catheterization},
  journal = {Health Services and Outcomes Research Methodology},
  year    = {2001},
  volume  = {2},
  pages   = {259--278},
  doi     = {10.1023/A:1020371312283}
}

@inproceedings{lee2024effective,
  author    = {Lee, Joowon and Huling, Jared D. and Chen, Guanhua},
  title     = {An Effective Framework for Estimating Individualized Treatment Rules},
  booktitle = {Advances in Neural Information Processing Systems},
  volume    = {37},
  year      = {2024},
  doi       = {10.52202/079017-0270}
}

@article{bousquet2002bennett,
  title={A Bennett concentration inequality and its application to suprema of empirical processes},
  author={Bousquet, Olivier},
  journal={Comptes Rendus Mathematique},
  volume={334},
  number={6},
  pages={495--500},
  year={2002},
  publisher={Elsevier}
}

\clearpage
\pdfbookmark[0]{Supplementary Material}{supplementary-material}

\setcounter{page}{1}
\renewcommand{\thepage}{S\arabic{page}}
\setcounter{equation}{0}
\renewcommand{\theequation}{S\arabic{equation}}
\setcounter{figure}{0}
\renewcommand{\thefigure}{S\arabic{figure}}
\setcounter{table}{0}
\renewcommand{\thetable}{S\arabic{table}}
\setcounter{lemma}{0}
\renewcommand{\thelemma}{S\arabic{lemma}}

\renewcommand{\theHsection}{supp.\Alph{section}}
\renewcommand{\theHsubsection}{\theHsection.\arabic{subsection}}
\renewcommand{\theHsubsubsection}{\theHsubsection.\arabic{subsubsection}}
\renewcommand{\theHequation}{supp.\arabic{equation}}
\renewcommand{\theHfigure}{supp.\arabic{figure}}
\renewcommand{\theHtable}{supp.\arabic{table}}
\makeatletter
\@ifundefined{theHlemma}
  {\newcommand{\theHlemma}{supp.\arabic{lemma}}}
  {\renewcommand{\theHlemma}{supp.\arabic{lemma}}}
\makeatother

\allowdisplaybreaks
\spacingset{1.5}

\if1\blind
{
\begin{center}
  {\large\bfseries Supplementary Material for ``Orthogonal double residual learning for optimal individualized treatment rules''\par}
  \vspace{0.55em}
  {\normalsize Jiaqi Tong$^{1}$ and Fan Li$^{1,*}$\par}
  \vspace{0.15em}
  {\small $^{1}$Department of Biostatistics, Yale School of Public Health, New Haven, CT, USA\par}
  {\small $^{*}$\emph{email}: fan.f.li@yale.edu\par}
\end{center}
}
\fi

\if0\blind
{
  \bigskip
  \bigskip
  \bigskip
  \begin{center}
    {\Large\bf Supplementary Material for\\ ``Orthogonal double residual learning for optimal individualized treatment rules''}
  \end{center}
  \medskip
}
\fi

\smallskip

\section*{Supplementary Material}
\label{SM}
This supplementary material is organized as follows. Appendix A contains proofs of the technical results in the main manuscript. Appendix B provides supporting information for the simulation experiments and data application. Throughout the supplementary material, M.$(v)$ denotes equation $(v)$ in the main manuscript.

\hypersetup{linkcolor=black}
\setcounter{tocdepth}{2}
\etocsettocstyle{\section*{\contentsname}}{}
\etocsetlocaltop{part}
\footnotesize
\localtableofcontents
\normalsize
\hypersetup{linkcolor=red}
\clearpage


\appendix

\setcounter{section}{1}\setcounter{subsection}{0}
\section*{Appendix A: Proofs of technical results}
\addcontentsline{toc}{section}{Appendix A: Proofs of technical results}

\subsection{Proof of Theorem \ref{thm:valid-loss}}\label{ss:proof-fisher-consistency}
\begin{proof}
At the true nuisance functions, write $Z_0=(A-e_0(\bX))(Y-m_0(\bX))$. The part of the loss that depends on $d$ is
\begin{align*}
 -\frac{1}{2}\bE\{Z_0d(\bX)\}.
\end{align*}
By iterated expectation, we obtain
\begin{align*}
 &h_0(\bX):=\bE(Z_0\mid\bX)
 =\bE(AY\mid\bX)-\bE(A\mid\bX)\bE(Y\mid\bX)\\
 =& \pi(1,\bX)\mu(1,\bX)-\pi(-1,\bX)\mu(-1,\bX)-\lb \pi(1,\bX)-\pi(-1,\bX)\rb \lb \pi(1,\bX)m(1,\bX)+\pi(-1,\bX)m(-1,\bX)\rb\\
 =&2\pi(1,\bX)\pi(-1,\bX)\{\mu(1,\bX)-\mu(-1,\bX)\}\\
 =&2\pi(1,\bX)\pi(-1,\bX)\tau(\bX).
\end{align*}
Let $d^\ast=\operatorname{sign}(\tau)=\operatorname{sign}(h_0)$, using weak positivity. For every binary rule $d$ and fixed $\bX=\bx$, we observe the following pointwise identity:
\begin{align*}
    \frac{1}{2}h_0(\bx)\{d^\ast(\bx)-d(\bx)\}=|h_0(\bx)|\mone\{d(\bx)\neq d^\ast(\bx)\},
\end{align*}
which further implies 
\begin{align*}
 \calL(d;\eta_0)-\calL(d^\ast;\eta_0)
 =&\frac{1}{2}\bE\lbb h_0(\bX)\{d^\ast(\bX)-d(\bX)\}\rbb\\
 =&\bE\lbb |h_0(\bX)|\mone\{d(\bX)\neq d^\ast(\bX)\}\rbb\geq0.
\end{align*}
The difference is zero only if $d=d^\ast$ almost surely on $\{h_0\neq0\}=\{\tau\neq0\}$. On $\{\tau=0\}$ either treatment has the same loss and value. This proves all claims.
\end{proof}

\subsection{Proof of Theorem \ref{prop:Neyman-orthogonality}}
\begin{proof}
For an ambient target direction $v$,
\begin{align*}
 D_d\calL(\bar d;\eta)[v]
 =-\frac{1}{2}\bE\lbb Z_\eta(\mO)v(\bX)\rbb.
\end{align*}
Consider the nuisance path $e_t=e_0+t\dot e$ and $m_t=m_0+t\dot m$. Differentiating at $t=0$ gives
\begin{align*}
 D_\eta D_d\calL(\bar d;\eta_0)\lbb v,(\dot e,\dot m)\rbb
    =&\frac{\partial}{\partial t}\bE\lbb \frac{-\lb A-\lsb e_0+t\dot{e}\rsb\rb\lb Y-\lsb m_0+t\dot{m}\rsb\rb v}{2}\rbb|_{t=0}\\
 =&\frac{1}{2}\bE\lbb v(\bX)\{\dot e(\bX)(Y-m_0(\bX))
 +(A-e_0(\bX))\dot m(\bX)\}\rbb\\
 =&0,
\end{align*}
where differentiation uses the chain rule and the final equality follows from iterated expectation. The expression does not depend on $\bar d$, which proves universal Neyman orthogonality.
\end{proof}

\subsection{Proof of Theorem \ref{thm:sieve-regret}}\label{ss:proof-of-theorem{thm:sieve-regret}}
\begin{proof}
Let $\calD_n\subseteq\{-1,1\}^{\calX}$ be a fixed sieve or a sieve measurable with respect to the first split $\mathcal{F}_1$. Therefore, conditional on $\mathcal{F}_1$, both $\widehat\eta$ and $\calD_n$ are fixed. 
Moreover, fix $d\in\calD_n$ and a nuisance direction
$\dot\eta=(\dot e,\dot m)$. Consider the one-dimensional nuisance path
\[
\eta_t=(e_t,m_t):=(e+t\dot e,m+t\dot m),
\qquad t\in\mathbb R,
\]
and write $r(t):=\mathcal R(d;\eta_t)$. By the definition of the
target-dependent linear risk,
\begin{align*}
r(t)
=&-\frac{1}{2}
P\lbb
d(\bX)
\{A-e(\bX)-t\dot e(\bX)\}
\{Y-m(\bX)-t\dot m(\bX)\}
\rbb.
\end{align*}
The product of the two perturbed residuals has the exact expansion
\begin{align*}
&\{A-e(\bX)-t\dot e(\bX)\}
 \{Y-m(\bX)-t\dot m(\bX)\} \\
=&
\{A-e(\bX)\}\{Y-m(\bX)\}
-t\dot e(\bX)\{Y-m(\bX)\} -t\{A-e(\bX)\}\dot m(\bX)
+t^2\dot e(\bX)\dot m(\bX).
\end{align*}
Consequently,
\begin{align*}
r(t)
=&\mathcal R(d;\eta)
+\frac{t}{2}P\lbb
d(\bX)
\big\{
\dot e(\bX)(Y-m(\bX))
+(A-e(\bX))\dot m(\bX)
\big\}
\rbb -\\
&\frac{t^2}{2}
P\{d(\bX)\dot e(\bX)\dot m(\bX)\}.
\end{align*}
Thus, (i) using 
\begin{align*}
    D_\eta\mathcal R(d;\eta)
\lbb\dot\eta\rbb
=
\left.\frac{\partial}{\partial t}
\mathcal R(d;\eta_t)\right|_{t=0},
\end{align*}
we obtain
\begin{align}
 D_\eta\mathcal R(d;\eta)\lbb\dot\eta\rbb
 =&\frac{1}{2}P\lbb d(\bX)\{\dot e(\bX)(Y-m(\bX))
 +(A-e(\bX))\dot m(\bX)\}\rbb;
 \label{eq:affine-risk-first-derivative}    
\end{align}
and (ii) using
\begin{equation*}
D_\eta^2\mathcal R(d;\eta)
\lbb\dot\eta,\dot\eta\rbb
=
\left.\frac{\partial^2}{\partial t^2}
\mathcal R(d;\eta_t)\right|_{t=0},    
\end{equation*}
we obtain 
\begin{align}\label{eq:affine-risk-second-derivative} 
    D_\eta^2\mathcal R(d;\eta)
\lbb\dot\eta,\dot\eta\rbb=-P\{d(\bX)\dot e(\bX)\dot m(\bX)\}
\end{align}
and 
\begin{align}\label{eq:affine-risk-nuisance-curvature}
\left|
D_\eta^2\mathcal R(d;\eta)
\lbb\dot\eta,\dot\eta\rbb
\right|
&\leq
P\{|\dot e(\bX)\dot m(\bX)|\}
\notag\\
&\leq
\{P\dot e(\bX)^2\}^{1/2}
\{P\dot m(\bX)^2\}^{1/2}
\notag\\
=&\|\dot e\|_2\|\dot m\|_2,
\end{align}
where the first inequality uses
$|d(\bX)|=1$ for $d\in\calD_n\subseteq\{-1,1\}^{\calX}$, and the
second follows from the Cauchy--Schwarz inequality. 

Write $\delta_\eta=(\delta_e,\delta_m):=(\widehat e-e_0,\widehat m-m_0)$.
We first prove the generic oracle inequality: 
\begin{align}\label{eq:generic-orthogonal-sieve-bound}
 {\calL(\widehat d_n;\eta_0)-\calL(d_n^\circ;\eta_0)}
 {\leq r_{2,n}+\|\widehat e-e_0\|_2\|\widehat m-m_0\|_2.}
\end{align}
Suppose that \(\widehat d_n\) and \(d_n^\circ\in\calD_n\) satisfy M.\eqref{eq:generic-second-stage-rate}. Adding and subtracting the plug-in risks gives
\begin{align}\label{eq:foster-style-decomposition}
 &\mathcal R(\widehat d_n;\eta_0)-\mathcal R(d_n^\circ;\eta_0)\notag\\
 =&
 \lbb\mathcal R(\widehat d_n;\eta_0)-\mathcal R(\widehat d_n;\widehat\eta)\rbb
 +\lbb\mathcal R(d_n^\circ;\widehat\eta)-\mathcal R(d_n^\circ;\eta_0)\rbb
 +\lbb\mathcal R(\widehat d_n;\widehat\eta)-\mathcal R(d_n^\circ;\widehat\eta)\rbb\notag\\
 \leq&\lbb\mathcal R(\widehat d_n;\eta_0)-\mathcal R(\widehat d_n;\widehat\eta)\rbb
 +\lbb\mathcal R(d_n^\circ;\widehat\eta)-\mathcal R(d_n^\circ;\eta_0)\rbb
 +r_{2,n}.
\end{align}


Apply a second-order Taylor expansion with respect to $\eta$ to the two bracketed terms in \eqref{eq:foster-style-decomposition}. There exist nuisance values $\bar\eta_1,\bar\eta_2$ on the line segment joining $\eta_0$ and $\widehat\eta$ such that their sum equals
\begin{align}\label{eq:nuisance-taylor-pair}
 &-D_\eta\mathcal R(\widehat d_n;\eta_0)\lbb\delta_\eta\rbb
 +D_\eta\mathcal R(d_n^\circ;\eta_0)\lbb\delta_\eta\rbb-\frac{1}{2}D_\eta^2\mathcal R(\widehat d_n;\bar\eta_1)
 \lbb\delta_\eta,\delta_\eta\rbb
 +\frac{1}{2}D_\eta^2\mathcal R(d_n^\circ;\bar\eta_2)
 \lbb\delta_\eta,\delta_\eta\rbb.
\end{align}
The first-order terms cancel. To make this explicit, consider the
line segment in the ambient function space $\mathcal F$,
\begin{align*}
    d_t
    :=d_n^\circ+t(\widehat d_n-d_n^\circ),
    \qquad t\in[0,1].
\end{align*}
Although the intermediate functions $d_t$ need not be binary-valued,
the path is well defined in the ambient space $\mathcal F$. Since
$d_0=d_n^\circ$, $d_1=\widehat d_n$, and
$\dot d_t:=\partial d_t/\partial t=\widehat d_n-d_n^\circ$, the fundamental theorem of
calculus and the chain rule give
\begin{align*}
&D_\eta\mathcal R(\widehat d_n;\eta_0)
 \lbb\delta_\eta\rbb
-D_\eta\mathcal R(d_n^\circ;\eta_0)
 \lbb\delta_\eta\rbb
\\
=&
\int_0^1
\frac{\mathrm d}{\mathrm dt}
D_\eta\mathcal R(d_t;\eta_0)
\lbb\delta_\eta\rbb\,\mathrm dt
\\
=&
\int_0^1
D_dD_\eta\mathcal R(d_t;\eta_0)
\lbb\widehat d_n-d_n^\circ,\delta_\eta\rbb
\,\mathrm dt.
\end{align*}
Because $\mathcal R$ is linear in the target rule $d$, $ D_d^2D_\eta\mathcal R(d;\eta_0)\equiv0$. Consequently, $D_dD_\eta\mathcal R(d;\eta_0)$ does not depend on
the base point $d$. Hence, for any ambient $\bar d\in\mathcal F$,
\begin{align*}
&D_\eta\mathcal R(\widehat d_n;\eta_0)
 \lbb\delta_\eta\rbb
-D_\eta\mathcal R(d_n^\circ;\eta_0)
 \lbb\delta_\eta\rbb
\\
=&
D_dD_\eta\mathcal R(\bar d;\eta_0)
\lbb\widehat d_n-d_n^\circ,\delta_\eta\rbb
\\
=&
D_\eta D_d\mathcal R(\bar d;\eta_0)
\lbb\widehat d_n-d_n^\circ,\delta_\eta\rbb
\\
=&0.
\end{align*}
The mixed derivatives commute because $\mathcal R$ is polynomial in
$(d,\eta)$. The final equality follows from universal Neyman
orthogonality in Theorem~\ref{prop:Neyman-orthogonality}, noting that
the target derivatives of $\mathcal R$ and $\calL$ coincide.

Using \eqref{eq:affine-risk-second-derivative}, the remaining second-order terms in \eqref{eq:nuisance-taylor-pair} reduce exactly to
\begin{align*}
 \frac{1}{2}P\lbb\delta_e(\bX)\delta_m(\bX)
 \{\widehat d_n(\bX)-d_n^\circ(\bX)\}\rbb
 \leq\|\delta_e\|_{2}\|\delta_m\|_{2}=\|\widehat{e}-e_0\|_{2}\|\widehat{m}-m_0\|_{2},
\end{align*}
where $|\widehat d_n-d_n^\circ|\leq2$ and Cauchy--Schwarz were used. Substitution into \eqref{eq:foster-style-decomposition}, followed by the excess risk equivalence $\calL(d;\eta)-\calL(d';\eta)
 =\mathcal R(d;\eta)-\mathcal R(d';\eta)$, proves \eqref{eq:generic-orthogonal-sieve-bound}.

We next verify {M.\eqref{eq:generic-second-stage-rate} with \(r_{2,n}=\Delta_n(\calD_n;\widehat\eta)+\rho_n\)} for the approximate empirical minimizer. For any $d_n^\circ\in\calD_n$, M.\eqref{eq:sieve-erm} implies
\begin{align*}
 &\mathcal R(\widehat d_n;\widehat\eta)-\mathcal R(d_n^\circ;\widehat\eta)\\
 =&[\{P-\mathbb P_2\}\ell^{\mathrm{linear}}(\widehat d_n;\widehat\eta)]+[\{\mathbb P_2-P\}\ell^{\mathrm{linear}}(d_n^\circ;\widehat\eta)]+[\mathbb P_2\ell^{\mathrm{linear}}(\widehat d_n;\widehat\eta)-\mathbb P_2 \ell^{\mathrm{linear}}(d_n^\circ;\widehat\eta)]\\
 \leq&\left|\{P-\mathbb P_2\}\ell^{\mathrm{linear}}(\widehat d_n;\widehat\eta)\right|
 +\left|\{\mathbb P_2-P\}\ell^{\mathrm{linear}}(d_n^\circ;\widehat\eta)\right|+\rho_n\\
 \leq&\Delta_n(\calD_n;\widehat\eta)+\rho_n,
\end{align*}
where the first inequality follows from 
\begin{equation*}
\mathbb P_2\ell^{\mathrm{linear}}(\widehat d_n;\widehat\eta)
 \leq\inf_{d\in\calD_n}\mathbb P_2\ell^{\mathrm{linear}}(d;\widehat\eta)+\rho_n\leq \mathbb P_2\ell^{\mathrm{linear}}(d_n^\circ;\widehat\eta)+\rho_n,    
\end{equation*}
and the last inequality follows from
\begin{align*}
\{\mathbb P_2-P\}\ell^{\mathrm{linear}}(d_n^\circ;\widehat\eta)
=&(\mathbb P_2-P)\lb -2^{-1}Z_{\widehat\eta}(\mO)d_n^\circ(\bX)\rb\\
\leq& 2^{-1}\sup_{d\in\calD_n}
 \left|(\mathbb P_2-P)\lb Z_{\widehat\eta}(\mO)d(\bX)\rb\right|\\
 =&\Delta_n/2    
\end{align*}
and similarly, $\{P-\mathbb P_2\}\ell^{\mathrm{linear}}(\widehat d_n;\widehat\eta)\leq \Delta_n/2$. This proves  M.\eqref{eq:generic-second-stage-rate}.

The preceding empirical process argument holds simultaneously for every
$d_n^\circ\in\mathcal D_n$. Hence, by \eqref{eq:generic-orthogonal-sieve-bound},
for every $d_n^\circ\in\mathcal D_n$,
\begin{align*}
\mathcal L(\widehat d_n;\eta_0)
-\mathcal L(d_n^\circ;\eta_0)
\leq
\Delta_n(\mathcal D_n;\widehat\eta)+\rho_n
+\|\widehat e-e_0\|_2\|\widehat m-m_0\|_2.
\end{align*}
Adding and subtracting $\mathcal L(d_n^\circ;\eta_0)$ therefore gives
\begin{align*}
\mathcal L(\widehat d_n;\eta_0)-\mathcal L(d^\ast;\eta_0)
\leq\Delta_n(\mathcal D_n;\widehat\eta)+\rho_n
+\|\widehat e-e_0\|_2\|\widehat m-m_0\|_2 +\mathcal L(d_n^\circ;\eta_0)-\mathcal L(d^\ast;\eta_0).
\end{align*}
Taking the infimum over $d_n^\circ\in\mathcal D_n$ and using the
definition of $a_n$ yields
\begin{align}\label{eq:excess-risk-double-residual}
\mathcal L(\widehat d_n;\eta_0)-\mathcal L(d^\ast;\eta_0)
\leq
a_n+\Delta_n(\mathcal D_n;\widehat\eta)+\rho_n
+\|\widehat e-e_0\|_2\|\widehat m-m_0\|_2.
\end{align}


It remains to translate excess risk into value regret. By the calculation in the proof of Theorem \ref{thm:valid-loss}, for every binary rule $d$,
\begin{align*}
 \calL(d;\eta_0)-\calL(d^\ast;\eta_0)
 =&\bE\lbb2\pi(1,\bX)\pi(-1,\bX)|\tau(\bX)|
 \mone\{d(\bX)\neq d^\ast(\bX)\}\rbb\\
 &\geq2\epsilon_n(1-\epsilon_n)\{V(d^\ast)-V(d)\},
\end{align*}
where the last inequality follows from the fact that the minimum of $2x(1-x)$ over $x\in[\epsilon_n,1-\epsilon_n]$ is $2\epsilon_n(1-\epsilon_n)$. Combining this display with \eqref{eq:excess-risk-double-residual} yields the value function regret bound.
\end{proof}

\subsection{Proof of Theorem \ref{thm:vc-sieve-regret}}
\begin{proof}
Because $|Y|\leq M$ almost surely, $\|\widehat{m}\|_{\infty}\leq M$, and $\|\widehat{e}\|_\infty\leq 1$, we obtain 
\begin{align}\label{eq:bound-Zeta-4M}
 |Z_{\widehat\eta}|
 =|A-\widehat e(\bX)|\,|Y-\widehat m(\bX)|
 \leq (1+1)\times(M+M)=4M.
\end{align}
Consider the following class of functions:
\begin{align*}
 \mathcal F_n
 =&\{f_d(\mO)=Z_{\widehat\eta}(\mO)d(\bX):d\in\calD_n\},\\
 \mathcal F_n^{\pm}
 =&\mathcal F_n\cup(-\mathcal F_n)\cup\{0\}.
\end{align*}
Conditional on $\mathcal{F}_1$, the nuisance predictions $\widehat{\eta}$ and the sieve $\mathcal{D}_n$ are fixed, implying that $\mathcal F_n$ and $\mathcal F_n^{\pm}$ are nonrandom classes of functions.

Let $\{\xi_i:i\in\mathcal F_2\}$ be independent Rademacher random variables, obtained from an i.i.d. sequence $\{\xi_i\}_{i=1}^n$ that is independent of the data and satisfies $\Pr(\xi_i=1)=\Pr(\xi_i=-1)=1/2$. The enlargement by signs and the zero function is useful. To see this, we let 
\begin{equation*}
T(f):=\frac{1}{n_2}\sum_{i\in\mathcal F_2}\xi_i f(\mO_i)  
\end{equation*}
and obtain
\begin{align}
\sup_{f\in\mathcal F_n^{\pm}}
 \frac{1}{n_2}\sum_{i\in\mathcal F_2}\xi_i f(\mO_i)
 =\sup_{f\in\mathcal F_n^{\pm}}T(f)
=&
\max\lbb
    \sup_{d\in\mathcal D_n}T(f_d),\,
    \sup_{d\in\mathcal D_n}\bigl\{-T(f_d)\bigr\},\,
    0
\rbb
\notag\\
=&
\sup_{d\in\mathcal D_n}
\max\left\{T(f_d),-T(f_d),0\right\}
\notag\\
=&
\sup_{d\in\mathcal D_n}\lvert T(f_d)\rvert
\notag\\
=&
\sup_{d\in\calD_n}
 \left|\frac{1}{n_2}\sum_{i\in\mathcal F_2}
 \xi_i f_d(\mO_i)\right|.
\label{eq:symmetrization-centered}
\end{align}
It also avoids a degeneracy when $\calD_n$ is a singleton and $v_n=0$.

Fix an arbitrary probability measure \(Q\) on the observation space
\(\mathcal O\), and let \(Q_{\bX}\) denote its \(\bX\)-marginal.
For each \(d\in\mathcal D_n\), define $C_d:=\{\bx\in\mathcal X:d(\bx)=1\}\subseteq\mathcal X$ and $
\mathcal C_n:=\{C_d:d\in\mathcal D_n\}$. The class \(\mathcal C_n\) has VC dimension \(v_n\). Under the
VC-index convention of \citet{vanderVaartWellner2023}, $V(\mathcal C_n)=v_n+1
\leq 2\bar v_n$, where $\bar v_n:=v_n\vee1$. Define $\mone_{C_d}(\bx)
:=
\mone(\bx\in C_d)
=
\mone\{d(\bx)=1\}$. Moreover, because $d(\bx)=2\mone_{C_d}(\bx)-1$,
we have
\begin{align}
\|d-d'\|_{L^2({Q_{\bX}})}
=&
2\|\mone_{C_d}-\mone_{C_{d'}}\|_{L^2({Q_{\bX}})}.
\label{eq:pseudometric-relationship}
\end{align}
Let $N\{u,\mathcal D_n,L^2({Q_{\bX}})\}$ denote the covering number of \(\mathcal D_n\) by
\(L^2({Q_{\bX}})\)-balls of radius \(u\). It follows from
\eqref{eq:pseudometric-relationship} that
\[
N\{u,\mathcal D_n,L^2({Q_{\bX}})\}
\leq
N\{u/2,\mathcal C_n,L^2({Q_{\bX}})\},
\]
because a \(u/2\)-radius cover for \(\mathcal C_n\) induces a
\(u\)-radius cover for \(\mathcal D_n\).

Therefore, Theorem~2.6.4 of
\citet{vanderVaartWellner2023}, applied with \(r=2\),
\(\varepsilon=u/2\), and the arbitrary probability measure
\(Q_{\bX}\), gives, for some universal constant \(K\),
\[
\begin{aligned}
N\{u,\mathcal D_n,L^2({Q_{\bX}})\}
&\leq
K V(\mathcal C_n)(4e)^{V(\mathcal C_n)}
\left(\frac{2}{u}\right)^{2V(\mathcal C_n)}.
\end{aligned}
\]
Replacing \(K\) by \(K\vee1\) if necessary, taking logarithms, and
using $\log V(\mathcal C_n)\leq V(\mathcal C_n)$, $
V(\mathcal C_n)\leq2\bar v_n$, and $
\bar v_n\geq1$, we define $C_D
:=
4\vee
\left\{
\log K+2+2\log(4e)+4\log2
\right\}$ and obtain, for \(0<u<1\),
\begin{align}
\log N\{u,\mathcal D_n,L^2({Q_{\bX}})\}
&\leq
C_D\bar v_n\{1+\log(1/u)\}.
\label{eq:covering-number-Dn}
\end{align}

Next, by $f_d(\mO)=Z_{\widehat\eta}(\mO)d(\bX)$ and \eqref{eq:bound-Zeta-4M}, we obtain
\[
\begin{aligned}
\|f_d-f_{d'}\|_{L^2(Q)}^2
=&
\int
Z_{\widehat\eta}(o)^2
\{d(x)-d'(x)\}^2\,dQ(o)
\\
&\leq
16M^2
\int
\{d(x)-d'(x)\}^2\,dQ(o)
\\
=&
16M^2
\|d-d'\|_{L^2({Q_{\bX}})}^2.
\end{aligned}
\]
Consequently, $N\{4Mu,\mathcal F_n,L^2(Q)\}
\leq
N\{u,\mathcal D_n,L^2({Q_{\bX}})\}$. Moreover, because $\mathcal F_n^\pm
=
\mathcal F_n\cup(-\mathcal F_n)\cup\{0\}$, we have
\[
\begin{aligned}
N\{4Mu,\mathcal F_n^\pm,L^2(Q)\}
&\leq
2N\{4Mu,\mathcal F_n,L^2(Q)\}+1
\\
&\leq
2N\{u,\mathcal D_n,L^2({Q_{\bX}})\}+1
\\
&\leq
3N\{u,\mathcal D_n,L^2({Q_{\bX}})\}.
\end{aligned}
\]
Taking logarithms and applying
\eqref{eq:covering-number-Dn}, we obtain, for \(0<u<1\),
\begin{align}
\log N\{4Mu,\mathcal F_n^\pm,L^2(Q)\}
&\leq
\log3+
C_D\bar v_n\{1+\log(1/u)\}
\notag\\
&\leq
C_1\bar v_n\{1+\log(1/u)\},
\label{eq:weighted-vc-entropy-proof}
\end{align}
where $C_1:=C_D+\log3$. Indeed, $\bar v_n\{1+\log(1/u)\}\geq1$, so the term \(\log3\) is absorbed by \(C_1\). Because \(Q\) was arbitrary,
\eqref{eq:weighted-vc-entropy-proof} holds for every probability
measure \(Q\) on \(\mathcal O\).

Let
\[
Q_{2,n}
:=
\frac{1}{n_2}
\sum_{i\in\mathcal F_2}\delta_{\mO_i}
\]
denote the empirical measure of the second-split observations, and
define the conditionally centered Rademacher process by
\[
\mathbb R_{2,n}^{\xi}f
:=
\frac{1}{\sqrt{n_2}}
\sum_{i\in\mathcal F_2}\xi_i f(\mO_i),
\qquad
f\in\mathcal F_n^{\pm}.
\]
Conditional on \(\{\mO_i:i\in\mathcal F_2\}\), for every
\(f,g\in\mathcal F_n^{\pm}\) and every \(\lambda\in\mathbb R\),
\begin{align*}
&\bE_\xi\left[
 \exp\left\{
 \lambda\bigl(
 \mathbb R_{2,n}^{\xi}f
 -
 \mathbb R_{2,n}^{\xi}g
 \bigr)
 \right\}
 \biggm|
 \{\mO_i:i\in\mathcal F_2\}
 \right]
\\
=&
\prod_{i\in\mathcal F_2}
\cosh\left\{
\frac{\lambda
\{f(\mO_i)-g(\mO_i)\}}
{\sqrt{n_2}}
\right\}
\\
\leq&
\exp\left\{
\frac{\lambda^2}{2n_2}
\sum_{i\in\mathcal F_2}
\{f(\mO_i)-g(\mO_i)\}^2
\right\}
\\
=&
\exp\left\{
\frac{\lambda^2}{2}
\|f-g\|_{L^2(Q_{2,n})}^2
\right\},
\end{align*}
where the inequality uses
\(\cosh(x)\leq\exp(x^2/2)\) and
\[
\|f-g\|_{L^2(Q_{2,n})}
:=
\left[
\frac{1}{n_2}
\sum_{i\in\mathcal F_2}
\{f(\mO_i)-g(\mO_i)\}^2
\right]^{1/2}.
\]
Thus, conditional on the second-split observations,
\(\{\mathbb R_{2,n}^{\xi}f:f\in\mathcal F_n^{\pm}\}\)
is a sub-Gaussian process under the semimetric
\[
\rho_{2,n}(f,g)
:=
\|f-g\|_{L^2(Q_{2,n})}.
\]
Therefore, because
\[
\rho_{2,n}(f,0)
=
\|f\|_{L^2(Q_{2,n})}
\leq4M
\]
for every \(f\in\mathcal F_n^\pm\), we obtain
\[
\sup_{f\in\mathcal F_n^\pm}
|\mathbb R_{2,n}^{\xi}f|
=
\sup_{f\in\mathcal F_n^\pm}
|\mathbb R_{2,n}^{\xi}f-\mathbb R_{2,n}^{\xi}0|
\leq
\sup_{\rho_{2,n}(f,g)\leq4M}
|\mathbb R_{2,n}^{\xi}f-\mathbb R_{2,n}^{\xi}g|.
\]
Corollary~2.2.9 of
\citet{vanderVaartWellner2023}, applied conditionally with
\(T=\mathcal F_n^\pm\), \(d=\rho_{2,n}\), and \(\delta=4M\),
gives
\begin{align*}
&\bE_\xi\left[
\sup_{f\in\mathcal F_n^\pm}
|\mathbb R_{2,n}^{\xi}f|
\biggm|
\{\mO_i:i\in\mathcal F_2\}
\right]
\\
\leq&
\bE_\xi\left[
\sup_{\rho_{2,n}(f,g)\leq4M}
|\mathbb R_{2,n}^{\xi}f-\mathbb R_{2,n}^{\xi}g|
\biggm|
\{\mO_i:i\in\mathcal F_2\}
\right]
\\
\leq&
C_F\int_0^{4M}
\lb
\log D(
\varepsilon,\mathcal F_n^\pm,\rho_{2,n}
)
\rb^{1/2}
\,d\varepsilon
\\
\leq&
C_F\int_0^{4M}
\lb
\log N(
\varepsilon/2,\mathcal F_n^\pm,\rho_{2,n}
)
\rb^{1/2}
\,d\varepsilon,
\end{align*}
where \(C_F\) is a universal constant, \(D\) denotes the packing
number, and the last inequality uses the standard packing-covering
relation
\[
D(
\varepsilon,\mathcal F_n^\pm,\rho_{2,n}
)
\leq
N(
\varepsilon/2,\mathcal F_n^\pm,\rho_{2,n}
).
\]

Because \eqref{eq:weighted-vc-entropy-proof} holds for every
probability measure \(Q\), it can be applied, conditional on the
second-split observations, with \(Q=Q_{2,n}\). Moreover, by the
definition of \(\rho_{2,n}\),
\[
N\{
r,\mathcal F_n^\pm,\rho_{2,n}
\}
=
N\{
r,\mathcal F_n^\pm,L^2(Q_{2,n})
\}
\]
for every \(r>0\). Therefore, taking
\[
u=\frac{\varepsilon}{8M},
\qquad
4Mu=\frac{\varepsilon}{2},
\]
\eqref{eq:weighted-vc-entropy-proof} gives, for
\(0<\varepsilon<4M\),
\begin{align}
\log N\{
\varepsilon/2,\mathcal F_n^\pm,\rho_{2,n}
\}
=&
\log N\{
\varepsilon/2,\mathcal F_n^\pm,L^2(Q_{2,n})
\}
\notag\\
\leq&
C_1\bar v_n
\left\{
1+\log\left(\frac{8M}{\varepsilon}\right)
\right\}.\label{eq:bound-final-covering-number}
\end{align}

Using \eqref{eq:symmetrization-centered}, we obtain
\begin{align*}
&\bE_\xi\left[
 \sup_{d\in\calD_n}
 \left|
 \frac{1}{n_2}
 \sum_{i\in\mathcal F_2}
 \xi_i f_d(\mO_i)
 \right|
 \biggm|
 \{\mO_i:i\in\mathcal F_2\}
 \right]
\\
=&
\frac{1}{\sqrt{n_2}}
\bE_\xi\left[
 \sup_{f\in\mathcal F_n^{\pm}}
 \mathbb R_{2,n}^{\xi}f
 \biggm|
 \{\mO_i:i\in\mathcal F_2\}
 \right]
\\
\leq&
\frac{C_F}{\sqrt{n_2}}
\int_0^{4M}
\lb
\log N(
\varepsilon/2,
\mathcal F_n^{\pm},
\rho_{2,n}
)
\rb^{1/2}
\,d\varepsilon
\\
\leq&
8C_F\sqrt{C_1}M\sqrt{\frac{\bar v_n}{n_2}}
\int_0^1
\{1+\log(1/u)\}^{1/2}
\,du
\\
\leq&
C_2M\sqrt{\frac{\bar v_n}{n_2}},
\end{align*}
where $C_2:=8C_F\sqrt{C_1}\int_0^1
\{1+\log(1/u)\}^{1/2}
\,du$ and the penultimate inequality follows from
\eqref{eq:bound-final-covering-number} with
an extension of the upper integration limit from \(1/2\) to \(1\). Finally, symmetrization implies 
\begin{equation}\label{eq:symmetrization-bound}
    \bE\{
\Delta_n(\mathcal D_n;\widehat\eta)
\mid\mathcal F_1
\}
\leq
2 C_2M\sqrt{\frac{\bar v_n}{n_2}}.
\end{equation}

Changing one observation in the second split changes the supremum
defining \(\Delta_n\) by at most \(8M/n_2\). Indeed, let
\(\mathbb P_2^{(j)}\) denote the empirical measure obtained by
replacing \(\mO_j\) by \(\mO_j'\), and let $\Delta_n^{(j)}
:=
\sup_{d\in\mathcal D_n}
\left|
(\mathbb P_2^{(j)}-P)f_d
\right|$. Then
\[
\begin{aligned}
\left|\Delta_n-\Delta_n^{(j)}\right|
\leq&
\sup_{d\in\mathcal D_n}
\left|
(\mathbb P_2-\mathbb P_2^{(j)})f_d
\right|
\\
=&
\frac1{n_2}
\sup_{d\in\mathcal D_n}
\left|
f_d(\mO_j)-f_d(\mO_j')
\right|
\\
\leq&
\frac{8M}{n_2},
\end{aligned}
\]
where the last inequality follows from
\eqref{eq:bound-Zeta-4M}. Therefore, the conditional McDiarmid's inequality \citep{mcdiarmid1989method} gives
\begin{align*}
\Pr\left[
\left|
\Delta_n(\mathcal D_n;\widehat\eta)
-
\bE\{
\Delta_n(\mathcal D_n;\widehat\eta)
\mid\mathcal F_1
\}
\right|
\geq t
\biggm|
\mathcal F_1
\right]\leq
2\exp\left\{
-\frac{2t^2}
{n_2(8M/n_2)^2}
\right\}
=
2\exp\left\{
-\frac{n_2t^2}{32M^2}
\right\}.
\end{align*}

Taking $t
=
4\sqrt{2}\,M
\sqrt{{\log(2/\delta)}/{n_2}}$
and combining the above inequality with \eqref{eq:symmetrization-bound}, we obtain, with conditional probability at least \(1-\delta\),
\begin{align}\label{eq:vc-empirical-process-proof}
 \Delta_n(\calD_n;\widehat\eta)
 \leq&2 C_2M\sqrt{\frac{\bar v_n}{n_2}}+
4\sqrt{2}\,M
\sqrt{\frac{\log(2/\delta)}{n_2}}\notag\\
 \leq&
 CM
 \sqrt{
 \frac{\bar v_n+\log(2/\delta)}{n_2}
 }
 =
 CMq_{n,\mathrm{VC}}(\delta),
\end{align}
where $C/\sqrt{2}=(2C_2)\vee (4\sqrt{2})$ and the second inequality uses
\(\sqrt a+\sqrt b\leq\sqrt{2(a+b)}\). Substituting \eqref{eq:vc-empirical-process-proof} into
Theorem \ref{thm:sieve-regret} proves M.\eqref{eq:vc-sieve-regret}.

We next prove the variance adaptive refinement. Define $\widehat s_{2,n}^2
:=\mathbb P_2 Z_{\widehat\eta}^2$.
If $\widehat s_{2,n}=0$, then
$Z_{\widehat\eta,i}=0$ for every $i\in\mathcal F_2$. Hence every
function in $\mathcal F_n^{\pm}$ vanishes on the second-stage sample,
and the corresponding conditional Rademacher process is identically
zero. Otherwise, define a weighted empirical probability measure on the
measurable space \((\mathcal X,\mathscr X)\) by, for $B\in\mathscr X$,
\begin{align*}
Q_{Z,2,n}(B):=
\frac{
\mathbb P_2\!\left\{
Z_{\widehat\eta}^2\mone(\bX\in B)
\right\}
}{
\widehat s_{2,n}^2
}.
\end{align*}
Equivalently, for every measurable \(Q_{Z,2,n}\)-integrable function
\(h\), write
\begin{align*}
Q_{Z,2,n}h
:=
\int h\,dQ_{Z,2,n}
=
\frac{\mathbb P_2\{Z_{\widehat\eta}^2h(\bX)\}}
{\widehat s_{2,n}^2}.
\end{align*}
Observe that for $d,d'\in\calD_n$, $\|f_d-f_{d'}\|_{L^2(Q_{2,n})}
=
\widehat s_{2,n}
\|d-d'\|_{L^2(Q_{Z,2,n})}$.
Moreover, $\|f\|_{L^2(Q_{2,n})}\leq\widehat s_{2,n}$ for every $f\in\mathcal F_n^{\pm}$. Applying \eqref{eq:covering-number-Dn} with $Q_{\bX}=Q_{Z,2,n}$ and repeating the sign enlargement argument in \eqref{eq:weighted-vc-entropy-proof} gives, for $0<u<1$,
\begin{align}
\log N\{\widehat s_{2,n}u,\mathcal F_n^{\pm},L^2(Q_{2,n})\}
\leq
C_1\bar v_n\{1+\log(1/u)\}.
\label{eq:variance-vc-entropy-proof}
\end{align}
Repeating the conditional chaining calculation above with radius $\widehat s_{2,n}$ and using \eqref{eq:variance-vc-entropy-proof}, we obtain a universal constant $C_3$ such that
\begin{align*}
&\bE_\xi\left[
 \sup_{d\in\calD_n}
 \left|
 \frac{1}{n_2}\sum_{i\in\mathcal F_2}
 \xi_i f_d(\mO_i)
 \right|
 \biggm|
 \{\mO_i:i\in\mathcal F_2\}
 \right]
\leq
C_3\widehat s_{2,n}
\sqrt{\frac{\bar v_n}{n_2}}.
\end{align*}
Similarly, conditional symmetrization and Jensen's inequality therefore yield
\begin{align}
\bE\{\Delta_n(\calD_n;\widehat\eta)\mid\mathcal F_1\}
\leq&
2C_3\bE(\widehat s_{2,n}\mid\mathcal F_1)
\sqrt{\frac{\bar v_n}{n_2}}
\notag\\
\leq&
2C_3s_n(\widehat\eta)
\sqrt{\frac{\bar v_n}{n_2}},
\label{eq:variance-sensitive-mean-bound}
\end{align}
where $s_n^2(\widehat\eta)=\bE(Z_{\widehat\eta}^2)$.

To convert the bound in \eqref{eq:variance-sensitive-mean-bound} into a high probability bound, define $\mathcal G_n^0
:=
\{
 f_d-Pf_d,Pf_d-f_d:
 d\in\calD_n
\}$ and observe $\Delta_n(\calD_n;\widehat\eta)
=
\sup_{g\in\mathcal G_n^0}\mathbb P_2 g$. By \eqref{eq:bound-Zeta-4M}, every $g\in\mathcal G_n^0$ is bounded above by $8M$. By the fact that $d^2=1$, 
\begin{align*}
\sup_{g\in\mathcal G_n^0}Pg^2
=
\sup_{d\in\calD_n}
P\left[
\{f_d-Pf_d\}^2
\right]
\leq
\sup_{d\in\calD_n}Pf_d^2
=
PZ_{\widehat\eta}^2
=s_n^2(\widehat\eta).
\end{align*}
Applying Theorem~2.3 of \citet{bousquet2002bennett}
conditionally on $\mathcal{F}_1$ to the rescaled centered
class \((8M)^{-1}\mathcal G_n^0\), for every \(t>0\), with conditional
probability at least \(1-e^{-t}\),
\begin{align*}
\Delta_n(\calD_n;\widehat\eta)
\leq
\bE\{\Delta_n(\calD_n;\widehat\eta)\mid\mathcal F_1\}+
\left[
\frac{2t}{n_2}
\left\{
 s_n^2(\widehat\eta)
 +16M\bE\{\Delta_n(\calD_n;\widehat\eta)\mid\mathcal F_1\}
\right\}
\right]^{1/2}
+
\frac{8Mt}{3n_2}.
\end{align*}
Set $t=\log(1/\delta)$ and use \eqref{eq:variance-sensitive-mean-bound}. Since $\sqrt{\bar v_n/n_2}\leq q_{n,\mathrm{VC}}(\delta)$ and $\sqrt{t/n_2}\leq q_{n,\mathrm{VC}}(\delta)$, the inequalities $2\sqrt{ab}\leq a+b$ and $\sqrt{a+b}\leq \sqrt{a}+\sqrt{b}$ gives, for some sufficiently large universal constant $C$,
\begin{align}
\Delta_n(\calD_n;\widehat\eta)
\leq
C\left\{
 s_n(\widehat\eta)q_{n,\mathrm{VC}}(\delta)
 +Mq_{n,\mathrm{VC}}^2(\delta)
\right\}.
\label{eq:vc-empirical-process-variance-proof}
\end{align}

It remains to bound $s_n^2(\widehat\eta)$. Conditional on $\mathcal F_1$, the fitted nuisance functions are fixed. Using $|Y-\widehat m(\bX)|\leq2M$ and the conditional mean $e_0(\bX)=\bE(A\mid\bX)$, we obtain
\begin{align*}
s_n^2(\widehat\eta)
=&
\bE\left[
\{A-\widehat e(\bX)\}^2
\{Y-\widehat m(\bX)\}^2
\right]
\\
\leq&
4M^2\bE\{A-\widehat e(\bX)\}^2
\\
=&
4M^2\bE\left[
\operatorname{Var}(A\mid\bX)
+\{e_0(\bX)-\widehat e(\bX)\}^2
\right]
\\
=&
4M^2\left[
4\bE\{\pi(1,\bX)\pi(-1,\bX)\}
+\|\widehat e-e_0\|_2^2
\right]
\\
=&
4M^2\left\{
4\bar\omega_n
+\|\widehat e-e_0\|_2^2
\right\}.
\end{align*}
Finally, combining the above with \eqref{eq:vc-empirical-process-variance-proof} and Theorem~\ref{thm:sieve-regret} proves the variance-adaptive regret bound.

\end{proof}

\subsection{Proof of the decision tree VC bound M.\eqref{eq:tree-vc-bound}}\label{ss:tree-VC-dimension}
\begin{proof}
We first introduce the Hamming covering numbers used below. For a
binary rule class
\(\calD\subseteq\{-1,1\}^{\calX}\), a positive integer \(m\), and a
finite covariate sample
\(x_{1:m}=(x_1,\ldots,x_m)\in\calX^m\), define the normalized empirical
Hamming pseudometric by
\begin{align*}
 H_{x_{1:m}}(d,d')
 :=
 \frac{1}{m}\sum_{i=1}^m
 \mone\{d(x_i)\neq d'(x_i)\},
 \qquad d,d'\in\calD.
\end{align*}
It is a pseudometric rather than necessarily a metric because two distinct rules may agree at all sample points, so that $H_{x_{1:m}}(d,d')=0$ may hold even when $d\neq d'$. The sample specific Hamming covering number is
\begin{align*}
 &N_H(\varepsilon,\calD;x_{1:m})\\
 :=&
 \min\Bigl\{
 |\mathcal V|:\,
 \mathcal V\subseteq\calD,\ 
 \text{for every }d\in\calD
 \text{ there exists }\widetilde d\in\mathcal V
 \text{ such that }
 H_{x_{1:m}}(d,\widetilde d)\leq\varepsilon
 \Bigr\}.
\end{align*}
The corresponding uniform Hamming covering number is
\begin{align}\label{eq:uniform-hamming-definition}
 N_H(\varepsilon,\calD)
 :=
 \sup_{m\geq1}\ 
 \sup_{x_{1:m}\in\calX^m}
 N_H(\varepsilon,\calD;x_{1:m}).
\end{align}
Thus, ``uniform'' means that the required number of covering rules is
bounded uniformly over every finite sample size and every configuration
of covariate values. It does not require one fixed collection of
covering rules to work for all samples: the covering set
\(\mathcal V\) may depend on \(x_{1:m}\), but its cardinality is
controlled uniformly.

We now apply the decision tree covering number bound of
\citet[Lemma~4]{zhou2023offline}. Relabeling the terminal treatments
\(\{-1,1\}\) as \(\{1,2\}\) does not change any Hamming distance.
Moreover, Zhou et al.\ count the total number of tree layers, whereas
our convention assigns depth \(0\) to a single leaf. Hence, our
depth-\(L\) class corresponds to their class \(\Pi_{L+1}\). Trees of
depth strictly less than \(L\) can be represented by padding them with
redundant splits whose descendant leaves have the same treatment label. Consequently, Lemma~4, specialized to two
terminal labels with $d=2$, gives, for every \(0<\varepsilon<1\),
\begin{align}\label{eq:zhou-tree-covering}
 \log N_H(\varepsilon,\calD_{L,p}^{\mathrm{tree}})
 \leq&
 (2^{L+1}-1)\log p
 +2^{L+1}\log 2+
 (2^{L+1}-1)
 \log\left(
 \frac{L+1}{\varepsilon}+1
 \right).
\end{align}

It remains to relate the uniform Hamming covering number to VC
dimension. Let \(v\geq 1\), and suppose that
\(\calD\) shatters the points \(x_1,\ldots,x_v\). Define the
restriction map $\rho_{x_{1:v}}:\calD\longrightarrow\{-1,1\}^{v}$, where $\rho_{x_{1:v}}(d)
 :=
 \bigl(d(x_1),\ldots,d(x_v)\bigr)$. Because \(x_1,\ldots,x_v\) are shattered, every assignment of binary
labels to these points is realized by some rule in \(\calD\). Hence
\begin{align}\label{eq:shattered-cube}
 \rho_{x_{1:v}}(\calD)=\{-1,1\}^{v}.
\end{align}
Thus, when restricted to the shattered sample, the rule class is the
entire binary cube $\{-1,1\}^{v}$, containing \(2^v\) distinct label vectors.

For \(a,b\in\{-1,1\}^{v}\), write
\begin{align*}
 H_v(a,b)
 :=
 \frac{1}{v}\sum_{i=1}^{v}\mone\{a_i\neq b_i\}
\end{align*}
for their normalized Hamming distance. Fix a center
\(a\in\{-1,1\}^{v}\). A vector \(b\) lies in the Hamming ball of
radius \(1/4\) around \(a\) precisely when $H_v(a,b)\leq 1/4$, or, equivalently, when \(a\) and \(b\) differ in at most
\(\lfloor v/4\rfloor\) coordinates.

For each \(j\), there are exactly \(\binom{v}{j}\) vectors that differ
from \(a\) in exactly \(j\) coordinates: one first chooses the \(j\)
coordinates on which the vectors differ, and, because the labels are
binary, the values of \(b\) on those coordinates are then uniquely
determined as the opposites of the corresponding values of \(a\).
Consequently, every normalized Hamming ball of radius \(1/4\) in
\(\{-1,1\}^{v}\) contains exactly $B_v
 :=
 \sum_{j=0}^{\lfloor v/4\rfloor}\binom{v}{j}$ vertices.

We next bound \(B_v\). Let \(q=1/4\), and define the binary entropy
function, using natural logarithms, by $ h(q):=-q\log q-(1-q)\log(1-q)$. For \(0\leq j\leq qv\), the quantity
\(q^j(1-q)^{v-j}\) is decreasing in \(j\), since
\(q/(1-q)<1\). Therefore, 
\begin{equation*}
q^j(1-q)^{v-j}\geq
 q^{qv}(1-q)^{(1-q)v}
 =
 \exp\{-v h(q)\}.    
\end{equation*}
Applying this observation to the binomial expansion
\begin{equation*}
1
 =
 \{q+(1-q)\}^{v}
 =
 \sum_{j=0}^{v}
 \binom{v}{j}q^j(1-q)^{v-j},    
\end{equation*}
we obtain
\begin{align*}
 1\geq
 \sum_{j=0}^{\lfloor qv\rfloor}
 \binom{v}{j}q^j(1-q)^{v-j}\geq
 \exp\{-v h(q)\}
 \sum_{j=0}^{\lfloor qv\rfloor}\binom{v}{j}.
\end{align*}
Setting \(q=1/4\) and rearranging gives the standard entropy bound
\begin{align}\label{eq:hamming-ball-entropy}
 B_v
 =
 \sum_{j=0}^{\lfloor v/4\rfloor}\binom{v}{j}
 \leq
 \exp\{v h(1/4)\}.
\end{align}

Now let \(\mathcal V\subseteq\calD\) be any \(1/4\)-Hamming cover of
\(\calD\) on the sample \(x_{1:v}\). By
\eqref{eq:shattered-cube}, every vector
\(a\in\{-1,1\}^{v}\) corresponds to some rule \(d\in\calD\) such that $\rho_{x_{1:v}}(d)=a$.
Because \(\mathcal V\) is a cover, there must be some
\(\widetilde d\in\mathcal V\) such that 
\begin{equation*}
H_v\bigl(
 \rho_{x_{1:v}}(d),
 \rho_{x_{1:v}}(\widetilde d)
 \bigr)
 \leq 1/4.    
\end{equation*}
Hence, the union of the Hamming balls of radius \(1/4\) centered at \(\rho_{x_{1:v}}(\widetilde d)\), \(\widetilde d\in\mathcal V\), covers the entire cube \(\{-1,1\}^{v}\). Recall that each such ball contains at most \(B_v\) vertices. Even though the
balls may overlap, their union therefore contains at most
\(|\mathcal V|B_v\) vertices. Since this union must contain all
\(2^v\) vertices of the cube, $2^v
 \leq
 |\mathcal V|B_v$. Because this inequality holds for every \(1/4\)-cover
\(\mathcal V\), it follows that
\begin{align*}
 N_H(1/4,\calD;x_{1:v})
 \geq
 \frac{2^v}{B_v}.
\end{align*}
The uniform covering number is the supremum over all finite samples,
so
\begin{align*}
 N_H(1/4,\calD)\geq
 N_H(1/4,\calD;x_{1:v})\geq
 \frac{2^v}{B_v}\geq
 \exp\left[
 v\{\log 2-h(1/4)\}
 \right],
\end{align*}
where the last inequality follows from
\eqref{eq:hamming-ball-entropy}. Taking logarithms yields
\begin{align*}
 v
 \leq
 \frac{\log N_H(1/4,\calD)}
 {\log 2-h(1/4)}.
\end{align*}
Notice that the denominator is strictly positive: $ \log 2-h(1/4)
 =
 3/4\log 3-\log 2
 >0$. Since the displayed inequality holds for every integer \(v\) such
that some set of \(v\) points is shattered by \(\calD\), we conclude
that
\begin{align}\label{eq:vc-from-hamming-cover}
 \operatorname{VC}(\calD)
 \leq
 C_H\log N_H(1/4,\calD),
\end{align}
where $C_H:=
 \{\log 2-h(1/4)\}^{-1}$. The particular radius \(1/4\) is chosen only for convenience. The
same argument works with any fixed radius \(q<1/2\), with
\(\{\log 2-h(q)\}^{-1}\) replacing \(C_H\).

Combining \eqref{eq:zhou-tree-covering} at
\(\varepsilon=1/4\) with \eqref{eq:vc-from-hamming-cover} gives
\begin{align*}
 \operatorname{VC}(\calD_{L,p}^{\mathrm{tree}})
 \leq C_H\Bigl[(2^{L+1}-1)\log p
 +2^{L+1}\log 2+
 (2^{L+1}-1)\log\{4(L+1)+1\}
 \Bigr].
\end{align*}
Since \(2^{L+1}-1\leq2^{L+1}\) and
\(4(L+1)+1\leq4(L+2)\),
\begin{align*}
 \operatorname{VC}(\calD_{L,p}^{\mathrm{tree}})\leq
 C_H2^{L+1}
 \left[
 \log p+\log2+\log\{4(L+2)\}
 \right]=
 2C_H2^L\log\{8p(L+2)\}.
\end{align*}
Finally, \(2p(L+2)\geq4\) for \(p\geq1\) and \(L\geq0\), so
\begin{equation*}
\log\{8p(L+2)\}
 =
 \log\!\left[4\{2p(L+2)\}\right]
 \leq
 2\log\{2p(L+2)\}.    
\end{equation*}
Thus, with $C_{\mathrm{tree}}:=4C_H$, $\operatorname{VC}(\calD_{L,p}^{\mathrm{tree}})
 \leq
 C_{\mathrm{tree}}2^L\log\{2p(L+2)\}$, which proves M.\eqref{eq:tree-vc-bound}. Substitution into
Theorem~\ref{thm:vc-sieve-regret} gives
M.\eqref{eq:tree-sieve-regret}.
\end{proof}

\subsection{Proof of Theorem \ref{thm:surrogate-oracle}}
\begin{proof}
We condition throughout on the nuisance training split $\mathcal F_1$, so that $\widehat\eta$ and any $\mathcal F_1$-measurable score sieve $\calG_n$ are fixed. Write $Z_0=Z_{\eta_0}$ and $S_0=\operatorname{sign}(Z_0)$. 

We first prove the population calibration inequality below:
\begin{align}\label{eq:weighted-surrogate-calibration}
 \psi_\phi\left\{
 \frac{\operatorname{Reg}_{\omega}(d_g)}{\bar c_0}
 \right\}
 \leq
 \frac{\calL_\phi(g;\eta_0)-\calL_\phi^\ast}{\bar c_0}.
\end{align}
This inequality holds when $\bar c_0=\bE(|Z_0|)>0$. 
Consider the following change of measure: define a probability measure $Q$ on the observed data space by
\begin{align}\label{eq:surrogate-tilted-measure}
 \frac{dQ}{dP}(\mO)=\frac{|Z_0|}{\bar c_0}.
\end{align}
Under $Q$, with $\bX$ as the feature and $S_0\in\{-1,1\}$ as the class label, the weighted $0$--$1$ classification problem under $P$, $\calL(d_g;\eta_0)=\bE\lbb |Z_0|\mone\{d_g(\bX)\neq S_0\}\rbb$, reduces to the corresponding unweighted problem. 

For a measurable score $g$, define its misclassification probability or $0$--$1$ classification risk and population surrogate risk under $Q$ by 
\begin{equation*}
R_Q(g)=Q\{d_g(\bX)\neq S_0\}=\bE_{Q}[\mone \{d_g(\bX)\neq S_0\}]    
\end{equation*}
and $R_{\phi,Q}(g)=\bE_Q\lbb\phi\{S_0g(\bX)\}\rbb$, respectively. Let 
\begin{equation*}
  R_Q^\ast=\inf_{g:\calX\to\mathbb R}R_Q(g)\qquad  \text{and}\qquad R_{\phi,Q}^\ast=\inf_{g:\calX\to\mathbb R}R_{\phi,Q}(g).
\end{equation*}
By \eqref{eq:surrogate-tilted-measure} and the identities $\ell_\phi(g;\eta)
 =|Z_\eta|\phi\{S_\eta g(\bX)\}$ and $\calL_\phi(g;\eta)=\bE\{\ell_\phi(g;\eta)\}$, it follows that 
\begin{align}\label{eq:tilted-surrogate-risk}
 R_{\phi,Q}(g)=\frac{\calL_\phi(g;\eta_0)}{\bar c_0},
 \qquad
 R_{\phi,Q}^\ast=\frac{\calL_\phi^\ast}{\bar c_0}.
\end{align}

To identify the Bayes classifier under $Q$, let $w_+(\bX)=\bE\{(Z_0)_+\mid\bX\}$ and $
 w_-(\bX)=\bE\{(-Z_0)_+\mid\bX\}$. Since $z_++(-z)_+=|z|$ and $z_+-(-z)_+=z$,
\begin{align}\label{eq:surrogate-w-identities}
 w_+(\bX)+w_-(\bX)=&c_0(\bX),\notag\\
 w_+(\bX)-w_-(\bX)
 =&\bE(Z_0\mid\bX)
 =h_0(\bX)
 =2\pi(1,\bX)\pi(-1,\bX)\tau(\bX),
\end{align}
where the last equality follows from the proof in Section
\ref{ss:proof-fisher-consistency}. On the event $\{c_0>0\}$, define $\zeta(\bX)=w_+(\bX)/c_0(\bX)$. Its value on $\{\bX:c_0(\bX)=0\}$ may be chosen arbitrarily, because
$dQ_{\bX}/dP_{\bX}=c_0/\bar c_0$ and hence
\begin{equation*}
Q_{\bX}\{c_0(\bX)=0\}
=
\bar c_0^{-1}\bE_P[
c_0(\bX)\mone\{c_0(\bX)=0\}]
=0,    
\end{equation*}
where $Q_{\bX}$ and $P_{\bX}$ are the corresponding $\bX$-marginal measures.

We then observe
\begin{align}\label{eq:tilted-label-probability}
Q(S_0=1\mid\bX)=&
\frac{
 \bE_P\!\lb
 |Z_0|\mone(S_0=1)\mid\bX
 \rb
}{
 \bE_P\!\lsb
 |Z_0|\mid\bX
 \rsb
}
=
\frac{
 \bE_P\{(Z_0)_+\mid\bX\}
}{
 c_0(\bX)
}
=
\frac{w_+(\bX)}{c_0(\bX)}
=
\zeta(\bX),\notag\\
 2\zeta(\bX)-1=&\frac{2w_{+}(\bX)-c_0(\bX)}{c_0(\bX)}
 =\frac{w_{+}(\bX)-w_{-}(\bX)}{c_0(\bX)}
 =\frac{h_0(\bX)}{c_0(\bX)}.
\end{align}
Weak positivity and \eqref{eq:surrogate-w-identities} imply that
\begin{equation*}
 \operatorname{sign}\{2\zeta(\bX)-1\}=\operatorname{sign}\{\tau(\bX)\}=d^\ast(\bX)   
\end{equation*}
on $\{\tau\neq0\}$. On $\{\tau=0\}$, either treatment label is Bayes optimal. Hence $d^\ast$ is a Bayes classifier under $Q$. For a fixed covariate value $\bX=\bx$, the conditional $0$--$1$ risk of a decision $d\in\{-1,1\}$ is
\begin{equation*}
Q\{d\neq S_0\mid \bX=\bx\}=\zeta(\bx)\mone(d=-1)+\{1-\zeta(\bx)\}\mone(d=1).    
\end{equation*}
Subtracting its minimum, 
\begin{equation*}
 Q\{\operatorname{sign}\{2\zeta(\bx)-1\}\neq S_0\mid \bX=\bx\}=\min\{\zeta(\bx),1-\zeta(\bx)\},   
\end{equation*}
gives the pointwise excess risk
\begin{equation*}
|2\zeta(\bx)-1|
 \mone\lbb d\neq\operatorname{sign}\{2\zeta(\bx)-1\}\rbb.    
\end{equation*}
Integrating this identity under $Q_{\bX}$ and using \eqref{eq:tilted-label-probability} yields
\begin{align}
 R_Q(g)-R_Q^\ast
 =&\bE_Q\lbb
 |2\zeta(\bX)-1|
 \mone\{d_g(\bX)\neq d^\ast(\bX)\}
 \rbb\notag\\
 =&\frac{1}{\bar c_0}
 \bE\lbb
 |h_0(\bX)|
 \mone\{d_g(\bX)\neq d^\ast(\bX)\}
 \rbb\notag\\
 =&\frac{\operatorname{Reg}_{\omega}(d_g)}{\bar c_0}.
 \label{eq:tilted-classification-regret}
\end{align}
The last equality follows from the definition of $\operatorname{Reg}_{\omega}$ in the main manuscript. Applying Theorem~1, part~1, of \citet{bartlett2006convexity} to the classification problem under $Q$ gives
\begin{align*}
 \psi_\phi\{R_Q(g)-R_Q^\ast\}
 \leq R_{\phi,Q}(g)-R_{\phi,Q}^\ast.
\end{align*}
Although \citet{bartlett2006convexity} define $\operatorname{sign}(0)=0$, their pointwise argument remains valid under $\operatorname{sign}(0)=1$. Indeed, when $2\zeta-1>0$, an incorrectly classified score satisfies $g<0$, and continuity of the convex loss makes the infimum over $g<0$ equal to the infimum over $g\leq0$ used in their Definition~2; when $2\zeta-1<0$, an error satisfies $g\geq0$ directly. Combining \eqref{eq:tilted-surrogate-risk} and \eqref{eq:tilted-classification-regret} proves \eqref{eq:weighted-surrogate-calibration}.

If $g_\phi^\ast$ attains $\mathcal{L}^\ast_{\phi}$, substituting $g=g_\phi^\ast$ into the right-hand side of \eqref{eq:weighted-surrogate-calibration} yields zero. By Lemma~2, part~9, of \citet{bartlett2006convexity}, classification calibration implies $\psi_\phi(t)>0$ for every $t\in(0,1]$, and therefore $\operatorname{Reg}_{\omega}(d_{g_\phi^\ast})=0$. Define the following event as the intersection of $\{d_{g_\phi^\ast}(\bX)\neq d^\ast(\bX)\}$ and $\{\tau(\bX)\neq0\}$:
\[
B
=
\left\{
d_{g_\phi^\ast}(\bX)\neq d^\ast(\bX),
\ \tau(\bX)\neq0
\right\}.
\]
Weak positivity implies $|h_0(\bX)|>0$ on $\{\tau(\bX)\neq0\}$. Therefore,
\begin{align*}
0
=&
\operatorname{Reg}_{\omega}(d_{g_\phi^\ast})\\
=&
\bE\!\left[
|h_0(\bX)|
\mone\{d_{g_\phi^\ast}(\bX)\neq d^\ast(\bX)\}
\right]\\
&\geq
\bE\!\left[
|h_0(\bX)|\mone(B)
\right],
\end{align*}
where the last inequality follows from $\mone\{d_{g_\phi^\ast}(\bX)\neq d^\ast(\bX)\}\geq \mone(B)$ as
\begin{equation*}
B=\{d_{g_\phi^\ast}(\bX)\neq d^\ast(\bX)\}\cap \{\tau(\bX)\neq 0\}\subseteq \{d_{g_\phi^\ast}(\bX)\neq d^\ast(\bX)\}.    
\end{equation*}
The last integrand is nonnegative and is strictly positive on $B$.
A nonnegative random variable has expectation zero only if it is zero
almost surely. It follows that
\[
P(B)
=
P\!\left\{
d_{g_\phi^\ast}(\bX)\neq d^\ast(\bX),
\ \tau(\bX)\neq0
\right\}
=0.
\]
Thus, the events $\{d_{g_\phi^\ast}(\bX)\neq d^\ast(\bX)\}$ and $\{\tau(\bX)\neq 0\}$ cannot hold simultaneously almost surely; in other words, it follows that 
\[
d_{g_\phi^\ast}(\bX)=d^\ast(\bX)
\qquad\text{almost surely on }
\{\tau(\bX)\neq0\}.
\]
On $\{\tau(\bX)=0\}$, the two treatments have the same conditional
mean outcome, so either treatment label is optimal. Consequently,
Fisher consistency holds on the event $\{\tau(\bX)\neq 0\}$.

We next establish the following bound on the surrogate excess risk of the empirical minimizer:
\begin{align}\label{eq:surrogate-value-regret}
 V(d^\ast)-V(d_{\widehat g_n})
 \leq\frac{\bar c_0}{2\epsilon_n(1-\epsilon_n)}
 \psi_\phi^{-1}\left(
 \frac{a_{n,\phi}+2\Delta_{n,\phi}(\calG_n;\widehat\eta)
 +\rho_n+\Gamma_{n,\phi}(\calG_n;\widehat\eta)}{\bar c_0}
 \right),
\end{align}
If $a_{n,\phi}=\infty$, \eqref{eq:surrogate-value-regret} is immediate. To see this, for every rule $d$, 
\begin{equation*}
 \operatorname{Reg}_{\omega}(d)
\leq
\bE|h_0(\bX)|
\leq
\bE\{\bE(|Z_{\eta_0}|\mid\bX)\}
=
\bar c_0.   
\end{equation*}
Together with $2\epsilon_n(1-\epsilon_n)
\{V(d^\ast)-V(d)\}
\leq
\operatorname{Reg}_{\omega}(d)$, this proves \eqref{eq:surrogate-value-regret}. 

Otherwise, fix $\varepsilon>0$ and choose $g_{n,\varepsilon}^\circ\in\calG_n$ such that
\begin{align}\label{eq:surrogate-epsilon-comparator}
 \calL_\phi(g_{n,\varepsilon}^\circ;\eta_0)-\calL_\phi^\ast
 \leq a_{n,\phi}+\varepsilon.
\end{align}
Recall 
\begin{align}
 \Gamma_{n,\phi}(\calG_n;\widehat\eta)
 &:=\sup_{g,g'\in\calG_n}
 \left|\{\calL_\phi(g;\eta_0)-\calL_\phi(g';\eta_0)\}
 -\{\calL_\phi(g;\widehat\eta)-\calL_\phi(g';\widehat\eta)\}\right|.
 \label{eq:surrogate-nuisance-sensitivity}
\end{align}
By \eqref{eq:surrogate-nuisance-sensitivity},
\begin{align}
 &\calL_\phi(\widehat g_n;\eta_0)
 -\calL_\phi(g_{n,\varepsilon}^\circ;\eta_0)\notag\\
 =&
 \calL_\phi(\widehat g_n;\widehat\eta)
 -\calL_\phi(g_{n,\varepsilon}^\circ;\widehat\eta)
 +[\calL_\phi(\widehat g_n;\eta_0)
 -\calL_\phi(g_{n,\varepsilon}^\circ;\eta_0)-\{\calL_\phi(\widehat g_n;\widehat\eta)
 -\calL_\phi(g_{n,\varepsilon}^\circ;\widehat\eta)\}]\notag\\
 \leq&
 \calL_\phi(\widehat g_n;\widehat\eta)
 -\calL_\phi(g_{n,\varepsilon}^\circ;\widehat\eta)
 +\Gamma_{n,\phi}(\calG_n;\widehat\eta).
 \label{eq:surrogate-change-nuisance}
\end{align}
For the first difference on the right-hand side, add and subtract the two empirical risks:
\begin{align*}
 &\calL_\phi(\widehat g_n;\widehat\eta)
 -\calL_\phi(g_{n,\varepsilon}^\circ;\widehat\eta)\\
 =&\{P-\mathbb P_2\}\ell_\phi(\widehat g_n;\widehat\eta)
 +\mathbb P_2\{\ell_\phi(\widehat g_n;\widehat\eta)
 -\ell_\phi(g_{n,\varepsilon}^\circ;\widehat\eta)\}+\{\mathbb P_2-P\}\ell_\phi(g_{n,\varepsilon}^\circ;\widehat\eta)\\
 \leq&2\Delta_{n,\phi}(\calG_n;\widehat\eta)+\rho_n,
\end{align*}
where the last line follows from M.\eqref{eq:surrogate-erm} and \eqref{eq:surrogate-nuisance-sensitivity} using arguments similar to those used in Section \ref{ss:proof-of-theorem{thm:sieve-regret}}. Combining this inequality with \eqref{eq:surrogate-epsilon-comparator} and \eqref{eq:surrogate-change-nuisance} gives
\begin{align}\label{eq:surrogate-excess-proof}
 \calL_\phi(\widehat g_n;\eta_0)-\calL_\phi^\ast
 \leq a_{n,\phi}+\varepsilon
 +2\Delta_{n,\phi}(\calG_n;\widehat\eta)+\rho_n
 +\Gamma_{n,\phi}(\calG_n;\widehat\eta).
\end{align}

Apply \eqref{eq:weighted-surrogate-calibration} with
$g=\widehat g_n$, and define $r_n
=
{\operatorname{Reg}_{\omega}
(d_{\widehat g_n})}/{\bar c_0}$ and 
\[
B_{n,\varepsilon}
=
\frac{
a_{n,\phi}+\varepsilon
+2\Delta_{n,\phi}(\calG_n;\widehat\eta)
+\rho_n
+\Gamma_{n,\phi}(\calG_n;\widehat\eta)
}{\bar c_0}.
\]
Since
\begin{equation*}
 0\leq
\operatorname{Reg}_{\omega}(d)
\leq
\bE|h_0(\bX)|
\leq
\bE\{\bE(|Z_{\eta_0}|\mid\bX)\}
=
\bar c_0,   
\end{equation*}
we have $r_n\in[0,1]$. Combining
\eqref{eq:weighted-surrogate-calibration} with
\eqref{eq:surrogate-excess-proof} gives $\psi_\phi(r_n)\leq B_{n,\varepsilon}$. By the definition $\psi_\phi^{-1}(s)
=
\sup\{u\in[0,1]:\psi_\phi(u)\leq s\}$, the preceding inequality implies
$r_n\leq\psi_\phi^{-1}(B_{n,\varepsilon})$. Therefore,
\[
\operatorname{Reg}_{\omega}(d_{\widehat g_n})
\leq
\bar c_0\psi_\phi^{-1}\left(
\frac{
a_{n,\phi}+\varepsilon
+2\Delta_{n,\phi}(\calG_n;\widehat\eta)
+\rho_n
+\Gamma_{n,\phi}(\calG_n;\widehat\eta)
}{\bar c_0}
\right).
\]
Finally, \eqref{eq:surrogate-value-regret} follows by letting $\varepsilon\downarrow0$ and using $ 2\epsilon_n(1-\epsilon_n)\{V(d^\ast)-V(d)\}
 \leq\operatorname{Reg}_{\omega}(d)$.

If $\bar c_0=\bE(|Z_{\eta_0}|)=0$, nonnegativity implies $Z_0=0$ almost surely, and hence $h_0(\bX)=\bE(Z_0\mid\bX)=0$ almost surely. By \eqref{eq:surrogate-w-identities} and weak positivity, $\tau(\bX)=0$ almost surely. The value regret is therefore zero for every rule.
\end{proof}

\subsection{Proof of the nuisance sensitivity bound M.\eqref{eq:generic-surrogate-nuisance-bound}}\label{ss:proof-nuisance-surrogate-bound}
\begin{proof}
Condition on $\mathcal F_1$ throughout. For $g\in\calG_n$, $z\in\mathbb R$, and $\bx\in\calX$, define
\[
r_g(z,\bx)
=
z_+\phi\{g(\bx)\}
+
(-z)_+\phi\{-g(\bx)\}.
\]
Equivalently,
\[
r_g(z,\bx)
=
\begin{cases}
z\phi\{g(\bx)\}, & z>0,\\
0, & z=0,\\
(-z)\phi\{-g(\bx)\}, & z<0.
\end{cases}
\]
Consequently, $r_g(z,\bx)
=
|z|\phi\{\operatorname{sign}(z)g(\bx)\}$. The identity at $z=0$ is valid irrespective of the convention for
$\operatorname{sign}(0)$, because both sides are zero. Therefore,
$$r_g(Z_\eta,\bX)
=
|Z_\eta|
\phi\{\operatorname{sign}(Z_\eta)g(\bX)\}
=
\ell_\phi(g;\eta)$$
holds almost surely.

For $g,g'\in\calG_n$, define $q_{g,g'}(z,\bx)
=
r_g(z,\bx)-r_{g'}(z,\bx)$. For $z>0$, 
\begin{equation*}
q_{g,g'}(z,\bx)
=
z[
\phi\{g(\bx)\}-\phi\{g'(\bx)\}
],    
\end{equation*}
whereas for $z<0$,
\begin{equation*}
q_{g,g'}(z,\bx)
=
(-z)[
\phi\{-g(\bx)\}-\phi\{-g'(\bx)\}
].    
\end{equation*}
Thus, the slopes of $q_{g,g'}(\cdot,\bx)$ on the positive and
negative half-lines are, respectively, $\phi\{g(\bx)\}-\phi\{g'(\bx)\}$ and $\phi\{-g'(\bx)\}-\phi\{-g(\bx)\}$. Because $|g|,|g'|\leq B_n$ and
$0\leq\phi(u)\leq K_{\phi,n}$ for $|u|\leq B_n$, both slopes
have absolute value at most $K_{\phi,n}$. Moreover,
\begin{equation*}
q_{g,g'}(0,\bx)=\lim_{z\to0^{+}}q_{g,g'}(z,\bx)=\lim_{z\to0^{-}}q_{g,g'}(z,\bx)=0,    
\end{equation*}
so $q_{g,g'}(\cdot,\bx)$ is continuous
at zero. It follows that, for all $z,z'\in\mathbb R$, 
\begin{equation*}
|q_{g,g'}(z,\bx)-q_{g,g'}(z',\bx)|
\leq
K_{\phi,n}|z-z'|.    
\end{equation*}
To see this, if $z$ and $z'$ have the same sign, the Lipschitz property follows immediately. If they have
opposite signs, then $q_{g,g'}(0,\bx)=0$ and
\begin{align*}
|q_{g,g'}(z,\bx)-q_{g,g'}(z',\bx)|
\leq&
|q_{g,g'}(z,\bx)-q_{g,g'}(0,\bx)|
+
|q_{g,g'}(0,\bx)-q_{g,g'}(z',\bx)|\\
\leq&
K_{\phi,n}(|z|+|z'|)\\
=&
K_{\phi,n}|z-z'|,
\end{align*}
where the last equality follows because $zz'\leq0$. Hence, we obtain
\begin{align*}
&\left|
\{\calL_\phi(g;\eta_0)-\calL_\phi(g';\eta_0)\}
-
\{\calL_\phi(g;\widehat\eta)
  -\calL_\phi(g';\widehat\eta)\}
\right|\\
=&
\left|
\bE\lb
q_{g,g'}(Z_{\eta_0},\bX)
-
q_{g,g'}(Z_{\widehat\eta},\bX)
\rb
\right|\\
\leq&
K_{\phi,n}\bE(|Z_{\widehat\eta}-Z_{\eta_0}|).
\end{align*}
Taking the supremum over $g,g'\in\calG_n$ gives $\Gamma_{n,\phi}(\calG_n;\widehat\eta)
\leq
K_{\phi,n}\bE(|Z_{\widehat\eta}-Z_{\eta_0}|)$.

Let $\delta_e=\widehat e-e_0$ and $\delta_m=\widehat m-m_0$. Expanding the two residual products gives the exact identity
\begin{align}
 Z_{\widehat\eta}-Z_{\eta_0}
 =-\{Y-m_0(\bX)\}\delta_e(\bX)
 -\{A-e_0(\bX)\}\delta_m(\bX)
 +\delta_e(\bX)\delta_m(\bX).\notag
\end{align}
By the triangle inequality and Cauchy--Schwarz,
\begin{align}
 \bE(|Z_{\widehat\eta}-Z_{\eta_0}|)
 &\leq
 \|Y-m_0(\bX)\|_2\|\delta_e\|_2
 +\|A-e_0(\bX)\|_2\|\delta_m\|_2
 +\|\delta_e\|_2\|\delta_m\|_2.
 \label{eq:surrogate-z-l2-bound}
\end{align}
Since $A\in\{-1,1\}$ and $e_0(\bX)=\bE(A\mid\bX)$, $ \|A-e_0(\bX)\|_2^2
 =\bE\{1-e_0(\bX)^2\}\leq1$.
Substituting this bound into \eqref{eq:surrogate-z-l2-bound}, and then substituting the result into $\Gamma_{n,\phi}(\calG_n;\widehat\eta)
\leq
K_{\phi,n}\bE(|Z_{\widehat\eta}-Z_{\eta_0}|)$, proves M.\eqref{eq:generic-surrogate-nuisance-bound}.
\end{proof}

\subsection{Proof of Theorem \ref{thm:bounded-hinge-sieve}}
\begin{proof}
The proof follows arguments similar to those used in Section \ref{ss:proof-of-theorem{thm:sieve-regret}}. For $g,g'\in\calG_n$, it follows that $\calL_H(g;\eta)-\calL_H(g';\eta)
 =-P\lbb Z_\eta\{g(\bX)-g'(\bX)\}\rbb$, which is exactly the difference of the linear risks in \eqref{eq:hinge-affine-criterion}. We first verify orthogonality for this linear criterion. For an ambient target direction $v$,
\begin{align*}
 D_g\calL_H^{\mathrm{linear}}(\bar g;\eta)[v]
 =-P\{Z_\eta v(\bX)\}.
\end{align*}
Differentiating in a nuisance direction $(\dot e,\dot m)$ at $\eta_0$ gives
\begin{align*}
 D_\eta D_g\calL_H^{\mathrm{linear}}(\bar g;\eta_0)
 \lbb v,(\dot e,\dot m)\rbb=P\left[v(\bX)\{\dot e(\bX)(Y-m_0(\bX))
 +(A-e_0(\bX))\dot m(\bX)\}\right]=0,
\end{align*}
where the final equality follows from $\bE(Y-m_0\mid\bX)=0$ and $\bE(A-e_0\mid\bX)=0$. This proves the universal Neyman orthogonality. 

For $g_{n}^\circ\in\calG_n$, it follows that
\begin{align}\label{eq:hinge-epsilon-comparator}
 \calL_H(g_{n}^\circ;\eta_0)
 -\calL^\ast_{H}(\eta_0)
 \leq a_n^H.
\end{align}
Because $\mathbb P_2|Z_{\widehat\eta}|$ is common to all scores, approximate empirical hinge minimization is equivalent to approximate minimization of $\mathbb P_2\ell_H^{\mathrm{linear}}(g;\widehat\eta)$. Adding and subtracting the two population linear risks therefore yields 
\begin{equation}\label{eq:bound-approximation-error-maximal-fluctuation}
    \calL_H^{\mathrm{linear}}(\widehat g_n;\widehat\eta)
 -\calL_H^{\mathrm{linear}}(g_{n}^\circ;\widehat\eta)\leq
 2\Delta_{n,H}^{\mathrm{linear}}(\calG_n;\widehat\eta)+\rho_n.
\end{equation}

Let $\delta_e=\widehat e-e_0$ and $\delta_m=\widehat m-m_0$. Direct expansion gives
\begin{align}
 Z_{\widehat\eta}-Z_{\eta_0}
 =-\{Y-m_0(\bX)\}\delta_e(\bX)
 -\{A-e_0(\bX)\}\delta_m(\bX)
 +\delta_e(\bX)\delta_m(\bX),\notag
\end{align}
so that $ \bE(Z_{\widehat\eta}-Z_{\eta_0}\mid\bX)
 =\delta_e(\bX)\delta_m(\bX)$. Consequently,
\begin{align}\label{eq:hinge-exact-nuisance-proof}
 \{\calL_H(\widehat{g}_n;\eta_0)-\calL_H(g_n^\circ;\eta_0)\}
 -\{\calL_H(\widehat{g}_n;\widehat\eta)-\calL_H(g_n^\circ;\widehat\eta)\}=&P\left[\delta_e(\bX)\delta_m(\bX)
 \{\widehat{g}_n(\bX)-g_n^\circ(\bX)\}\right]\notag\\
 \leq&2\|\delta_e\|_2\|\delta_m\|_2,
\end{align}
where the inequality follows from the fact that both scores are bounded by one and from the Cauchy--Schwarz inequality. Combining \eqref{eq:hinge-epsilon-comparator}--\eqref{eq:hinge-exact-nuisance-proof} gives
\begin{align}\label{eq:hinge-excess-proof}
 \calL_H(\widehat g_n;\eta_0)-\calL^\ast_{H}(\eta_0)
 \leq a_n^H
 +2\Delta_{n,H}^{\mathrm{linear}}(\calG_n;\widehat\eta)+\rho_n
 +2\|\delta_e\|_2\|\delta_m\|_2.
\end{align}

Finally, let $g_H^\ast(\bX)=\operatorname{sign}\{h_0(\bX)\}$. Under weak positivity, we obtain $g_H^\ast(\bX)=d^\ast$. Moreover, for the bounded hinge loss,
$g_H^\ast$ attains the population minimum, so that $\calL_H^\ast(\eta_0)
=
\calL_H(g_H^\ast;\eta_0)$. To see this, for any measurable score $g$, let
$\bar g=(-1)\vee(g\wedge1)$. Pointwise, $\{1-S_0\bar g(\bX)\}_+
\leq
\{1-S_0g(\bX)\}_+$, so clipping to $[-1,1]$ cannot increase the hinge risk. It therefore
suffices to minimize over scores satisfying $|g|\leq1$. For every such
score,
\begin{align*}
\calL_H(g;\eta_0)
=&
P\!\left[|Z_0|\{1-S_0g(\bX)\}\right]\\
=&
P|Z_0|-P\{h_0(\bX)g(\bX)\}\\
&\geq
P|Z_0|-P|h_0(\bX)|,
\end{align*}
where the inequality follows from
$h_0(\bX)g(\bX)\leq|h_0(\bX)|$ as $|g|\leq1$. Equality holds for the measurable score $g_H^\ast(\bX)=\operatorname{sign}\{h_0(\bX)\}$, because $h_0(\bX)g_H^\ast(\bX)=|h_0(\bX)|$. Hence $\calL_H^\ast(\eta_0)
=
\calL_H(g_H^\ast;\eta_0)$. Therefore, for every $g\in[-1,1]^{\calX}$,
\begin{align*}
\calL_H(g;\eta_0)-\calL_H^\ast(\eta_0)
=&
P\left[
|h_0(\bX)|
-h_0(\bX)g(\bX)
\right]\\
=&
P\left[
|h_0(\bX)|
-|h_0(\bX)|\text{sign}\{h_0(\bX)\}g(\bX)
\right]\\
=&
P\left[
|h_0(\bX)|
-|h_0(\bX)|g_H^\ast(\bX)g(\bX)
\right]\\
=&
P\left[
|h_0(\bX)|
\{1-g_H^\ast(\bX)g(\bX)\}
\right]\\
=&
P\left[
|h_0(\bX)|
\{1-d^\ast(\bX)g(\bX)\}
\right]\\
\geq&
P\left[
|h_0(\bX)|
\mone\{\operatorname{sign}(g(\bX))\neq d^\ast(\bX)\}
\right]\\
\geq&
2\epsilon_n(1-\epsilon_n)
\left[
V(d^\ast)-V\{\operatorname{sign}(g)\}
\right].
\end{align*}
The conclusion follows by applying the preceding result with $g=\widehat g_n$ in \eqref{eq:hinge-excess-proof}.
\end{proof}

\subsection{Proof of the equivalence between penalized ERM and unpenalized ERM over an RKHS ball}\label{ss:proof-penalized-equivalence}
\begin{proof}
We show that the penalized empirical risk minimization (ERM) problem over the RKHS,
\[
\widehat g_{\lambda_n}
\in\arg\min_{g\in\calH_{K_n}}
\left\{\mathbb P_2\ell_\phi(g;\widehat\eta)
+\lambda_n\|g\|_{\calH_{K_n}}^2\right\},
\]
is equivalent to unpenalized ERM over the radius-$R_n$ RKHS ball
\[
\calG_n^{\mathrm{svm}}(R_n)
=\{g\in\calH_{K_n}:\|g\|_{\calH_{K_n}}\leq R_n\}.
\]
Indeed, write
$Q_n(g)=\mathbb P_2\ell_\phi(g;\widehat\eta)$ and let
$R_{n,\lambda}=\|\widehat g_{\lambda_n}\|_{\calH_{K_n}}$.
For every $g$ with
$\|g\|_{\calH_{K_n}}\leq R_{n,\lambda}$,
penalized optimality gives
\[
 Q_n(\widehat g_{\lambda_n})
 +\lambda_nR_{n,\lambda}^2
 \leq
 Q_n(g)+\lambda_n\|g\|_{\calH_{K_n}}^2
 \leq
 Q_n(g)+\lambda_nR_{n,\lambda}^2.
\]
Therefore
$Q_n(\widehat g_{\lambda_n})\leq Q_n(g)$, proving that
$\widehat g_{\lambda_n}$ minimizes $Q_n$ over the matched RKHS
ball.

Conversely, fix $R_n>0$ and let $\widehat g_{R_n}\in
\arg\min_{g\in\calG_n^{\mathrm{svm}}(R_n)}
Q_n(g)$. Then we can write 
\[Q_n(g)
=
\frac{1}{n_2}\sum_{i\in\mathcal F_2}
|Z_{\widehat\eta,i}|\phi\{\operatorname{sign}(Z_{\widehat\eta,i}) g(\bX_i)\}.\]
Because $g\mapsto \operatorname{sign}(Z_{\widehat\eta,i}) g(\bX_i)$ is linear, convexity of
$\phi$ and nonnegativity of $|Z_{\widehat\eta,i}|$ imply that $Q_n$ is convex. Moreover, for
$g,h\in\calG_n^{\mathrm{svm}}(R_n)$, we have
$|g(\bX_i)|\vee|h(\bX_i)|\leq \kappa_nR_n$. If $\phi$ is
$L_{\phi,n}$-Lipschitz on $[-\kappa_nR_n,\kappa_nR_n]$, then the reproducing
property and Cauchy--Schwarz give
\begin{align*}
|Q_n(g)-Q_n(h)|
\leq &\frac{L_{\phi,n}}{n_2}
 \sum_{i\in\mathcal F_2}|Z_{\widehat\eta,i}|
 |g(\bX_i)-h(\bX_i)| \\
\leq &\frac{L_{\phi,n}}{n_2}
 \sum_{i\in\mathcal F_2}|Z_{\widehat\eta,i}|
 K_n(\bX_i,\bX_i)^{1/2}
 \|g-h\|_{\calH_{K_n}} \\
\leq &4M L_{\phi,n}\kappa_n
 \|g-h\|_{\calH_{K_n}},
\end{align*}
where the last inequality follows from $K_n(\bX_i,\bX_i)\leq \kappa_n^2$ and \eqref{eq:bound-Zeta-4M}. Hence $Q_n$ is Lipschitz continuous, and therefore continuous, on
$\calG_n^{\mathrm{svm}}(R_n)$.

Let
$c_{n,R}(g)=\|g\|_{\calH_{K_n}}^2-R_n^2$.
By \citet[Chapter~8, Section~8.3, Theorem~1]
{luenberger1969optimization}, applied with
$X=\calH_{K_n}$, $Z=\mathbb R$, and $P=\mathbb R_+$,
strict feasibility $c_{n,R}(0)=-R_n^2<0$ implies that there
exists $\lambda_{n,R}\geq0$ such that
\[
\widehat g_{R_n}
\in
\arg\min_{g\in\calH_{K_n}}
\left\{
Q_n(g)+\lambda_{n,R}\lsb \|g\|_{\calH_{K_n}}^2-R_n^2\rsb
\right\}
\]
and
\begin{equation}\label{eq:KKT-condition}
    \lambda_{n,R}
\left\{
\|\widehat g_{R_n}\|_{\calH_{K_n}}^2-R_n^2
\right\}=0.
\end{equation}
Moreover, Hilbert space Fermat and subdifferential sum rules
\citep[Theorem~16.3 and Corollary~16.48(iii)]
{bauschke2017convex} yield
\[0\in
\partial Q_n(\widehat g_{R_n})
+
2\lambda_{n,R}\widehat g_{R_n}.\]
If
$\|\widehat g_{R_n}\|_{\calH_{K_n}}<R_n$, \eqref{eq:KKT-condition} implies $\lambda_{n,R}=0$. Conversely, if
$\lambda_{n,R}=0$, then $\widehat g_{R_n}$ is an unconstrained
minimizer of $Q_n$; hence, if it is not an unconstrained minimizer,
then $\lambda_{n,R}>0$.
\end{proof}

\subsection{Proof of Corollary \ref{cor:bounded-svm}}\label{ss:svm-technical-supp}

For completeness, we first state the empirical process bound needed for general calibrated SVM losses.

\begin{lemma}[\emph{Hilbert score surrogate complexity}]\label{lem:svm-surrogate-complexity}
Assume the conditions of Theorem \ref{thm:surrogate-oracle}, $|Y|\leq M$,
$\|\widehat m\|_\infty\leq M$, and $\|\widehat e\|_\infty\leq1$. Let
$B_n=\kappa_nR_n$, suppose that $\phi$ is $L_{\phi,n}$-Lipschitz on
$[-B_n,B_n]$, and set $K_{\phi,n}=\sup_{|u|\leq B_n}\phi(u)$. Conditional on the nuisance training split, with probability at least $1-\delta$,
\begin{align}\label{eq:svm-surrogate-empirical-bound}
 \Delta_{n,\phi}(\calG_n^{\mathrm{svm}}(R_n);\widehat\eta)
 \leq C M\left\{
 \frac{L_{\phi,n}\kappa_nR_n}{\sqrt{n_2}}
 +K_{\phi,n}\sqrt{\frac{\log(2/\delta)}{n_2}}
 \right\}.
\end{align}
Consequently, Theorem \ref{thm:surrogate-oracle} applies with
\eqref{eq:svm-surrogate-empirical-bound} and the nuisance bound
M.\eqref{eq:generic-surrogate-nuisance-bound}.
\end{lemma}

\begin{proof}[Proof of Lemma \ref{lem:svm-surrogate-complexity} and Corollary \ref{cor:bounded-svm}]
Condition on the nuisance training split and write
$W_i=|Z_{\widehat\eta,i}|\leq4M$ by \eqref{eq:bound-Zeta-4M} and $H_i=\operatorname{sign}(Z_{\widehat\eta,i})$. For $g\in\calG_n^{\mathrm{svm}}(R_n)$, the reproducing property gives
\begin{align*}
 |g(x)|=\langle g,K_n(x,\cdot)\rangle_{\calH_{K_n}}\leq\|g\|_{\calH_{K_n}}\|K_n(x,\cdot)\|_{\calH_{K_n}}\leq\|g\|_{\calH_{K_n}}K_n(x,x)^{1/2}
 \leq\kappa_nR_n=B_n,
\end{align*}
where the first inequality follows from the Cauchy-Schwarz inequality in the RKHS and the second inequality follows from the reproducing property applied to the kernel section with
\[\|K_n(x,\cdot)\|_{\calH_{K_n}}=\sqrt{\langle K_n(x,\cdot),K_n(x,\cdot)\rangle_{\calH_{K_n}}}=\sqrt{K_n(x,x)}.\]
Define the centered coordinate maps
\begin{align*}
 \varphi_i(t)=W_i\lb\phi(H_it)-\phi(0)\rb.
\end{align*}
They (i) vanish at zero with $\varphi_i(0)=0$ for all $i$; and (ii) are $4ML_{\phi,n}$-Lipschitz on $[-B_n,B_n]$. To see (ii), for any $s,t\in[-B_n,B_n]$, since $|H_i|\leq1$,
$H_is,H_it\in[-B_n,B_n]$. Hence, using the
$L_{\phi,n}$-Lipschitz continuity of $\phi$ on this interval and
$W_i\leq4M$,
\begin{align*}
    |\varphi_i(t)-\varphi_i(s)|
=&
W_i\bigl|\phi(H_it)-\phi(H_is)\bigr|\\
\leq&
W_iL_{\phi,n}|H_i(t-s)|\\
\leq&
4ML_{\phi,n}|t-s|.
\end{align*}

To apply the Ledoux--Talagrand contraction method, let
$L_n=4ML_{\phi,n}$,
$\pi_{B_n}(u)=(-B_n)\vee(u\wedge B_n)$, and 
\[\rho_i(u)
=
L_n^{-1}W_i
\left[\phi\{\pi_{B_n}(u)\}-\phi(0)\right].\]
By construction, $\pi_{B_n}$ is $1$-Lipschitz because it is the composition of the two $1$-Lipschitz maps $u\mapsto u\wedge B_n$ and $u\mapsto(-B_n)\vee u$. Since $W_i\leq4M$, $\phi$ is $L_{\phi,n}$-Lipschitz on
$[-B_n,B_n]$, and $\pi_{B_n}$ is $1$-Lipschitz, each $\rho_i$
satisfies $\rho_i(0)=0$ and $|\rho_i(u)-\rho_i(v)|\leq|u-v|$, for $u,v\in\mathbb R$.
Hence $\rho_i$ is a contraction \citep[\S4.5, p.~112, Theorem~4.12]{ledoux1991probability}. Apply
\citet[Theorem~4.12]{ledoux1991probability} with the contractions $\rho_i$, $F(u)=u$ and
\[
T_n
=
\left\{
\bigl(H_i g(\bX_i)\bigr)_{i\in\mathcal F_2}:
\|g\|_{\calH_{K_n}}\leq R_n
\right\}
\subset[-B_n,B_n]^{n_2}.
\]
Since
$L_n\rho_i\{H_i g(\bX_i)\}
=\varphi_i\{g(\bX_i)\}$, the theorem gives
\begin{equation*}
\bE_\xi\sup_{\|g\|_{\calH_{K_n}}\leq R_n}
\left|
\frac1{n_2}\sum_{i\in\mathcal F_2}
\xi_i\varphi_i\{g(\bX_i)\}
\right|\leq
8ML_{\phi,n}
\bE_\xi\sup_{\|g\|_{\calH_{K_n}}\leq R_n}
\left|
\frac1{n_2}\sum_{i\in\mathcal F_2}
\xi_iH_i g(\bX_i)
\right|.    
\end{equation*}

For each realization of $\xi$, let $S_\xi
=
\sum_{i\in\mathcal F_2}
\xi_iH_iK_n(\bX_i,\cdot)$. By the reproducing property and the bilinearity of the inner product, 
\[
\sum_{i\in\mathcal F_2}\xi_iH_i g(\bX_i)
=\sum_{i\in\mathcal F_2}\xi_iH_i \langle g,K_n(\bX_i,\cdot)\rangle_{\calH_{K_n}}=
\langle g,S_\xi\rangle_{\calH_{K_n}}.
\]
The Cauchy-Schwarz inequality then gives 
\[
\sup_{\|g\|_{\calH_{K_n}}\leq R_n}
|\langle g,S_\xi\rangle_{\calH_{K_n}}|
\leq R_n\|S_\xi\|_{\calH_{K_n}},
\]
where equality is attained, when $S_\xi\neq0$, at
$g=R_nS_\xi/\|S_\xi\|_{\calH_{K_n}}$. Therefore, we obtain
\begin{align*}
 &\bE_\xi\sup_{\|g\|_{\calH_{K_n}}\leq R_n}
 \left|\frac{1}{n_2}\sum_{i\in\mathcal F_2}\xi_i\varphi_i\{g(\bX_i)\}\right|\\
 \leq& 8ML_{\phi,n}
 \bE_\xi\sup_{\|g\|_{\calH_{K_n}}\leq R_n}
 \left|\frac{1}{n_2}\sum_{i\in\mathcal F_2}
 \xi_iH_i g(\bX_i)\right|\\
 \leq & 8ML_{\phi,n}\frac{R_n}{n_2}
 \bE_\xi\left\|\sum_{i\in\mathcal F_2}
 \xi_iH_iK_n(\bX_i,\cdot)\right\|_{\calH_{K_n}}\\
 \leq&8ML_{\phi,n}\frac{R_n}{n_2}\lsb
\bE_\xi\|S_\xi\|_{\calH_{K_n}}^2
\rsb^{1/2}\\
 \leq&8ML_{\phi,n}\frac{R_n}{n_2}\left\{
\sum_{i,j\in\mathcal F_2}
H_iH_j\bE_\xi(\xi_i\xi_j)
K_n(\bX_i,\bX_j)
\right\}^{1/2}\\
 \leq&8ML_{\phi,n}\frac{R_n}{n_2}\left\{
\sum_{i\in\mathcal F_2}
H_i^2K_n(\bX_i,\bX_i)
\right\}^{1/2}\\
 \leq& 8ML_{\phi,n}\frac{\kappa_nR_n}{\sqrt{n_2}},
\end{align*}
where the third inequality follows from Jensen's inequality, the fourth inequality follows from the bilinearity of the inner product, the fifth inequality follows from the fact that $\bE_\xi(\xi_i\xi_j)=0$ for any $i\neq j$, and the last inequality follows from $K_n(\bX_i,\bX_i)\leq \kappa_n^2$.

Let $f_g(\mO_i)
=
W_i\phi\{H_i g(\bX_i)\}$. Since $W_i\leq4M$, $|g(\bX_i)|\leq B_n$, and $\phi\geq0$, $0\leq f_g(\mO_i)\leq 4M K_{\phi,n}$. By conditional symmetrization and
$f_g(\mO_i)=\varphi_i\{g(\bX_i)\}+W_i\phi(0)$,
\begin{align*}
\bE(\Delta_{n,\phi}\mid\mathcal F_1)
\leq&
2\bE_{O,\xi}\sup_{\|g\|_{\calH_{K_n}}\leq R_n}
\left|
\frac1{n_2}\sum_i\xi_i
\varphi_i\{g(\bX_i)\}
\right|+
2\bE_{O,\xi}
\left|
\frac1{n_2}\sum_i\xi_iW_i\phi(0)
\right|\\
\leq&
16ML_{\phi,n}
\frac{\kappa_nR_n}{\sqrt{n_2}}
+
\frac{8M K_{\phi,n}}{\sqrt{n_2}},
\end{align*}
where the last term follows from
\[
\bE_\xi
\left|
\frac1{n_2}\sum_i\xi_iW_i\phi(0)
\right|
\leq
\frac1{n_2}
\left\{\sum_iW_i^2\phi(0)^2\right\}^{1/2}
\leq
\frac{4M K_{\phi,n}}{\sqrt{n_2}}.
\]

To control deviations, let $S^{(j)}$ be obtained from the
second-stage sample $S$ by replacing $\mO_j$ with $\mO_j'$. Then
\[
\begin{aligned}
|\Delta_{n,\phi}(S)-\Delta_{n,\phi}(S^{(j)})|
\leq&
\frac1{n_2}
\sup_{\|g\|_{\calH_{K_n}}\leq R_n}
|f_g(\mO_j)-f_g(\mO_j')|\\
\leq&
\frac{4M K_{\phi,n}}{n_2}.
\end{aligned}
\]
Consequently, the conditional bounded-difference inequality \citep{mcdiarmid1989method} yields
\[
\Pr\left(
\left|
\Delta_{n,\phi}
-\bE(\Delta_{n,\phi}\mid\mathcal F_1)
\right|>t
\,\middle|\,
\mathcal F_1
\right)
\leq
2\exp\left\{
-\frac{n_2t^2}{8M^2K_{\phi,n}^2}
\right\}.
\]
Taking $t
=
2\sqrt{2}\,M K_{\phi,n}
\sqrt{{\log(2/\delta)}/{n_2}}$
and absorbing numerical constants gives, with conditional
probability at least $1-\delta$,
\[
\Delta_{n,\phi}
\leq
CM\left\{
\frac{L_{\phi,n}\kappa_nR_n}{\sqrt{n_2}}
+
K_{\phi,n}
\sqrt{\frac{\log(2/\delta)}{n_2}}
\right\}.
\]

For the bounded SVM construction, define the clipped image class $\widetilde{\calG}_n^{\mathrm{bsvm}}(R_n)
=
\{
T_1\circ f:
f\in\calG_n^{\mathrm{svm}}(R_n)
\}$, where \(T_1(u)=(-1)\vee(u\wedge1)\). Let \(\widehat f_n\) be a
\(\rho_n\)-approximate empirical hinge minimizer over
\(\calG_n^{\mathrm{svm}}(R_n)\) and set
\(\widehat g_n=T_1\circ\widehat f_n\). Clipping need not preserve RKHS membership, so $\widehat g_n$ may not belong to $\calG_n^{\mathrm{bsvm}}(R_n)$; instead, $\widehat g_n\in\widetilde{\calG}_n^{\mathrm{bsvm}}(R_n)$. Since clipping cannot increase hinge loss and $\calG_n^{\mathrm{bsvm}}(R_n)\subseteq \calG_n^{\mathrm{svm}}(R_n)$,
\[
\mathbb P_2\ell_H(\widehat g_n;\widehat\eta)
\leq
\mathbb P_2\ell_H(\widehat f_n;\widehat\eta)
\leq
\inf_{g\in\calG_n^{\mathrm{bsvm}}(R_n)}
\mathbb P_2\ell_H(g;\widehat\eta)+\rho_n.
\]
Moreover, $\calG_n^{\mathrm{bsvm}}(R_n)\subseteq\widetilde{\calG}_n^{\mathrm{bsvm}}(R_n)$ because $T_1\circ g=g$ for every bounded comparator $g$. Thus, for any $g_n^\circ\in\calG_n^{\mathrm{bsvm}}(R_n)$, the same add-and-subtract argument used in the proof of Theorem \ref{thm:bounded-hinge-sieve} gives
\begin{align*}
 \calL_H^{\mathrm{linear}}(\widehat g_n;\widehat\eta)
 -\calL_H^{\mathrm{linear}}(g_n^\circ;\widehat\eta)
 \leq
 2\Delta_{n,H}^{\mathrm{linear}}
 \bigl(\widetilde{\calG}_n^{\mathrm{bsvm}}(R_n);\widehat\eta\bigr)
 +\rho_n.
\end{align*}
Both scores are bounded by one, so the nuisance comparison in \eqref{eq:hinge-exact-nuisance-proof} remains unchanged. It remains to control the empirical process over the clipped image class. Conditional on the nuisance training split, \(T_1\) is \(1\)-Lipschitz, satisfies \(T_1(0)=0\), and is a contraction. Hence, by the contraction inequality, we obtain
\begin{align*}
&\bE_\xi
\sup_{f\in\calG_n^{\mathrm{svm}}(R_n)}
\left|
\frac{1}{n_2}
\sum_{i\in\mathcal F_2}
\xi_i Z_{\widehat\eta,i}
T_1\{f(\bX_i)\}
\right|\\
\leq&
4M\bE_\xi
\sup_{f\in\calG_n^{\mathrm{svm}}(R_n)}
\left|
\frac{1}{n_2}
\sum_{i\in\mathcal F_2}
\xi_i f(\bX_i)
\right|\\
\leq&
\frac{4M\kappa_nR_n}{\sqrt{n_2}}.
\end{align*}
Conditional symmetrization and the same bounded-difference argument
therefore give, with conditional probability at least \(1-\delta\),
\begin{align*}
\Delta_{n,H}^{\mathrm{linear}}
\bigl(\widetilde{\calG}_n^{\mathrm{bsvm}}(R_n);
\widehat\eta\bigr)
\leq
8M\left\{
\frac{\kappa_nR_n}{\sqrt{n_2}}
+
\sqrt{\frac{\log(2/\delta)}{n_2}}
\right\}.
\end{align*}
Combining the empirical comparison above, the nuisance expansion in \eqref{eq:hinge-exact-nuisance-proof}, the bounded hinge calibration argument, and the contraction bound proves
\begin{align*}
    V(d^\ast)-V\{\operatorname{sign}(\widehat g_n)\}
 \leq \frac{a_{n,\mathrm{svm}}^H+16 Mq_{n,\mathrm{svm}}(\delta)
 +\rho_n+2\|\widehat e-e_0\|_2\|\widehat m-m_0\|_2}{2\epsilon_n(1-\epsilon_n)}.
\end{align*}
\end{proof}

\subsection{Proof of Corollary \ref{thm:relu-regret}}\label{sec:supp-relu-details}
We recap the basic setup. Without loss of generality, let $\bX\in[0,1]^p$. For a scalar $u$, write $(u)_+=u\vee 0$ and let $\sigma_v(u)=(u-v)_+$ act componentwise. Following \citet{schmidt2020nonparametric}, for a depth $L$, width vector $\mathbf p=(p_0,\ldots,p_{L+1})$ with $p_0=p$ and $p_{L+1}=1$, sparsity $s$, and output bound $F$, let $\mathcal F(L,\mathbf p,s,F)$ consist of networks
\begin{align}
 f(\bX)=W_L\sigma_{v_L}\circ W_{L-1}\sigma_{v_{L-1}}\circ\cdots \circ W_1\sigma_{v_1}\circ W_0\bX
\end{align}
satisfying
\begin{align*}
 &\max_{0\leq\ell\leq L}\|W_\ell\|_\infty
 \vee\max_{1\leq\ell\leq L}\|v_\ell\|_\infty\leq1,\\
 &\sum_{\ell=0}^L\|W_\ell\|_0
 +\sum_{\ell=1}^L\|v_\ell\|_0\leq s,
 \qquad \|f\|_\infty\leq F.
\end{align*}
Here, $\circ$ denotes composition. We consider the hard tanh output activation and set $F=1$.

For $\calG_n=\mathcal F(L_n,\mathbf p_n,s_n,1)$ and
$V_n=\prod_{\ell=0}^{L_n+1}(p_{\ell,n}+1)$, Lemma 5 of \citet{schmidt2020nonparametric} gives
\begin{align}\label{eq:schmidt-hieber-entropy}
 \log N\lsb\zeta,\calG_n,\|\cdot\|_\infty\rsb
 \leq(s_n+1)\log\lbb2\zeta^{-1}(L_n+1)V_n^2\rbb,
\end{align}
where $N\{\zeta,\calG_n,\|\cdot\|_\infty\}$ denotes the covering number, namely, the minimum number of balls of radius $\zeta$ under the $\|\cdot\|_\infty$ norm required to cover $\calG_n$.

\begin{proof}[Proof of Theorem \ref{thm:relu-regret}]
The class $\calG_n$ is bounded by one, so Theorem \ref{thm:bounded-hinge-sieve} applies. It remains to control the empirical process term $\Delta_{n,H}^{\mathrm{linear}}(\calG_n;\widehat\eta)=\sup_{g\in\calG_n}
 \left|\{\mathbb P_2-P\}
 \{Z_{\widehat\eta}g(\bX)\}\right|$. By \eqref{eq:bound-Zeta-4M}, $|Z_{\widehat\eta}|\leq4M$. 
Let \(h_1,\ldots,h_N\) be centers of a \(1/n_2\)-cover of
\(\calG_n\), where
\[
N=N\lsb\frac{1}{n_2},\calG_n,\|\cdot\|_\infty\rsb.
\]
Thus, for every \(g\in\calG_n\), there exists an index \(j=j(g)\)
such that
\[
\|g-h_j\|_\infty\leq\frac{1}{n_2}.
\]

Because the covering centers need not belong to \(\calG_n\), they
need not initially be bounded by one. Define the hard tanh clipping map
\[
T_1(u)
=
(-1)\vee(u\wedge1)
=
-1+(u+1)_+-(u-1)_+
\]
and the clipped centers $\bar h_j=T_1\circ h_j$, $j=1,\ldots,N$.
Since \(\calG_n\subseteq[-1,1]^{\calX}\), for every
\(g\in\calG_n\), every \(x\in\calX\), and every \(j\),
\[
|g(x)-\bar h_j(x)|
=
|g(x)-T_1\{h_j(x)\}|
\leq
|g(x)-h_j(x)|.
\]
Consequently, the clipped centers also form a \(1/n_2\)-cover:
for every \(g\in\calG_n\), there exists \(j=j(g)\) such that
\[
\|g-\bar h_j\|_\infty
\leq
\|g-h_j\|_\infty
\leq
\frac{1}{n_2}.
\]
Importantly, by definition, it follows that $\|\bar h_j\|_\infty\leq1$, $j=1,\ldots,N$.

For every \(g\in\calG_n\), choose a corresponding clipped covering
center \(\bar h_j\). Since
\(|Z_{\widehat\eta}|\leq4M\) by \eqref{eq:bound-Zeta-4M}, it follows that
\begin{align*}
&\left|
\{\mathbb P_2-P\}
\left[
Z_{\widehat\eta}
\{g(\bX)-\bar h_j(\bX)\}
\right]
\right|\\
\leq&
\mathbb P_2
\left[
|Z_{\widehat\eta}|
\,|g(\bX)-\bar h_j(\bX)|
\right]
+
P
\left[
|Z_{\widehat\eta}|
\,|g(\bX)-\bar h_j(\bX)|
\right]\\
\leq&
\frac{4M}{n_2}
+
\frac{4M}{n_2}
=
\frac{8M}{n_2}.
\end{align*}
Therefore,
\begin{align}
\Delta_{n,H}^{\mathrm{linear}}
(\calG_n;\widehat\eta)
=&
\sup_{g\in\calG_n}
\left|
\{\mathbb P_2-P\}
\{Z_{\widehat\eta}g(\bX)\}
\right|\notag\\
&\leq
\max_{1\leq j\leq N}
\left|
\{\mathbb P_2-P\}
\{Z_{\widehat\eta}\bar h_j(\bX)\}
\right|
+
\frac{8M}{n_2}.\label{eq:approximation-covering-bound}
\end{align}

Finally, condition on the nuisance training split
\(\mathcal F_1\). Then \(\widehat\eta\), the class \(\calG_n\), and
the clipped covering centers
\(\bar h_1,\ldots,\bar h_N\) are fixed, while the observations
\(\{\mO_i:i\in\mathcal F_2\}\) remain independent. For each
\(j=1,\ldots,N\), define $W_{i,j}
:=
Z_{\widehat\eta}(\mO_i)\bar h_j(\bX_i)$, $i\in\mathcal F_2$.
Since \(|Z_{\widehat\eta}|\leq4M\) by \eqref{eq:bound-Zeta-4M} and
\(\|\bar h_j\|_\infty\leq1\), $-4M\leq W_{i,j}\leq4M$. Define $S_{n_2,j}
:=
\sum_{i\in\mathcal F_2}W_{i,j}$. By independence of the two sample splits, we have $S_{n_2,j}
-
\bE(S_{n_2,j}\mid\mathcal F_1)
=
n_2
\{\mathbb P_2-P\}
\{Z_{\widehat\eta}\bar h_j(\bX)\}$.

For a fixed \(j\), Hoeffding's inequality, with
\(a_i=-4M\), \(b_i=4M\), and hence $\sum_{i\in\mathcal F_2}(b_i-a_i)^2
=
n_2(8M)^2
=
64M^2n_2$, gives, for every \(t>0\),
\begin{align*}
&\Pr\left(
\left|
S_{n_2,j}
-
\bE(S_{n_2,j}\mid\mathcal F_1)
\right|
\geq t
\ \middle|\ \mathcal F_1
\right)\leq
2\exp\left\{
-\frac{2t^2}{64M^2n_2}
\right\}
=
2\exp\left\{
-\frac{t^2}{32M^2n_2}
\right\}.
\end{align*}
Equivalently, setting \(t=n_2u\), for every \(u>0\),
\begin{align*}
&\Pr\left(
\left|
\{\mathbb P_2-P\}
\{Z_{\widehat\eta}\bar h_j\}
\right|
\geq u
\ \middle|\ \mathcal F_1
\right)\leq
2\exp\left\{
-\frac{n_2u^2}{32M^2}
\right\}.
\end{align*}

A union bound over \(j=1,\ldots,N\) now yields
\begin{align*}
&\Pr\left(
\max_{1\leq j\leq N}
\left|
\{\mathbb P_2-P\}
\{Z_{\widehat\eta}\bar h_j\}
\right|
\geq u
\ \middle|\ \mathcal F_1
\right)\leq
2N\exp\left\{
-\frac{n_2u^2}{32M^2}
\right\}.
\end{align*}
Choose
\[
u=u_\delta
:=
4\sqrt{2}\,M
\left\{
\frac{\log N+\log(2/\delta)}{n_2}
\right\}^{1/2}.
\]
For this choice,
\begin{align*}
2N\exp\left\{
-\frac{n_2u_\delta^2}{32M^2}
\right\}=
2N
\exp\left\{
-\log N-\log(2/\delta)
\right\}=\delta.
\end{align*}
Finally, it follows that, with conditional
probability at least \(1-\delta\),
\begin{equation}\label{eq:hoeffding-finite-bound}
    \max_{1\leq j\leq N}
\left|
\{\mathbb P_2-P\}
\{Z_{\widehat\eta}\bar h_j\}
\right|
\leq
4\sqrt{2}\,M
\left\{
\frac{H_n+\log(2/\delta)}{n_2}
\right\}^{1/2},
\end{equation}
where $H_n=\log N=\log N\lsb n_2^{-1},\calG_n,\|\cdot\|_\infty\rsb$. 

Combining \eqref{eq:hoeffding-finite-bound} with \eqref{eq:approximation-covering-bound} gives, with conditional probability at least \(1-\delta\),
\begin{align}
\Delta_{n,H}^{\mathrm{linear}}
(\calG_n;\widehat\eta)
\leq&
4\sqrt{2}\,M
\left\{
\frac{H_n+\log(2/\delta)}{n_2}
\right\}^{1/2}
+
\frac{8M}{n_2}\notag\\
\leq&
8 M\left\{
\left(
\frac{H_n+\log(2/\delta)}{n_2}
\right)^{1/2}
+
\frac{1}{n_2}
\right\}.
\label{eq:relu-empirical-process-proof}
\end{align}
Corollary \ref{thm:relu-regret} follows from applying \eqref{eq:schmidt-hieber-entropy} to \eqref{eq:relu-empirical-process-proof}.

\end{proof}



\setcounter{section}{2}\setcounter{subsection}{0}
\section*{Appendix B}
\addcontentsline{toc}{section}{Appendix B}

\subsection{Supporting information for the simulation experiments}
\label{ss:supporting-info-simulation}

We consider four data generating processes (DGPs). DGP~1 has a smooth nonlinear decision boundary under sufficient overlap; DGP~2 combines limited overlap with a decision boundary represented exactly by a depth-2 axis-aligned decision tree; DGP~3 is a fully parametric linear benchmark; and DGP~4 has sufficient true overlap but uses a misspecified logistic-linear propensity score working model, a commonly used specification in practice. Each Monte Carlo sample contains \(n=1000\) observations. Treatment
is coded as \(A\in\{-1,1\}\). Let the true nuisance functions be $\pi_0(\bX)=\mP(A=1\mid\bX)$, $m_0(\bX)=\bE(Y\mid\bX)$, $\mu_0(a,\bX)=\bE(Y\mid A=a,\bX)$, and $\tau_0(\bX)=\mu_{0}(1,\bX)-\mu_{0}(-1,\bX)$. The nuisance functions satisfy $\mu_{0}(1,\bX)=m_0(\bX)+\{1-\pi_0(\bX)\}\tau_0(\bX)$ and $\mu_{0}(-1,\bX)=m_0(\bX)-\pi_0(\bX)\tau_0(\bX)$. By \cite{robinson1988root}, if $W\mid\bX\sim\operatorname{Bernoulli}\{\pi_0(\bX)\}$, $A=2W-1$, and $Y=\mu_0(A,\bX)+\varepsilon$ with \(\bE(\varepsilon\mid A,\bX)=0\), then
\(\bE(Y\mid\bX)=m_0(\bX)\) and
\(\mu_{0}(1,\bX)-\mu_{0}(-1,\bX)=\tau_0(\bX)\).

\begingroup
\setlength{\jot}{2pt}

\paragraph{DGP 1: smooth nonlinear decision boundary under sufficient overlap.} Let
\(X_j\stackrel{\mathrm{iid}}{\sim}\operatorname{Unif}(0,1)\),
\(Z_j=2X_j-1\), \(j=1,\ldots,10\), and define
\[
\begin{aligned}
\eta(\bX)
=&1.2Z_6-Z_7+0.6Z_8+0.6Z_9Z_{10},\\
\pi_0(\bX)
=&0.20+0.60\operatorname{expit}\{\eta(\bX)\},\\
m_0(\bX)
=&2+Z_6+0.5Z_7+0.25Z_8.
\end{aligned}
\]
To define the decision boundary, set
\[
\widetilde b_{kj}
=
\sin\left(\frac{2\pi kj}{11}\right)
+0.55\cos\left\{\frac{2\pi(k+1)j}{11}\right\},
\quad k=1,\ldots,6,\quad j=1,\ldots,10.
\]
Writing \(\widetilde{\mathbf b}_k\) for row \(k\), let $\overline{\mathbf b}
=6^{-1}\sum_{k=1}^6\widetilde{\mathbf b}_k$ and 
\begin{equation*}
    \mathbf b_k
=\frac{6\left(\widetilde{\mathbf b}_k-\overline{\mathbf b}\right)}{\sum_{k=1}^6
\left\|\widetilde{\mathbf b}_k-\overline{\mathbf b}\right\|_2}.
\end{equation*}
For \(\mathbf Z=(Z_1,\ldots,Z_{10})^\top\), define $r_k
=\mathbf b_k^\top\mathbf Z$, $L_t(\mathbf r)=t\log\left\{\sum_{k=1}^6\exp(r_k/t)\right\}$, $g(\bX)
=L_{0.14}\{\mathbf r(\bX)\}
-L_{0.14}\{-\mathbf r(\bX)\}$, and $\tau_0(\bX)=3.87\tanh\{g(\bX)/0.15\}$. Thus \(d^*(\bX)=\operatorname{sign}\{\tau_0(\bX)\}\). The decision boundary is a complex, smooth difference of nonlinear polyhedral envelopes and depends on all ten covariates. Finally,
\(\varepsilon\sim\operatorname{Unif}(-2\sqrt{3},2\sqrt{3})\).

\paragraph{DGP 2: decision boundary characterized by a depth-2 axis-aligned decision tree under limited overlap.} Again let
\(X_j\stackrel{\mathrm{iid}}{\sim}\operatorname{Unif}(0,1)\),
\(Z_j=2X_j-1\), \(j=1,\ldots,10\), and define
\[
\begin{aligned}
d^*(\bX)
=&
\begin{cases}
\operatorname{sign}(Z_2), & X_1<1/2,\\
\operatorname{sign}(Z_3), & X_1\geq1/2,
\end{cases}
\end{aligned}
\]
where $\tau_0(\bX)=1.25d^*(\bX)$. This decision boundary is represented exactly by a
depth-2 axis-aligned tree. The propensity and marginal outcome regression
are
\[
\begin{aligned}
\pi_0(\bX)
=&\operatorname{expit}(3.5Z_4+1.5Z_5+0.5Z_6),\\
m_0(\bX)
=&2+5\sin(\pi Z_4)+2.5\cos(\pi Z_5)+0.75Z_6.
\end{aligned}
\]
The true propensity score follows a logit-linear model and exhibits limited overlap, with theoretical support approximately \([0.0041,0.9959]\). We use
\(\varepsilon\sim\operatorname{Unif}(-\sqrt{3},\sqrt{3})\). Because
\(\lvert\tau_0(\bX)\rvert=1.25\), this DGP satisfies $\operatorname{Regret}(\widehat d)
=1.25\,\Pr\{\widehat d(\bX)\neq d^*(\bX)\}$.

\paragraph{DGP 3: correctly specified linear benchmark.} Let
\(X_j\stackrel{\mathrm{iid}}{\sim}\operatorname{Unif}(0,1)\),
\(Z_j=2X_j-1\), \(j=1,\ldots,10\), and set
\[
\begin{aligned}
\pi_0(\bX)
=&\frac{1}{2},
\qquad
m_0(\bX)=2+0.5Z_3-0.3Z_4+0.2Z_5,\\
\tau_0(\bX)
=&3(Z_1+Z_2),
\qquad
d^*(\bX)=\operatorname{sign}(Z_1+Z_2).
\end{aligned}
\]
Thus, the true propensity score follows a logit-linear model, \(m_0\) and both arm specific regressions follow linear models, and the optimal ITR is induced by a linear decision rule. We generate
\(\varepsilon\sim\operatorname{Unif}(-\sqrt{3},\sqrt{3})\).

\paragraph{DGP 4: decision boundary characterized by a depth-2 axis-aligned decision tree with propensity score misspecification.} Let
\(\Pr(U=-0.5)=0.475\), \(\Pr(U=0.5)=0.475\), and \(\Pr(U=4)=0.05\), and
independently generate
\(V,Z_3,\ldots,Z_{30}\stackrel{\mathrm{iid}}{\sim}\operatorname{Unif}(-1,1)\).
Define \(B=\operatorname{sign}(V)\), \(P_U=\mone(U\geq0.5)\), and
\(R=\mone(U>2)\), and set
\[
\begin{aligned}
\pi_0(\bX)
=&
\begin{cases}
0.25, & U=-0.5,\\
0.75, & U\in\{0.5,4\},
\end{cases}\\
m_0(\bX)
=&3+3Z_3,
\qquad
\tau_0(\bX)=B\{1.25(1-R)-3.2R\}.
\end{aligned}
\]
The optimal ITR is
\[
d^*(\bX)=
\begin{cases}
B, & U\leq2,\\
-B, & U>2.
\end{cases}
\]
This rule is represented exactly by a depth-2 axis-aligned tree. Moreover,
\(m_0(\bX)\) lies in the linear span of \((1,Z_3)\), and both \(\mu_0(1,\bX)\) and
\(\mu_0(-1,\bX)\) lie in the linear span of
\((1,Z_3,B,P_UB,RB)\). The error distribution is a mixture of uniform distributions:
\[
\varepsilon\sim
\begin{cases}
\operatorname{Unif}(-\sqrt{3}/2,\sqrt{3}/2),
& \text{with probability }0.95,\\
\operatorname{Unif}(-\sqrt{3}\sqrt{15.25},
                     \sqrt{3}\sqrt{15.25}),
& \text{with probability }0.05,
\end{cases}
\]
which has conditional mean zero and conditional variance one. Although the
true overlap is sufficient, the fitted propensity working model is the
misspecified raw-\(U\) logistic regression
\(\overline\pi(U)=\operatorname{expit}(\beta_0+\beta_1U)\). Because $\{\overline\pi(-0.5),\overline\pi(0.5),\overline\pi(4)\}=(0.339,0.635,0.992)$, the fitted propensity score model may yield predictions close to \(1\).

\endgroup
\paragraph{Nuisance estimation.} 
All direct methods use five-fold cross-fitted nuisance predictions, with
common fold assignments and common predictions. For DGPs~1 and~2, nuisance functions are estimated using Super Learner \citep{vanderlaan2007super} with the following ensembles:
\[
\begin{split}
\mathcal L_1
=&
\{\texttt{SL.mean},\texttt{SL.glm},\texttt{SL.glmnet},
  \texttt{SL.earth},\texttt{SL.ranger},
  \texttt{SL.xgboost}\},\\
\mathcal L_2
=&
\{\texttt{SL.mean},\texttt{SL.glm},\texttt{SL.glmnet},
  \texttt{SL.gam},\texttt{SL.earth},\texttt{SL.nnet}\}.
\end{split}
\]
DGP~3 uses the singleton
\(\{\texttt{SL.glm}\}\) library, yielding correctly specified cross-fitted
logistic and linear regressions. For DGP~4, the propensity is estimated by
the misspecified logistic regression, whereas
\(m_0\) and \(\mu_0\) are estimated by correctly specified cross-fitted
ordinary least-squares regressions.

\paragraph{Sieve classes, surrogate relaxation, and tuning.}

We compare five ODRL implementations: exact optimization based on VC sieves, bounded SVM hinge (Example \ref{ss:svm-sieves}), SVM logistic (Example \ref{ss:svm-sieves}), ReLU hinge (Example \ref{ss:relu-sieves}), and ReLU logistic (Example \ref{ss:relu-sieves})---with Q-learning \citep{qian2011performance}, dWOLS \citep{wallace2015doubly},
A-learning \citep{schulte2015q}, E-learning \citep{mo2022efficient}, D-learning \citep{qi2018d}, EARL \citep{zhao2019efficient}, RWL \citep{zhou2017residual}, OWL \citep{zhao2012estimating}, Athey--Wager policy
learning \citep{athey2021policy}, and Kallus optimal retargeting \citep{kallus2021more}. In DGPs~1,~2, and~4, ODRL exact optimization,
Athey--Wager policy
learning, and Kallus optimal retargeting use the same exact depth-2 tree optimizer (Example \ref{ss:tree-sieves}). This policy class is correctly specified in DGPs~2
and~4 but not in DGP~1. In DGP~3, the three direct optimizers use the same
linear policy class (Example \ref{ss:linear-sieves}), which contains the true optimal ITR \(d^*\). The SVM-based methods under surrogate relaxation use the Gaussian radial
basis function (RBF) kernel \citep{scholkopf2002learning},
\[
K_{\gamma}(\bx,\bx')
=
\exp\left\{-\gamma\|\bx-\bx'\|_2^2\right\},
\qquad \gamma>0,
\]
in DGPs~1,~2, and~4, and the linear kernel
\(K(\bx,\bx')=\bx^\T\bx'\) in DGP~3. The ReLU classes and all tuning grids are
prespecified and held fixed across replications. All tuning uses
two-fold cross validation, except for bounded SVM hinge, which uses three
folds. DGP~1 uses the broader tuning grids needed for its complex nonlinear decision
boundary; DGPs~2 and~4 use the same reduced grids. The three exact optimizations over linear decision rules in DGP~3 are formulated as mixed integer programs, with a prespecified 40-minute time limit for each fit.  

\paragraph{Monte Carlo performance evaluation.}
For a fitted rule \(\widehat d\), we evaluate its welfare or value, regret, and misclassification rate (MCR) as $V(\widehat d)
=
\bE[\mu_{0}\{\widehat d(\bX),\bX\}]$,
\[
\operatorname{Regret}(\widehat d)
=
V(d^*)-V(\widehat d)
=
\bE\!\left[
|\tau_0(\bX)|
\mone\{\widehat d(\bX)\neq d^*(\bX)\}
\right],
\]
and $\operatorname{MCR}(\widehat d)
=\Pr\{\widehat d(\bX)\neq d^*(\bX)\}$, respectively. All performance metrics are computed over 500 Monte Carlo replications. The Monte Carlo standard errors of regret and MCR are also reported. The same fixed evaluation design is used for every method and replication:
\(20{,}000\) points for DGP~1, $125{,}000$ points for DGP~2,
\(100{,}000\) points for DGP~3, and \(120{,}000\) points for DGP~4. 


\subsection{Empirical analysis of the ONZ field experiment}
\label{sec:data-application}

We apply ODRL to a field experiment conducted in partnership with Oxford Net Zero (ONZ), a research and engagement program at the University of Oxford that promotes science-based approaches to climate neutrality \citep{pereira2025climate}. The published scientist-versus-peer
attribution effects have opposite signs in Europe and the United States,
motivating an exploratory analysis of individualized sender assignment. We
restrict the factorial experiment to the short-term frame and the common
two-source comparison, coding climate-scientist attribution as $A=1$ and
peer-politician attribution as $A=-1$. The reward $Y$ is the
government-level proportion of invited officials who click at least one
link. Our analytical sample contains $2284$ governments, including nonopeners
as zero outcomes; $1133$ receive climate-scientist attribution and 1151
receive peer-politician attribution. The corresponding click rates are
4.88\% and 4.56\%, while the scientist-minus-peer contrasts are 1.45
percentage points in Europe and $-3.02$ percentage points in the United
States, reproducing the reported geographic reversal.

We consider six baseline covariates: country, log population, the number of officials invited, and the shares of invited officials who are women, affiliated with left parties, and affiliated with right parties. Consistent with the simulation experiments, we consider five {ODRL implementations}: bounded ReLU with hinge loss (Example~\ref{ss:relu-sieves}), ReLU with logistic loss (Example~\ref{ss:relu-sieves}), bounded Gaussian RKHS with hinge loss (Example~\ref{ss:svm-sieves}), Gaussian RKHS with logistic loss (Example~\ref{ss:svm-sieves}), and exact optimization over depth-2 axis-aligned decision trees (Example~\ref{ss:tree-sieves}). Five-fold cross-fitted Super Learner \citep{vanderlaan2007super} estimates of $m(\bX)$ are obtained, while the known protocol assignment probability $\pi(\bX)=0.5$ is used. We select the tuning parameters for each policy learner by maximizing its mean three-fold cross-validated policy criterion. 

Table~\ref{tab:ajps-full-policy} reports pairwise agreement among the optimal ITRs estimated by the five ODRL learners in the ONZ field experiment. The proportions recommended to receive climate-scientist attribution range from 65.0\% to 74.8\%, with pairwise agreement rates ranging from 71.9\% to 89.4\%. Figure~\ref{fig:ajps-fit-the-fit} presents the estimated optimal ITR obtained by exact optimization over depth-2 axis-aligned decision trees. The rule recommends climate-scientist attribution for governments in which women constitute \(\leq10\%\) of invited officials when only one official is invited, and for governments with a female share \(>10\%\) when the local population is \(\leq212{,}988\); it recommends peer-politician attribution otherwise. Substantively, these splits suggest that the estimated relative effectiveness of sender attribution varies with the size and gender composition of the invited audience and, among governments with greater female representation, with local population.

\begin{table}[ht!]
\centering
\caption{Pairwise agreement among the estimated optimal ITRs in the ONZ
field experiment. The row labeled Observed $A$ reports agreement with the
observed treatment assignment. ReLU-H and ReLU-L denote ReLU with hinge and logistic losses, respectively; RBF-H and RBF-L denote Gaussian RKHS with hinge and logistic losses, respectively; and Tree denotes exact optimization over depth-2 axis-aligned decision trees.}
\label{tab:ajps-full-policy}

\begingroup
\footnotesize
\linespread{1}\selectfont
\compacttablesetup

\setlength{\tabcolsep}{10.2pt}
\renewcommand{\arraystretch}{0.86}
\setlength{\aboverulesep}{0.15ex}
\setlength{\belowrulesep}{0.9ex}

\begin{tabular}{@{}lccccc@{}}
\toprule
& ReLU-H & ReLU-L & RBF-H & RBF-L & Tree \\
\midrule
ReLU-H & 1 & 0.894 & 0.869 & 0.810 & 0.720 \\[1ex]
ReLU-L &   & 1     & 0.850 & 0.770 & 0.727 \\[1ex]
RBF-H  &   &       & 1     & 0.794 & 0.719 \\[1ex]
RBF-L  &   &       &       & 1     & 0.787 \\[1ex]
Tree   &   &       &       &       & 1     \\[1ex]
Observed $A$ & 0.490 & 0.491 & 0.464 & 0.493 & 0.493 \\
\midrule
Climate scientist, $n$ (\%) &
1585 (69.4) & 1489 (65.2) & 1484 (65.0) & 1666 (72.9) & 1708 (74.8) \\
\bottomrule
\end{tabular}

\endgroup
\end{table}
\FloatBarrier

\begin{figure}[ht!]
\centering
\includegraphics[width=0.96\textwidth]
{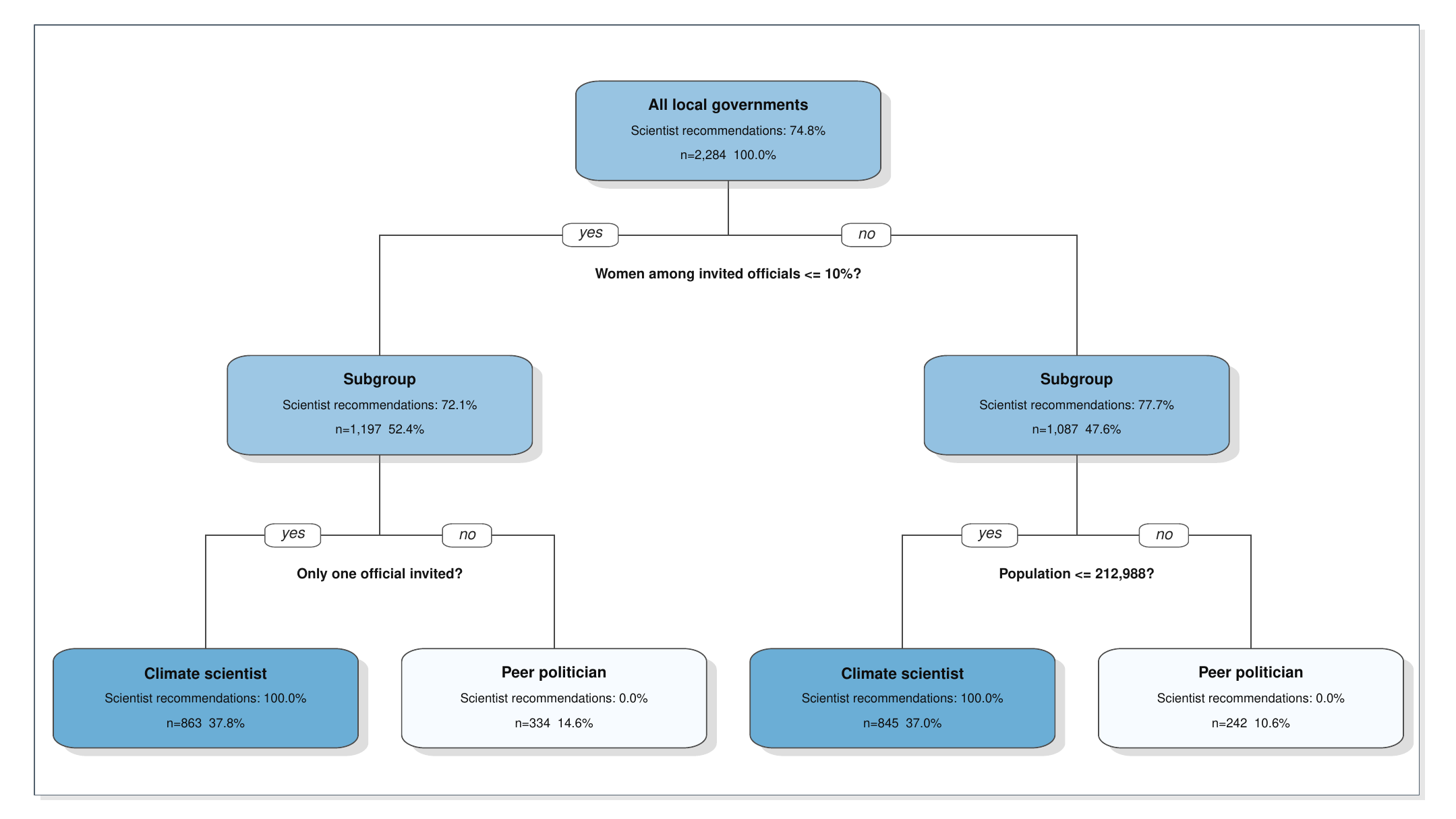}
\caption{Fitted depth-2 axis-aligned decision tree representing the estimated optimal ITR for the ONZ field experiment. Terminal nodes report the recommended treatment; boxes report the percentage of governments for which the estimated
optimal ITR recommends climate-scientist attribution, the node size, and
the corresponding sample proportion.}
\label{fig:ajps-fit-the-fit}
\end{figure}

Finally, post-learning value estimation using a doubly robust method yields estimated government-level click rates of \(6.56\%\), \(6.59\%\), \(7.99\%\), \(5.59\%\), and \(6.17\%\) under the estimated optimal ITRs obtained using ReLU-H, ReLU-L, RBF-H, RBF-L, and the decision tree method, respectively. The corresponding estimated values are \(4.84\%\) under assigning climate-scientist attribution to every government and \(5.67\%\) under a geographic policy assigning climate-scientist attribution in Europe and peer-politician attribution in the United States. Relative to assigning climate-scientist attribution to every government, the five estimated optimal ITRs yield apparent differences of \(1.72\), \(1.75\), \(3.15\), \(0.74\), and \(1.32\) percentage points, respectively. Relative to the geographic policy, the corresponding differences are \(0.89\), \(0.92\), \(2.32\), \(-0.08\), and \(0.49\) percentage points.

\clearpage

\subsection{Additional tables and figures}
\label{ss:additional-tables-figures}

\begin{figure}[ht!]
\centering
\includegraphics[width=0.93\textwidth]
{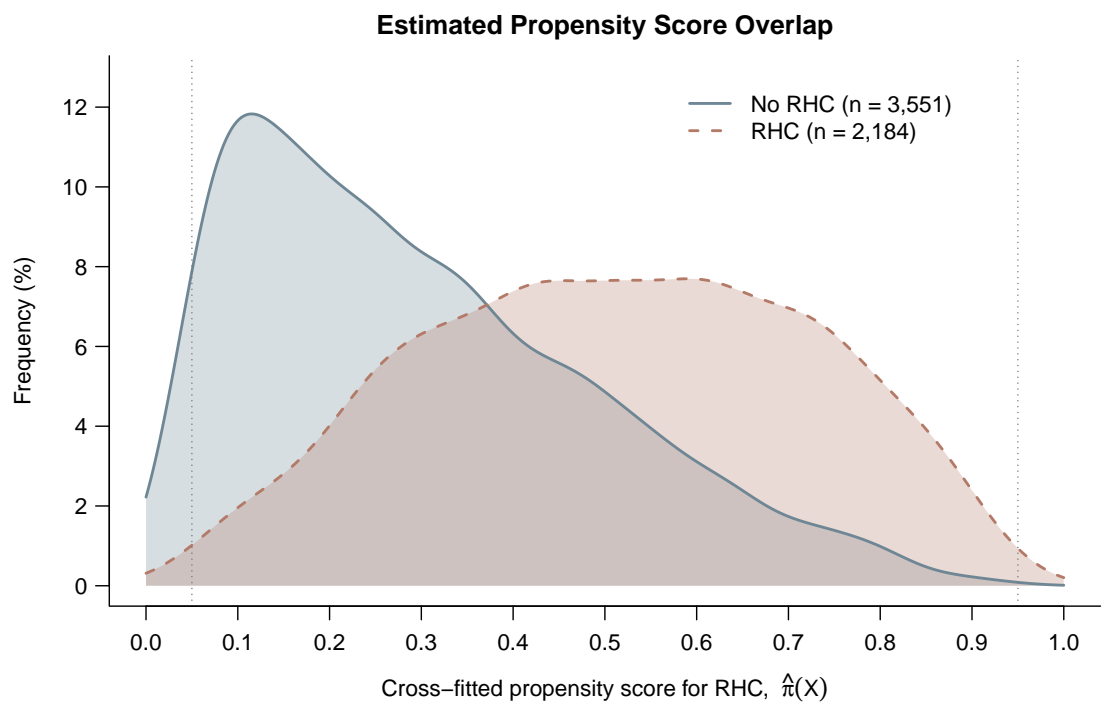}
\caption{Estimated propensity score overlap in the RHC study. The curves show ten-fold cross-fitted Super Learner estimates of the probability of receiving RHC, stratified by observed treatment. The dotted vertical lines mark 0.05 and 0.95.}
\label{fig:rhc-propensity-overlap}
\end{figure}

\clearpage


\end{document}